\documentclass[12pt, a4paper]{article}

\usepackage{fullpage,graphicx,xcolor}
\usepackage{amsmath,amssymb,amsfonts,amsthm,mathtools}
\usepackage{natbib,float,enumerate,tabularx,booktabs,url,comment,enumitem}
\usepackage{geometry}
\usepackage[onehalfspacing]{setspace}
\usepackage[titletoc,toc]{appendix}
\usepackage{minitoc}
\renewcommand \thepart{}
\renewcommand \partname{}
\usepackage[flushleft]{threeparttable}
\usepackage{ragged2e}
\usepackage{subcaption}
\usepackage{lscape}
\usepackage[thinc]{esdiff}
\usepackage{bm,bbm}

\DeclareMathOperator{\diag}{diag}

\DeclareMathOperator*{\argmin}{arg\,min}

\newcommand{\Var}{\mathrm{Var}}
\newcommand{\Cov}{\mathrm{Cov}}

\newtheoremstyle{newthmstyle}%
  {\topsep}% above space
{\topsep}% below space
{\itshape}% body font
{0pt}% indent amount
{\bfseries}% head font
{. }% punctuation between head and body
{0pt}% post head space
{}% head spec
\theoremstyle{newthmstyle}

\newtheorem{lemma}{{Lemma}}
\newtheorem{definition}{Definition}

\newtheorem{prop}{{Proposition}}

\theoremstyle{definition}

\newtheorem*{nonumexample}{{Example}}
\newtheorem{appprop}{{Proposition}}

\usepackage{tocvsec2}
\usepackage{xpatch}
\makeatletter
\xpatchcmd{\sv@part}{\huge \bfseries \partname \nobreakspace \thepart \par \vskip 20\p@ \fi \Huge \bfseries #2}{\fi \Huge \bfseries \thepart. #2}{}{}
\makeatother

\renewcommand \thepart{}
\renewcommand \partname{}

\begin{document}

\doparttoc % Tell to minitoc to generate a toc for the parts
\faketableofcontents % Run a fake tableofcontents command for the partocs

%\part{} % Start the document part

\title{\vspace{-0.0cm} Inventories, Demand Shocks Propagation, and Amplification in Supply Chains\footnote{This paper supersedes a previously circulated version under the title ``Global Value Chains and the Business Cycle''. I am especially grateful to Pol Antr{\`a}s, Vasco Carvalho,  Dalila Figueiredo, Nir Jaimovich, Philipp Kircher, Ramon Marimon, and  Alireza Tahbaz-Salehi for many fruitful discussions on this paper. I also
thank Daron Acemoglu, George Alessandria (discussant),  David Baqaee, Giacomo Calzolari, Maria Jose Carreras-Valle, Russell Cooper, Rafael Dix-Carneiro, Juan Dolado, David Dorn, Matt Elliott, Matteo Escud\'e, David H\'emous, Damian Kozbur, Andrei Levchenko, Michele Mancini, Isabelle M\'ejean, Konuray Mutluer, Ralph Ossa, Nitya Pandalai-Nayar, Mathieu Parenti (discussant), Florian Scheuer, Armin Schmutzler and Mathieu Taschereau-Dumouchel for their feedback. This paper also benefited from comments in the seminars at EUI, Collegio Carlo Alberto, Paris School of Economics, CERGE-EI Prague, CSEF Naples, University of Zurich, Queen Mary University London, HEC Montreal, ECB Research Department, EIEF, World Trade Institute and presentations at ASSET,  RIEF AMSE, WIEM, the EEA meeting, ETSG, Econometric Society European Winter Meeting, GVC Conference of Bank of Italy, Large Firms Conference at Bank of France, T2M King's College, the GEN Workshop and ASSA. 
 Lorenzo Arc{\`a}, Elie Gerschel, and Lorenzo Pesaresi provided excellent research assistance. All the remaining errors are mine.}}%\\\large{Notes}}
\author{\color{blue}\large Alessandro Ferrari\footnote{{alessandro.ferrari@upf.edu}}\\
\small{Universitat Pompeu Fabra, BSE \& CEPR}}
\date{\small{First Version: February 2019\\ This version: March 2025} \\\vspace{12pt}}   

\maketitle
\vspace{-30pt}
\begin{abstract} 
I study the role of industries' position in supply chains in shaping the transmission of final demand shocks. 
First, I use a novel shift-share design leveraging destination-specific final demand shocks and a new measure of destination exposure accounting for direct and indirect linkages. I  find that demand shocks amplify significantly as they propagate upstream, with upstream industries experiencing output elasticities up to three times larger than final good producers, consistent with the \textit{bullwhip effect}. 
To rationalize these empirical results, I develop a tractable production network model with inventories and study how the properties of the network and the cyclicality of inventories interact to determine whether final demand shocks amplify or dissipate upstream. I test the mechanism by directly estimating the model-implied relationship between output growth and demand shocks, mediated by network position and inventories. I find that the presence of inventories increases output elasticities by 18\% on average, highlighting the macroeconomic significance of this channel. Finally, I use the model to quantitatively study the effects of long-run trends of lengthening supply chains and rising inventories on the volatility of the economy.

%Using a novel shift-share identification strategy, we leverage destination-specific foreign demand shocks and construct a new measure of destination exposure that incorporates both direct and indirect linkages. 

\end{abstract}

\begin{raggedright} \small{Keywords: production networks, supply chains, inventories, shock amplification, bullwhip effect.}
\end{raggedright}
\newpage

\onehalfspacing

\section{Introduction}
%In modern economies, the production of goods and services is organized along complex supply chains, with goods crossing borders multiple times before reaching final consumers. 
The globalization of production has fundamentally reshaped economic activity, with goods and services now traversing complex supply chains that span multiple countries before reaching final consumers.
The rising interconnectedness raises critical questions about how economic shocks travel through these networks, thereby affecting economic volatility. In more complex and globalized value chains, shocks can be absorbed at multiple nodes and diversified away, thereby reducing overall volatility. On the other hand, shocks can travel longer distances and snowball across firms and countries. These questions recently gained salience as governments try to understand the causes and consequences of the recent supply chain crisis and whether a policy response is necessary. The policy discussion revolved around reshoring, multi-sourcing, and inventory management to shorten supply chains and reduce the propagation of shocks.\footnote{In June 2021, the Biden-Harris Administration instituted the \textit{Supply Chain Disruption Task Force} ``to provide a whole-of-government response to address near-term supply chain challenges to the economic recovery'', see \cite{whitehouse}. For the European context, see \cite{raza2021post}, a study commissioned by the International Trade Committee of the European Parliament considering reshoring options and the European Parliament resolution calling for ``smart reshoring [to] relocate industrial production in sectors of strategic importance for the Union'' and the creation of a program that ``helps make our supply chains more resilient and less dependent by reshoring, diversifying and strengthening them'', see \cite{euparl}.}

In this paper, I address these questions by studying how demand shocks travel in production networks. I investigate two fundamental forces. First, as supply chains become more complex, firms can become exposed to a broader set of destinations, allowing them to partially diversify demand risk away. As a consequence, for a given set of shocks, more complex chains can be less volatile. Second, the presence of inventories, whose adjustment can absorb or amplify shocks as they travel along the supply chain. The role of inventories as shock amplifiers or absorbers is of particular interest, as it is often discussed as a strategy to avoid prolonged disruptions in supply chains. %If firms respond to fluctuations in demand and supply by building up larger stocks, they might contribute to higher volatility and shock propagation. 

%I start by providing five empirical observations that motivate this paper: in the last decades, i) production chains have significantly increased in length, measured by the number of production steps goods undergo before reaching consumers; and ii) the spatial concentration of demand has significantly declined. These two empirical facts suggest that the rise of complex supply chains may better insulate from final demand shocks as they are absorbed in multiple steps and diversified away. However, I also show that iii) inventories are procylically adjusted, and, as a consequence, iv) output is more volatile than sales. In the context of a production network, these two observations imply upstream amplification, which is strengthened by longer and more complex chains. Finally, I confirm a recent finding by \cite{carrerasvalle}: v) that US manufacturing firms' inventory-to-sales ratios have been increasing since 2005. This observation suggests that, as a consequence, upstream amplification may have become more salient.

The empirical motivation for this paper lies with the three trends plotted in Figure \ref{fig:motiv}. First, the length of supply chains has increased over time. Figure \ref{fig:motiv_U} shows the average number of steps in the production process of goods and services over time. This measure has increased by 27\% between 2000 and 2014. Second, as shown in Figure \ref{fig:motiv_HHI}, the spatial concentration of demand has declined. The Herfindahl-Hirschman Index of destination shares has decreased by 14\% in the same period. Therefore, these two novel observations highlight that, on the one hand, supply chains have gotten longer, potentially giving rise to further propagation and amplification of shocks. On the other hand, they have become more diversified, allowing firms to partially hedge demand risk across space. Finally, in Figure \ref{fig:motiv_inv}, I report the finding of \cite{carrerasvalle}, who shows that inventory-to-sales ratios have increased in the last decades, reverting a secular decline. As the importance of inventories in the production process increases, their interaction with longer and more complex supply chains can make the economy experience larger or smaller fluctuations. 
\begin{figure}
     \centering
     \begin{subfigure}[b]{0.325\textwidth}
         \centering
         \includegraphics[width=\textwidth]{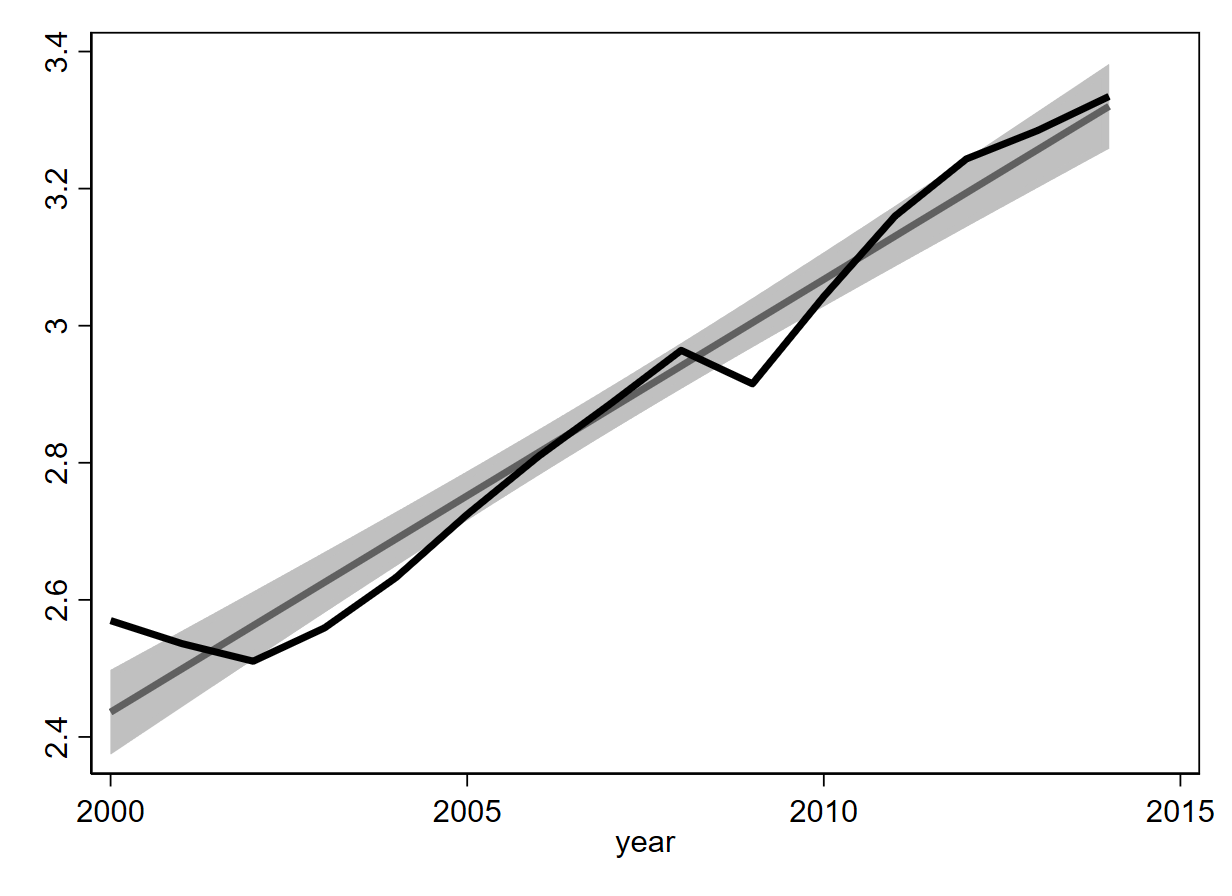}
         \caption{Supply Chains Length}
         \label{fig:motiv_U}
     \end{subfigure}
     \hfill
     \begin{subfigure}[b]{0.325\textwidth}
         \centering
         \includegraphics[width=\textwidth]{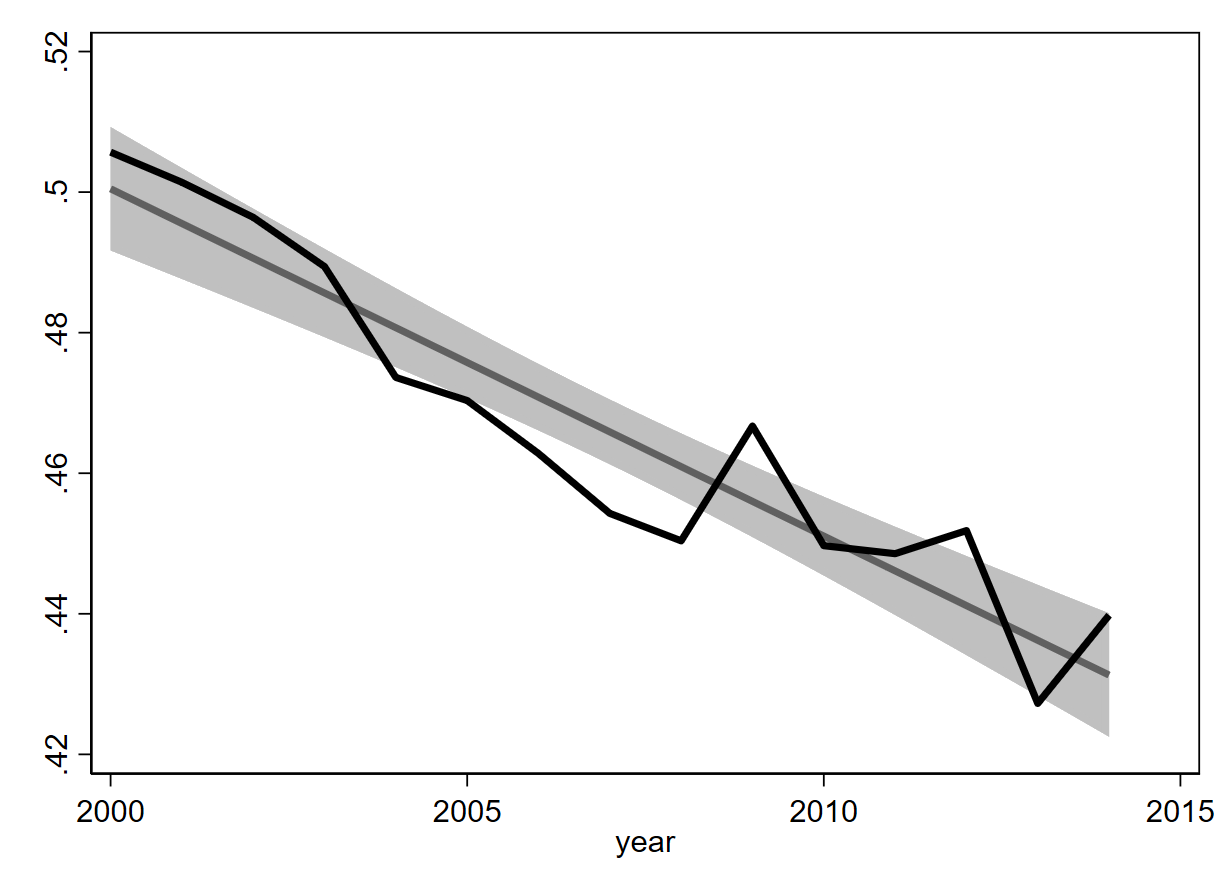}
         \caption{Demand Concentration}
         \label{fig:motiv_HHI}
     \end{subfigure}
     \hfill
     \begin{subfigure}[b]{0.325\textwidth}
         \centering
         \includegraphics[width=\textwidth]{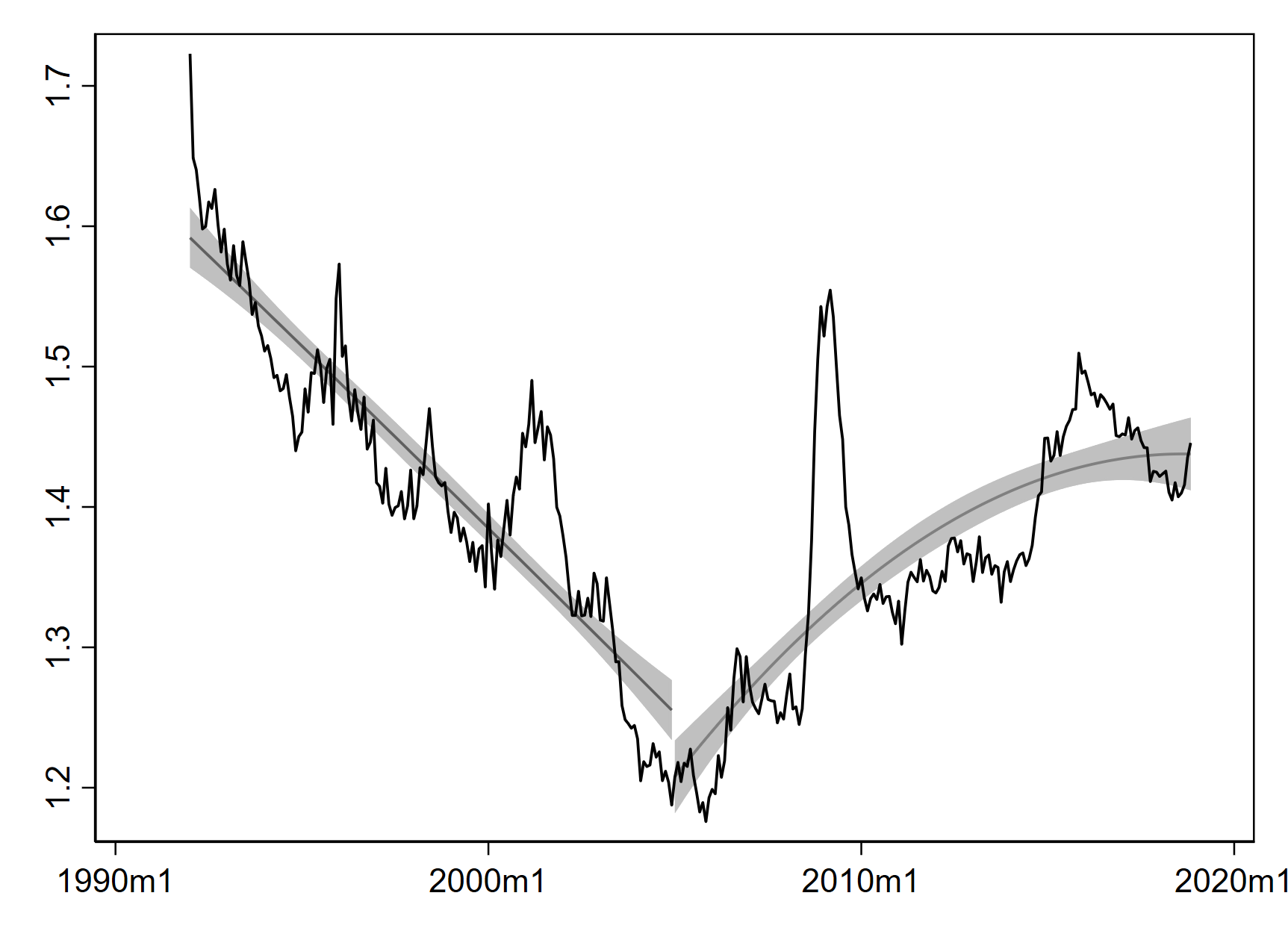}
         \caption{Inventory-to-Sales Ratio}
         \label{fig:motiv_inv}
     \end{subfigure}
        \caption{Motivating Evidence}
        \label{fig:motiv}
        \vspace{5pt}
\noindent\justifying
 \footnotesize Note: Panel a) shows the average Upstreamness over time across all industries in the WIOD data. A formal definition is provided in Section \ref{methoddatasec}. Panel b) shows the Herfindahl-Hirschman Index of destination sales shares in WIOD. This computation, formalized in Section \ref{methoddatasec}, accounts for both direct and indirect exposure to each destination country. Panel c) reports the same statistic from the Census data from Jan-1992 to Dec-2018. Both graphs include non-linear trends before and after 2005. I estimate separate trends as \cite{carrerasvalle} suggests that 2005 is when the trend reversal occurs. I provide extensive checks and alternative specifications of these observations in Appendix \ref{motivatingevidence}.
\end{figure}

Motivated by these empirical observations, in Section \ref{results}, I study the propagation of shocks along supply chains in the presence of inventories. %I show that industries at different positions in supply chains display heterogeneous output and inventory responses to final demand shocks.
I use a novel shift-share design based on destination-specific foreign demand shifters and a new measure of destination exposure. Importantly, I use the full network structure to account for direct and indirect linkages and exposure to each foreign country. The shift-share structure allows me to causally estimate the elasticity of output to final demand. With this design, I estimate a model in which I allow this elasticity to vary by upstreamness, which measures the network distance between an industry and final consumers. The first empirical result is that industries at different points of the value chains have significantly different responses. Quantitatively, I find that industries close to consumers have an elasticity of 1, while industries very far from consumption (5 or more steps of production away) have an elasticity larger than 2. To exemplify this finding, this implies that a 1\% increase in the consumers' demand for cars induces a 1.2\% increase in auto production, 1.4\% in tire production, 1.6\% in rubber production, and a 2.1\% increase in latex and oil extraction.
This result represents the first causal evidence of the \textit{bullwhip effect}, often discussed in the operations literature \citep{forrester1961industrial}. This theory posits that, when inventories are procylically adjusted, reactions to shocks increase upstream. I test this hypothesis and find that not only are inventories procylically adjusted, but their elasticity increases upstream. Quantitatively, I find an up to 8-fold increase going from one to five production steps away from consumption. These empirical results imply that inventories drive the upstream amplification of shocks and that this channel is economically important. 

Contrary to my finding on the \textit{bullwhip effect},  current models of shock propagation in production networks typically predict decreasing responsiveness as shocks travel in the network. To address this, I study a tractable production network model where firms hold inventories. I build this framework with three goals. First, it allows me to shed light on the key features of the network and of the inventory problem that lead to upstream amplification vs. dissipation of final demand shocks. Second, the model provides an estimating equation, which allows me to test the key mechanism directly. Third, through counterfactuals, I can separate the effects of changes in the structure of supply chains and rising inventories on volatility. 

I build my theoretical analysis in three steps. First, I study a vertical supply chain with a general inventory policy to isolate the role of inventories. Next, I specify an inventory problem and consider a general network with only one source of fluctuations to highlight the role of the network structure. Finally, I allow for many shocks to study the combination of inventories, network structure and diversification forces. I describe each step in more detail.

I start the analysis by establishing a general result in vertically integrated supply chains: procyclical inventory adjustment is a sufficient condition for upstream amplification of final demand shocks. Intuitively, firms respond to changes in demand more than 1-for-1 as they increase production to meet higher demand and update their inventory stocks in the same direction. In a supply chain, this implies that firms \textit{pass along} magnified fluctuations. Second, when more upstream firms have more procyclical inventory policies, as in the data, amplification increases in upstreamness. Taken together, these imply that procyclical inventories can rationalize the increasing and convex relation between elasticities and upstreamness I find in the data. While this result is general in terms of the inventory policy implemented by firms, it is derived in the context of a specific network structure, and it is hard to establish its quantitative implications. In particular, such a network cannot account for the diversification forces that might allow shocks to dissipate upstream. To make further headway in the analysis for general networks, I extend the model by studying a general production network structure and specifying the inventory problem.

I consider an economy populated by firms producing in a general production network. In each sector, there are two types of firms that share the same technology. Type $C$ firms operate competitively. Type $I$ firms are subject to an inventory management problem and have market power in their output markets. The two types of firms compete within each industry, producing perfect substitute goods. $I$ firms optimally target future sales when deciding how much inventory stock to hold.  

To isolate the role of inventories and their interplay with the network structure, I start with a special case in which there is only one foreign destination country whose demand fluctuates. In this context, the first theoretical result is that output growth responds more than one-for-one to final demand shocks and that this amplification scales with an inventory-weighted measure of upstreamness $\mathcal{U}$. This result stems from the combination of two properties of the model: i) when there is a single consumption destination, no diversification forces are in action, and ii) when firms have an inventory target rule, they amplify shocks upstream. Importantly, the presence of inventories implies that the network structure matters even to a first order.  The second result characterizes the behavior of output growth volatility. I show that an industry is more volatile than another if and only if it is more upstream, in the $\mathcal{U}$ sense. This measure of upstreamness has intuitive comparative statics: it increases if firms hold more inventories, if demand shocks are more persistent, and if supply chains are longer. I show that, in this economy, fragmenting production vertically always increases output volatility since it makes all connected industries weakly more upstream, generating more amplification at every step of the supply chain. 

Motivated by the rising importance of diversification discussed above, I conclude the theoretical analysis by considering a setting with many consumption destinations to allow for hedging of demand risk across space. In this context, I show two important features: a) the model-consistent way of aggregating final demand shocks across destinations is given exactly by the shift-share structure used in the empirical analysis, and b) a horse race between industry position and the spatial dispersion of demand governs upstream amplification. I then establish that, all else equal, i) more upstream industries are more volatile and ii) comparing two industries with similar supply chain positions, an industry has less volatile output if and only if its destination-exposure distribution second-order stochastically dominates the other's. In contrast to existing frameworks, the model can account for both the inventory and diversification effects of more complex supply chains. Depending on the structure of the network and the degree of inventory amplification, it can generate both increases and decreases in volatility as supply chains become more complex. 

The model provides a closed-form estimating equation linking output growth and the shift-share changes in final demand as a function of inventories and upstreamness $\mathcal{U}$. This relation has two useful properties. First, it can be taken to the data directly. Second, it provides a sharp prediction on the coefficients we should find if inventories did not matter for amplification. Using the model-implied estimating equation, I find that the data strongly rejects the null hypothesis of a model without inventories. I estimate that the presence of inventories increases the output elasticity to final demand by approximately 18\%, highlighting the economic significance of the \textit{bullwhip effect}.

Finally, I study a quantitative version of the model using the actual Input-Output data as the network structure. First, the model can replicate salient patterns of the data. In particular, it generates i) a negative relationship between the volatility of demand and upstreamness, highlighting diversification forces; ii) a positive relation between output growth and upstreamness, driven by the upstream propagation of shocks through inventories; and iii) the increasing output growth elasticity in upstreamness I estimate in Section \ref{results}. Second, with the calibrated model, I 
can disentangle the role of increasing supply chain length, rising diversification, and increasing inventories through counterfactuals. Comparing the network structure in 2000 to the one in 2014, I find that increasing supply chain length induces a 1.8\% increase in output volatility, while the rise of inventories generates a 5.4\% increase. When these two are combined, output volatility increases by 7.6\%. In contrast, the increasing spatial diversification of demand implies a reduction in demand volatility of 3.9\%. Combining all these forces, I find that, despite the smaller fluctuations in demand, output volatility increases by 2.4\%. While the rise of more complex supply chains carried smaller volatility of perceived demand thanks to diversification forces, the rise in inventories more than offset it and generated an overall more volatile economy.

More broadly, these results highlight a trade-off element to the debate on how to make supply chains more resilient. These results suggest that the fragmentation of production has brought about significant benefits in terms of diversification of demand risk. At the same time, more fragmented supply chains generate more scope for upstream amplification of shocks via inventories. If firms increase their inventory buffer to prevent prolonged disruption, this might come at the cost of permanently higher volatility in the economy. 

\paragraph{Contribution and Related Literature.}
This paper makes three contributions relative to the existing literature, which I discuss in turn. 

First, I provide the first causal evidence of the \textit{bullwhip effect} in production networks, showing that upstream industries experience up to three times larger output elasticities than final good producers. The theoretical underpinning on this theory dates back to \cite{forrester1961industrial} and more recently \cite{kahn1987inventories,blinder1991resurgence, metters, chenetal,ramey1999inventories,bils2000inventory,leeetal,khanthomas,alessandria_et_al_2010,alessandria2011us,alessandria2013trade,alessandria2023aggregate,khan2021does,ferrardelivery} suggest that when inventories are procyclically adjusted, they can amplify shocks upstream. These papers consider inventories in models without production networks, whereby there is no scope for a well-defined notion of network position and, therefore, of upstream amplification.\footnote{Recent work by \cite{alessandria2023aggregate} considers economies with inventories in roundabout production, which allows for input-output feedback but no explicit notion of network structure.} From an empirical standpoint, the effect of inventories as an amplification device has been studied at several levels of aggregation by \cite{alessandria_et_al_2010}, \cite{altomonte} and  \cite{Zavacka_thebullwhip}. These papers all consider exogenous variation given by the 2008 crisis to study the responsiveness of different sectors or firms to the shock, depending on whether they produce intermediate or final goods. Using an indicator for exposure to the shock creates an identification problem: it is not possible to separately identify differential exposure from different responses. Relative to these works, my approach based on the shift-share design allows me to isolate the heterogeneity in elasticities to the same change in final demand depending on an industry's position in the supply chain. My empirical findings are also related to \cite{aak,barrot2016input,boehm2019input, huneeus,carvalho2016supply,dhyne2021trade,korovkin2021production}, who study how shocks propagate in a production network.
Relative to these contributions, I build industry-level exogenous variation in the spirit of \cite{shea1993input} by exploiting destination and time-specific shocks to consumption. This approach enables me to estimate the heterogeneous response to final demand shocks across industries at different points of the supply chain, holding fixed the size of the shock itself. Further, I study directly the role of inventories as an additional channel contributing to the patterns of shock propagation.

The second contribution is the development of a theoretical model of production networks with inventories. My work builds on existing models of shock propagation in networks as in \cite{carvalho2010aggregate}, \cite{acemogluetal}, \cite{baqaee2019macroeconomic}, \cite{carvalho2016supply}. Contrary to these contributions, I include the pivotal role of inventories as a mechanism of upstream amplification of demand shocks to reconcile my empirical findings. This allows me to characterize a new sufficient statistic for upstream amplification, which simultaneously summarizes features of the network and of the inventory problem. I show that this statistic can be directly measured with data on the input-output structure and inventories. This framework also allows me to study the effect of fragmentation of supply chains as well as the role of spatial diversification on volatility. 

The last contribution of this paper is to quantify the macroeconomic implications of rising supply chain complexity and inventory trends, showing that these forces have contributed to an overall increase in economic volatility despite stronger diversification effects. These trends and their effects on supply chains' resilience, robustness, and shock propagation have been recently discussed by \cite{elliott2022supply,elliott2022networks,grossman2023resilience,grossman2023supply,elliott2023supply,acemoglu2024macroeconomics,ferrarispecialization}.

\paragraph{Roadmap.} The rest of the paper is structured as follows: Section \ref{methoddatasec} provides the details on the data and the empirical strategy. Section \ref{results} presents the reduced form results on upstream shock propagation. Section \ref{model} describes the model and provides the key comparative statics results and counterfactual exercises. Section \ref{conclusions} concludes.

\section{Data and Methodology}
\label{methoddatasec}
\subsection{Data}\label{data}

\paragraph{Input-Output Data.}
The primary data source is the World Input-Output Database (WIOD) 2016 release, see \cite{timmer2015illustrated}. It contains the Input-Output structure of sector-to-sector flows for $J=44$ countries from 2000 to 2014 yearly. The data is available at the 2-digit ISIC rev-4 level. The number of sectors in WIOD is $S=56$, which amounts to 6,071,296 industry-to-industry flows and 108,416 industry-to-country flows for every year in the sample. The data coverage of countries and industries is shown in Appendix \ref{app:data}.% The structure of WIOD is presented in Figure \ref{wiod}.

The World Input-Output Table is an $(S \times J)$ by $(S \times J)$. Each element $Z_{ij}^{rs}$ denotes the value of output of industry $r$ in country $i$ sold to industry $s$ in country $j$ for intermediate input use. Additionally, the table WIOT includes an $(S\times J)$ by $J$ matrix of final use. Each element $F_{ij}^r$ is the value of output of industry $r$ in country $i$ sold to and consumed by households, government, and non-profit organizations in country $j$. I denote $F_i^r=\sum_jF_{ij}^r$, namely the value of output of sector $r$ in country $i$ consumed in any country in the world. The total value of sold output is $Y_i^r=F^r_i+\sum_s\sum_j Z_{ij}^{rs}$.

In all the analyses, I drop country-industry-year triplets whose year-on-year output growth rate is lower than -90\% or larger than 57\% ($99^{th}$ percentile of the industry growth distribution). %The results are consistent with different cuts of the data and without dropping any entry. 
%Unless stated otherwise, the analysis in this section includes all industry-country-year combinations in the WIOD data.

%For a subset of results, I use the 2002 I-O Tables from the BEA, following \cite{antras_et_al}. These are a one-year snapshot of the US production network and cover 426 industries. 

\paragraph{Inventory Data.}

I complement the Input-Output data with information about sectoral inventories from the NBER-CES Manufacturing Industry Database. This dataset contains sales and end-of-the-period inventories for 473 6-digit 1997 NAICS U.S. manufacturing industries from 1958 to 2011. I collapse the inventory data to the level of aggregation of WIOD.

The second source of inventory data is the U.S. Census Manufacturing \& Trade Inventories \& Sales. This dataset covers NAICS 3-digit industries monthly since January 1992. The data includes information for finished products, materials, and work-in-progress inventories, which I sum into a single inventory measure. I use the seasonally adjusted version of the data.

Finally, note that WIOD includes information on changes in inventories by producing industry-year. I use this data whenever I am interested in changes since it has a significantly larger coverage than the U.S. data. I resort to the Census and NBER CES data whenever I need information on the stock of inventories, which is not contained in WIOD.

%%%%%%%%%%%%%%%%%%%%%%%%%%%%%%%%%%%%%%%%%%%%%%%%%%%%%%%%%%%%%%%%%%%%%%%%%%%%%%%%%%%%%%%%%%%%%%%%%%%%%%%%%%%%%%%%%%%%%%%%%%%%%%%%%%%%%%%%%%

\subsection{Measurement and Estimation}\label{methodology}

%This section describes the empirical methodology used. I start by reviewing the existing measure of upstreamness as distance from final consumption proposed by \cite{antras_et_al}. Next, I discuss the identification strategy based on the shift-share design. To build the instrument, first, I show how to compute the sales share accounting for indirect linkages. This allows me to back out the exposure of industry sales to each destination country's demand fluctuations, even when goods reach their final destination by passing through third countries. Then, I discuss the fixed-effect model used to extract and aggregate country and time-specific demand shocks from the final consumption data.

\subsubsection{Measuring the Position in Production Chains}

The measure of the upstreamness of each sector counts how many stages of production exist between the industry and final consumers, as proposed by \cite{antras_et_al}. 
The measure is bounded below by 1, which indicates that the entire sectoral output is used directly for final consumption. 
The index is constructed by assigning value 1 to the share of sales directly sold to final consumers, value 2 to the share sold to consumers after it is used as an intermediate good by another industry, and so on. Formally:
\begin{align}
U_i^r=1\times \frac{F_i^r}{Y_i^r}+2\times \frac{\sum_{s=1}^S\sum_{j=1}^Ja_{ij}^{rs}F_{j}^s}{Y_i^r}+3\times \frac{\sum_{s=1}^S\sum_{j=1}^J\sum_{t=1}^T\sum_{k=1}^Ka_{ij}^{rs}a_{jk}^{st}F_{k}^t}{Y_i^r}+...
\label{upstreamness}
\end{align}
where $F_i^r$ is the value of output of sector $r$ in country $i$ consumed anywhere in the world and $Y_i^r$ is the total value of output of sector $r$ in country $i$. $a_{ij}^{rs}$ is dollar amount of output of sector $r$ from country $i$ needed to produce one dollar of output of sector $s$ in country $j$, defined as $a_{ij}^{rs}=Z_{ij}^{rs}/Y_j^s$. This formulation of the measure is effectively a weighted average of distance, where the weights are the distance-specific share of sales and final consumption.%\footnote{This discussion implicitly assumes that a sector's input mix is independent of the output use or destination. \cite{degortari}, using customs data from Mexico, shows that this can lead to mismeasurement. These concerns cannot be fully addressed in this paper due to the limitation imposed by the WIOD aggregation level.}

Provided that $\sum_i \sum_r a_{ij}^{rs}<1$, which is a natural assumption given the definition of $a_{ij}^{rs}$ as input requirement, this measure can be computed by rewriting it in matrix form: $U=\hat Y^{-1}[I-\mathcal{A}]^{-2}F,
$ where $U$ is a $(J\times S)$-by-1 vector whose entries are the upstreamness measures of every industry in every country.\footnote{For the assumption $\sum_i \sum_r a_{ij}^{rs}<1$ to be violated, some industry would need to have negative value added since $\sum_i \sum_r a_{ij}^{rs}>1\Leftrightarrow \sum_i \sum_r Z_{ij}^{rs}/Y_j^s>1$, meaning that the sum of all inputs used by industry $s$ in country $j$ is larger than the value of its total output. To compute the measure of upstreamness, I apply the inventory correction suggested by \cite{antras_et_al}, the discussion of the method is left to Appendix \ref{app:data}.} $\hat Y$ denotes the $(J\times S)$-by-$(J\times S)$ diagonal matrix whose diagonal entries are the output values of all industries. The term $[I-\mathcal{A}]^{-2}$ is the square of the Leontief inverse, in which $\mathcal{A}$ is the $(J\times S)$-by-$(J\times S)$ matrix whose entries are all $a_{ij}^{rs}$ and finally the vector $F$ is an $(J\times S)$-by-1 whose entries are the values of the part of industry output that is directly consumed. 
Eq. (\ref{upstreamness}) shows the value of upstreamness of a specific industry $r$ in country $i$ is 1 if and only if all its output is sold to final consumers directly. Formally, this occurs if and only if $Z_{ij}^{rs}=0, \forall s,j$, which immediately implies that $a_{ij}^{rs}=0, \forall s,j$.

Table \ref{udescriptive} lists the most and least upstream industries in the WIOD sample. Predictably, services are very close to consumption, while raw materials tend to be distant.

\subsubsection{Identification Strategy}
%The goal is to evaluate the responsiveness of output to changes in demand for industries at different positions in the supply chain. To estimate this effect, the ideal setting would be one where I observe exogenous changes in final demand for each producing industry in the sample. As this is not possible, I approximate this setting using a shift-share instrument approach. In this paper, the shift-share design is a weighted average of destination-specific aggregate changes in final demand, where the weights measure how exposed an industry is to that destination. 

The object of interest is the elasticity of output to changes in final demand and its heterogeneity along the upstreamness dimension. Denote this $\beta_U$ from the regression
\begin{align}
    \Delta \log Y^r_{it}=\beta_U \Delta\log {D}_{it}^r+\upsilon_{it}^r,\label{eq:elast}
\end{align}
where $\Delta\log {D}_{it}^r$ is the change in demand for goods produced by industry $r$ in country $i$ at time $t$. Focusing on changes in final demand allows me to use a consistent definition of \textit{distance from the source of the shock}: Upstreamness. If I were to find that $\beta_U$ declines in $U$, we would conclude that firms further away from the source of the shock respond less to it. Contrarily, if $\beta_U$ increases in $U$, we would conclude that the changes in final demand amplify as the travel upstream in the supply chain. 

Naturally, I do not observe demand but rather sales. Sales are presumably chosen by firms jointly with output. As a consequence, estimating eq. (\ref{eq:elast}) based on realized sales would be plagued by simultaneity bias. I circumvent the problem by using a shift-share design where plausibly exogenous changes in demand are given by 
\begin{align}
    \sum_j \xi^r_{ij}\Delta\log F_{jt},\label{eq:ss}
\end{align}
where $\xi_{ij}^r$ represents the fraction of the value of output of industry $r$ in country $i$ consumed directly or indirectly in destination $j$ in the first sample period and $\Delta\log F_{jt}$ is the growth rate of final consumption expenditure in country $j$ at time $t$. The exposure share $\xi_{ij}^r$, assumed to be time-invariant, includes direct sales from the industry to consumers in $j$ and output sold to other industries, which eventually sell to consumers in $j$.\footnote{The standard measure of sales composition uses trade data to compute the relative shares in a firm's sales represented by different partner countries \citep[see][]{kramarz2020volatility}.
However, such a measure may overlook indirect dependencies through third countries. This is particularly important for countries highly interconnected through trade. Measuring sales exposure composition only via direct flows may ignore a relevant share of final demand exposure. \cite{dhyne2021trade}, for example, show that firms respond to changes in foreign demand even when they are only indirectly exposed to them.}
The input-output structure of the data allows a full account of these indirect linkages when analyzing sales composition. Formally, define the share of sales of industry $r$ in country $i$ that is consumed by country $j$ as
\begin{align}
\xi_{ij}^r=\frac{F_{ij}^r+\sum_s\sum_k a_{ik}^{rs}F_{kj}^s+\sum_s\sum_k\sum_q\sum_m a_{ik}^{rs}a_{km}^{sq}F_{mj}^q+...}{Y_i^r}.
\label{shareseq}
\end{align}
The first term in the numerator represents sales from sector $r$ in country $i$ directly consumed by $j$; the second term accounts for the fraction of sales of sector $r$ in $i$ sold to any producer in the world that is then sold to country $j$ for consumption. The same logic applies to higher-order terms.
By definition $\sum_j\xi_{ij}^r=1$. 

The identification of $\beta$ in eq.(\ref{eq:elast}) relies on the exogeneity of either the shares (industry exposure) or the shocks (destination-specific aggregate changes), see \citet{adao, goldsmith, borusyak}. In this context, it is implausible to assume that the destination shares are exogenous as firms choose the destinations they serve. Identification can be obtained by plausibly as good as randomly assigned shifters (destination-specific shocks). However, using eq.(\ref{eq:ss}), one might still worry about reverse causality and omitted variable bias. First, suppose that industry $r$ is large in country $i$, then a boom in the sector might induce increases in demand for all products $\Delta\log F_{it}$. In this case, the causality would run from a change in $Y$ generating movements in the shift-share shock. A first identifying assumption is, therefore, that industries are small relative to destination markets. Secondly, Changes in some $Y_{it}^r$ might be related to changes in some $F_{jt}$ for reasons unrelated to changes in demand. For example, consider a global positive productivity shock to industry $r$. This might induce an increase in $Y_{it}^r$ and a contemporaneous increase in $F_{jt}$ by income effects. To alleviate this concern, I adopt a different set of shifters $\eta_{jt}$ estimated from 
\begin{align}
\Delta \log F_{kjt}^s=\eta_{jt}(i,r)+\nu_{kjt}^s\quad k\neq i, s\neq r.
\label{fixedeffects2}
\end{align}
For each industry $r$ in country $i$, I estimate country $j$'s fixed effect using all other industries $s$ of all countries except those of country $i$. This is equivalent to identifying the variation of interest through the trade flows of all other countries $k$ to the specific destination $j$. The estimated shifters $\hat\eta_{jt}(i,r)$ pick up any common variation in sales from all industries $s\neq r$ from all countries $k\neq i$ selling to market $j$ at time $t$. Such changes might be, for example, a change in fiscal policy in country $j$ that moves households' disposable income.  This method is similar to the approach in \cite{david2013china}, who instrument U.S. imports from China with imports from China of other developed economies.\footnote{\cite{kramarz2020volatility} and \cite{alfaro2021direct} use a similar fixed effect decomposition approach.}

Using the shares measured from the data and the estimated exogenous changes in final demand, shifters can be aggregated to create industry $r$ in country $i$ effective demand shocks at time $t$
\begin{align}
\hat\eta_{it}^r=\sum_j\xi_{ij}^r\hat\eta_{jt}(i,r),
\label{aggregation}
\end{align}
Where the effective sales shares are evaluated at time $t=0$ to eliminate the dependence of destination shares themselves on simultaneous demand changes. This procedure implies that sales shares from $i$ do not affect $\hat \eta_{jt}(i,r)$ and, therefore, $\hat\eta_{it}^r$.\footnote{In an alternative aggregation strategy in which I use time-varying sales shares, I follow \cite{borusyak} and test pre-trends, namely that the sales shares are conditionally uncorrelated to the destination shifters. The results are reported in Table \ref{orthogonality_test} in Appendix \ref{app:data}. I find no evidence of pre-trends.}
\begin{nonumexample}
Suppose that I observe US, Indian, and Chinese producers of cars, textiles, and furniture. When estimating the change in the final demand faced by the U.S. car manufacturing industry, I exclude the U.S. as a production country. In principle, this leaves me with identifying demand changes from sales of Indian and Chinese cars, textiles, and furniture producers to American, Indian, and Chinese consumers. However, if Chinese and American cars are very substitutable, then the observed sales of Chinese cars might be related to supply shocks to American car manufacturers. To avoid this type of reverse causality, I also restrict the analysis to sectors $s\neq r$, which, in this example, would be restricting to textile and furniture producers.
In summary, when studying the observed change in final demand for U.S. car manufacturers, I exploit variation from sales of Indian and Chinese textile and furniture manufacturers to US, Indian, and Chinese consumers. %This logic extends to the 56 sectors and 44 countries so that destination-time-specific changes in final demand are estimated using (44-1)*(56-1) observations every year.\footnote{Excluding all domestic flows does not change the results qualitatively or quantitatively. Formally, it would imply additionally imposing $k\neq j$ in eq. (\ref{fixedeffects2}).}  I provide robustness checks on this specification in the Online Appendix.

The type of changes picked up by the instrument are common shifts across products within a destination. For example, a fiscal expansion in China such that households' disposable income increases by 8\% represents a common shifter across all goods purchased by Chinese consumers. This would imply a $\hat\eta_{China}=0.08$. If, at the same time, Indian consumers' disposable income is also shifted down by 2\%, I would estimate a $\hat\eta_{India}=-0.02$. Then, the shift-share assigns to American car-makers a weighted average of 8\% and -2\% based on the shares of American cars purchased by Chinese and Indian consumers at baseline $\xi_{US,China}^{auto}$ and $\xi_{US,India}^{auto}$. 

As an example of shocks not picked up by the shift-share instrument, suppose that, at time $t$, Indian textile productivity increases. This would imply an increase in sales of Indian textiles to consumers everywhere in the world. Provided that Indian textiles are a small fraction of, say, Chinese consumers' consumption bundles (implying negligible income effects), this does not induce an increase in Chinese demand for American cars. Since this increase in sales is \textit{not} common across all sectors selling to Chinese consumers, it does not affect the estimated $\eta_{China,t}$.
\end{nonumexample}

As shown in \cite{borusyak}, the shift-share instrument estimator is consistent, provided that the destination-specific shocks are conditionally as good as randomly assigned and uncorrelated. The identification of demand shocks relies on the rationale that the fixed effect model in eq. (\ref{fixedeffects2}) captures the variation common to all industries selling to a specific partner country in a given year. When producing industries are small relative to the destination, the estimated demand shocks are exogenous to the producing industry, thereby providing the grounds for causal identification of their effects on the growth of sales.

\paragraph{Discussion.} Before delving into the empirical results, it is important to highlight the advantages and limitations of the data and methodology. The main downside of using WIOD is its coarseness in terms of sector aggregation, as it provides information on 56 sectors. On the other hand, it is the only dataset with a large country coverage (43 countries). This implies that measures of upstreamness will not account for within-sector-country trade between more disaggregated subsectors. To the extent that some sectors tend to have more fragmented within-sector trade than others, this might give rise to heterogeneous degrees of measurement error in computing upstreamness. The second disadvantage is the limited time series dimension, as the data at this level of disaggregation only covers 15 years. This constrains my ability to identify a richer set of parameters for the quantitative model, as I discuss in Section \ref{quantmodelsec}.

The key advantages that enable the identification strategy are twofold. First, WIOD has a large country coverage, which allows me to use changes in foreign consumers' expenditure to generate exogenous variation. For example, this empirical strategy would not be feasible if I focused on the U.S. BEA I-O tables. Second, WIOD is a \textit{closed economy} I-O table since it includes a rest of the world aggregate in the I-O matrix.\footnote{Note that single country I-O tables such as the BEA I-O Table are not closed economy tables. They typically include exports and imports as a residual, which are typically not broken down by producing sector or purchasing country. Due to this data limitation, such tables are not well suited for computing how distant industries are from consumers in other countries.} As a consequence, it is possible to measure upstreamness of each industry to each consumption point, which cannot be done in single-country IO tables. In principle, one could attempt the same approach in firm-level production network data. This would pose a similar measurement problem since these datasets only report a partial network, which would prevent the correct identification of each firm's position relative to the source of demand.\footnote{As an illustrative example, Belgian or Chilean firm-to-firm production network data only covers domestic or domestic and custom transactions. Therefore, it is not possible to estimate the distance of each firm from consumers since these may be reached after transactions unobserved in the firm-to-firm data. This is not the case in WIOD, as this is a closed economy IO table.}

Next, it is important to note that all the results are presented using within-country-sector variation only so that permanent differences between sectors, for example, the durability of their products, are controlled. Furthermore, given the shift-share structure, which embodies direct and indirect linkages, the construction of the demand shocks already accounts for diversification forces. The advantage of the empirical approach in this paper is that it allows me to isolate the differential output response, fixing the size of demand shocks.

Lastly, while the WIOD data covers agriculture, manufacturing, and services and, therefore, both tradeable and non-tradeable goods, the identification strategy naturally accounts for the fact that some sectors are not exposed to foreign demand directly (e.g., a non-traded service sector) but might be indirectly through their domestic customers. Consider, for example, the domestic construction sector. While its direct demand comes exclusively from domestic industries, it is indirectly affected by shifts in foreign demand. For example, if the domestic construction industry serves a manufacturing sector that sells to a foreign destination, which, whenever foreign demand increases, demands domestically more construction. This indirect exposure is accounted for since the shift-share is built using linkages from the full network through the matrix of $\xi^r_{ij}$.

\section{Empirical Results}\label{results}
This section provides the results from the empirical analysis of how demand shocks propagate along the supply chain to industry output. 
\subsection{Demand Shock Amplification and Supply Chain Positioning} \label{resultsindustrysec}

The goal is to estimate the elasticity $\beta$ in eq. (\ref{eq:elast}) and how it varies along the supply chain. I split the upstreamness distribution in bins through dummies taking values equal to 1 if $U_{it-1}^r\in [1,2]$ and $[2,3]$, and so on.\footnote{Since only 0.5\% of the observations are above 6, I include them in the last bin, $\mathbbm{1}\{ U_{it-1}^r \in [5,\infty) \}$.} I use lagged upstreamness as the contemporaneous one might be affected by the shock and represent a bad control.\footnote{The results are quantitatively identical if I use the baseline year upstreamness or average upstreamness so that industries cannot switch bin across year.} I estimate
\begin{align}
\Delta \log Y_{it}^r= \sum_{j=1}^5\beta_j \mathbbm{1}\{ U_{it-1}^r \in [j,j+1] \} \hat \eta_{it}^r+\delta_i^r+\nu_{it}^r .
\label{cardinalreg}
\end{align}
The estimated coefficients are plotted in Figure \ref{margins_cardinal}, while the regression output is reported in Table \ref{shock_iomatrix} in the Appendix. %I discuss the problem of inference in light of \cite{adao} in Appendix \ref{app_empirics_s4}.
\begin{figure}[ht]
\caption{Effect of Demand Shocks on Output Growth by Upstreamness Level}
\label{margins_cardinal}\centering
\includegraphics[width=.75\textwidth]{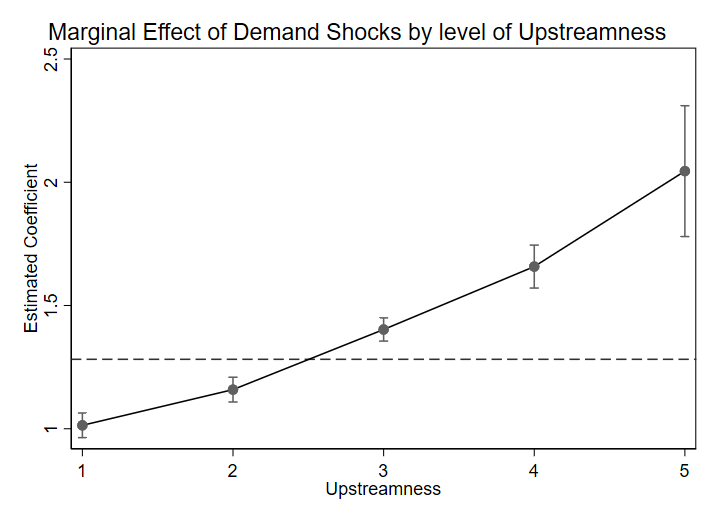} 
\\\vspace{5pt}
\noindent\justifying
\footnotesize Note: The figure shows the marginal effect of demand shocks on industry output changes by industry upstreamness level. The dashed horizontal line represents the average coefficient. The vertical bands illustrate the 95\% confidence intervals around the estimates. The regression includes country-industry fixed effects, and the standard errors are cluster-bootstrapped at the country-industry level. Note that due to relatively few observations above 6, all values above 6 have been included in the $U\in[5,\infty)$ category. The full regression results are reported in the first column of Table \ref{shock_iomatrix}.
\end{figure}

The results show that the same shock to the growth rate of final demand induces strongly heterogeneous responses in the growth rate of industry output. In particular, industries between one and two steps removed from consumers respond approximately 50\% less than industries five or more steps away. These novel results, which are robust across different fixed effects specifications, highlight how amplification along the production chain can generate sizable heterogeneity in output responses.
Quantitatively, each additional unit of distance from consumption raises the responsiveness of industry output to demand shocks by .22pp, approximately 17\% of the average response.

Throughout, I exploit only within-industry variation, controlling for all time-invariant country-industry factors with fixed effects. This rules out alternative explanations, such as that industries located more upstream tend to produce more durable goods and, therefore, be exposed to different intertemporal elasticity of substitution in household purchases.\footnote{Given the construction of $\hat\eta^r_{it}$, these shocks include indirect exposure to other sectors' durability through changes in demand. For example, insofar as computers are a key input in the financial services production function, and microprocessors an input in computers, the microprocessor industry is exposed to changes in the demand for both computers and financial services, despite computers being more durable than financial services.}

Figure \ref{margins_cardinal} shows an increasing output elasticity to demand shocks along the supply chain. On the one hand, if industries at different positions of the production network experience the same demand fluctuations, we should observe output volatility is higher for industries further away from consumption. On the other hand, if diversification forces are at play, we would expect that more upstream industries face lower demand volatility as shocks are partially diversified away. Therefore, it is a priori unclear how output volatility should change along supply chains. Empirically, I find support for both forces. I uncover two novel observations shown in Figure \ref{diversification}: i) the volatility of the measured demand shocks $\hat\eta^r_{it}$ negatively correlates with upstreamness, and ii) the volatility of output growth positively correlates with upstreamness. The first observation is intuitive: more upstream firms are indirectly connected to a larger number of destination countries, and therefore, destination-specific fluctuations are partly diversified away. However, if shocks dissipate as they travel along the supply chain, we should find an even stronger negative correlation between output growth volatility and upstreamness. Instead, we see that the upstream amplification shown in Figure \ref{margins_cardinal} dominates diversification forces, and output volatility is higher for more upstream sectors. This second result, therefore, requires a new mechanism underlying how shocks amplify upstream. In what follows, I study the role of inventories in explaining this puzzle.

\begin{figure}[ht]
\caption{Demand Shocks and Output Growth Volatility by Upstreamness}
\label{diversification}\centering
\begin{minipage}[b]{.5\linewidth}
\centering\includegraphics[width=.92\textwidth]{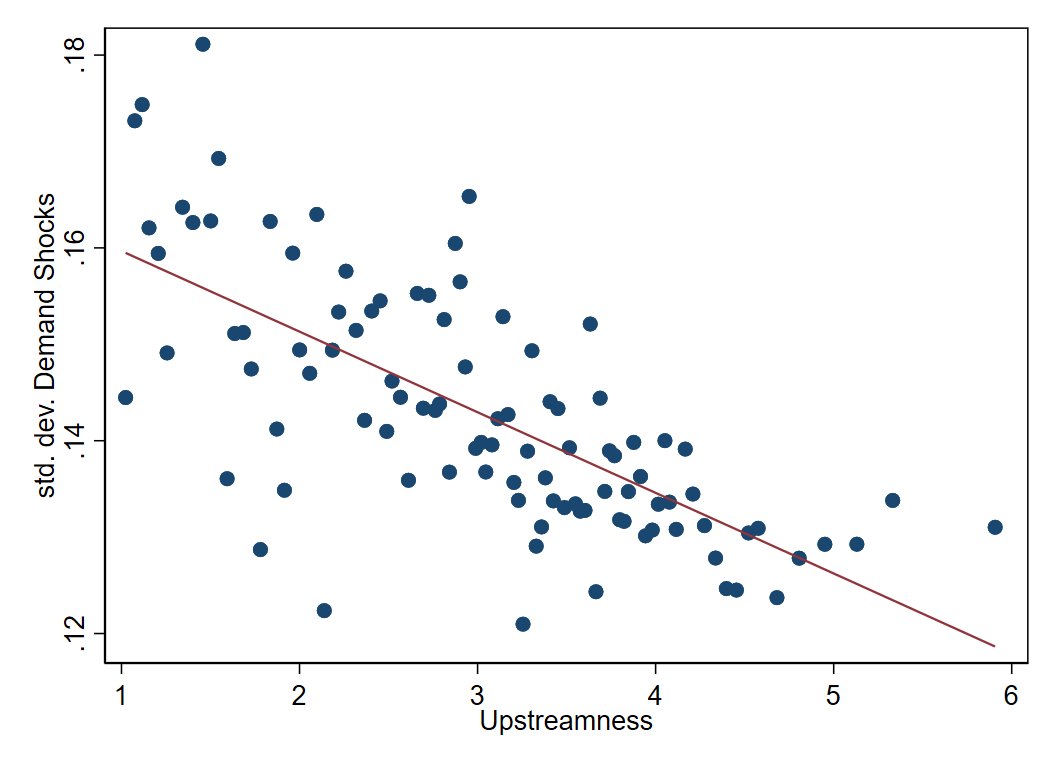}  \subcaption{Demand Shocks Volatility}
\end{minipage}%
\hspace*{.15cm}\begin{minipage}[b]{.5\linewidth}
\centering\includegraphics[scale=0.21]{{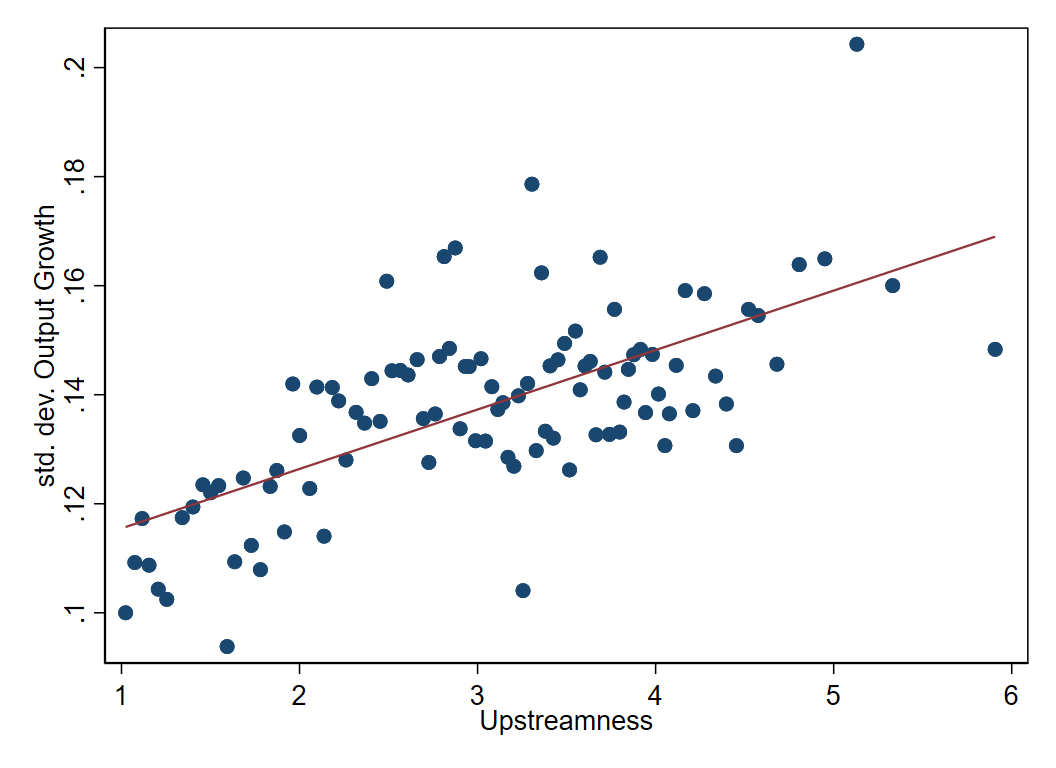}}\subcaption{Output Growth Volatility} 
\end{minipage}
\\\vspace{5pt}
\noindent\justifying
 \footnotesize Note: The graph shows the binscatter of the standard deviation of demand shocks and output growth within country-industry across time versus the industry average upstreamness across time.
\end{figure}

%In summary, more upstream industries face smaller fluctuations in their effective demand, yet they show higher volatility in their output growth. These differences are quantitatively important as the elasticity of output growth to changes in demand more than doubles along the upstreamness distribution. 

\paragraph{Robustness and Inference.} I confirm these results using alternative definitions of the final demand shocks. In particular, Appendix \ref{app_empirics_s4} confirms the findings using an alternative instrument based on foreign government consumption. Furthermore, earlier versions of this paper showed that the same results hold using the China Syndrome shocks from \cite{david2013china}, and the federal spending shocks from \cite{aak}.

To assess the robustness of these results, I run an extensive set of additional checks, which I discuss in detail and report in Appendix \ref{robustnessappendix}. These include using the re-centred instrument proposed by \cite{borusyak2020non} to solve potential omitted variable bias and a more general model to estimate the demand shifters, allowing for supply-side effects.

Lastly, note that absent any correction, the standard errors estimated in eq. (\ref{cardinalreg}) can be overestimated, as pointed out by \cite{adao}. In Appendix \ref{app_empirics_s4}, I show the exact procedure by which I apply the \cite{adao} correction. The results are, if anything, more precise.

\subsection{The Role of Inventories}
These results show that firms further upstream in the supply chain respond significantly more strongly to the same changes in demand. The operations literature on the \textit{bullwhip effect} posits that this finding can be explained by procyclical inventory adjustment. However, we lack consistent evidence of this mechanism. 

WIOD provides information on the change in inventories of a producing industry computed as the row residual in the I-O matrix. Intuitively, given the accounting identity that output equals sales plus the change in inventory stock, the tables provide the net change in inventories as residual between output and sales to other industries or final good consumers. I standardize the change in inventories by dividing it by total output to account for scale effects. I then estimate eq. (\ref{cardinalreg})
with $\Delta I^r_{it}/Y_{it}^r$ on the left hand side.\footnote{The WIOD data does not include the values of the stocks $I_{it}^r$, I divide by the value of output in a given period to avoid mechanical scale effects.}  I winsorize $\Delta I^r_{it}/Y_{it}^r$ at 1\%. The results are plotted in Figure \ref{margins_cardinal_inventories} and reported in Table \ref{cardinal_inv_table}.

\begin{figure}[ht]
\centering
\caption{Effect of Demand Shocks on Inventory Changes by Upstreamness Level}
\label{margins_cardinal_inventories}
  \centering
  \includegraphics[width=.75\linewidth]{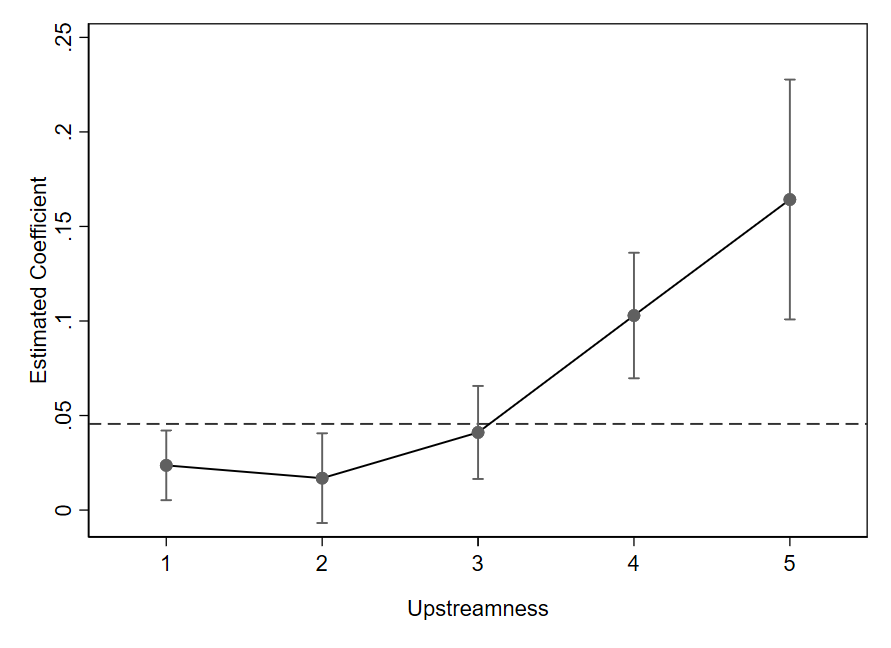}
\\\vspace{5pt}
\noindent\justifying
\footnotesize Note: The figure shows the marginal effect of demand shocks on industry inventory changes by industry upstreamness level. The dashed horizontal line represents the average coefficient. The vertical bands illustrate the 95\% confidence intervals around the estimates. The regression includes country-industry fixed effects, and the standard errors are cluster-bootstrapped at the country-industry level. Note that due to relatively few observations above 6, all values above 6 have been included in the $U\in[5,6]$ category. The full regression results are reported in the first column of Table \ref{cardinal_inv_table}.
\end{figure}

The estimation shows three important results. First, changes in final demand induce procyclical changes in inventories, independently of the industry position in the supply chain. Second, recalling that the average inventory-to-sales ratio in the data is approximately 10\%, the inventory response is quantitatively sizeable, ranging between 1\% and 14\% of annual output. Finally and most importantly, the response of inventories increases along the upstreamness distribution. Quantitatively, a 1pp increase in the growth of demand generates a .02pp increase in the change in inventories over output for industries closest to final consumers. The same figure for industries at 5 or more steps of production away is .16pp.

These findings suggest that inventories act as a force of upstream amplification in supply chains. Following the insights of the operations literature on the \textit{bullwhip effect}, in the next section, I build a novel production network model with inventories. I do so with the twofold goal of i) characterizing under which condition on the network and the inventory problem we would observe upstream amplification of shocks and ii) studying counterfactuals that can help shed light on the aggregate effects of the trends discussed in the introduction.

\section{Model }\label{model}
In this section, I study the interplay between supply chain structure and inventories in determining the extent of demand shocks propagation and amplification. I do so in three steps, adding one element at a time. I start by studying the problem in a vertical chain with inventories to isolate the sufficient condition on the inventory policy for upstream amplification. Next, I consider a general network with a single source of fluctuations. This allows me to isolate the interplay between inventories and supply chain position. Lastly, I consider a setting with many destination markets to isolate the role of diversification. Next, I provide a series of comparative statics to highlight how more or less fragmented economies compare in terms of output volatility. I conclude the section by directly estimating the model and providing counterfactual scenarios.

\subsection{Vertically Integrated Economy}\label{productionline}
Consider an economy with one final good whose demand is stochastic, and $N-1$ stages sequentially used to produce the final good. Throughout, I use industry and sector interchangeably. The structure of this production network is a line, where stage $N$ provides inputs to stage $N-1$ and so on until stage $0$, where goods are consumed.\footnote{I use the indexation of industries running from 0, final consumption, upward in the supply chain because, in this vertically integrated economy, industry $n$ has upstreamness equal to $n$.}

The demand for each stage $n$ in period $t$ is $D_t^n$ with $n\in\mathbb{N}$. Stage 0 demand, the final consumption stage, is stochastic and follows an AR(1) with persistence $\rho\in (0,1)$ and a positive drift $\bar D$. The error terms are distributed according to some finite variance distribution $F(0,\sigma^2)$ on a bounded support. $\bar D$ is assumed to be large enough relative to the variance of the error so that demand bounded away from 0. %\footnote{Including the positive drift does not change the inventory problem since, for storage, the relevant statistic is the first differenced demand.}
Formally, final demand in period $t$ is
$D_t^0=(1-\rho)\bar{D}+\rho D_{t-1}^0+\epsilon_t$. The production function is linear: for any stage $n$, if production is $Y_t^n$, it also represents the demand for stage $n+1$, $D_t^{n+1}$: $Y_t^n=D_t^{n+1}$.

Firms at stage $n$ form expectations on future demand $\mathbbm{E}_t D^n_{t+1}$ and produce to end the period with target inventories $I^n_t=I_n(\mathbb{E}_t D^n_{t+1})$. Where $I_n(\cdot)$ is some non-negative differentiable policy function whose derivative is $I^\prime_n$.

Finally, the following accounting identity has to hold at every stage $n$
\begin{align}
Y_t^n&=D_t^n+I_n(\mathbb{E}_t D^n_{t+1})- I_n(\mathbb{E}_{t-1} D^n_{t}),
\label{outputlinestagen}
\end{align}
Where $D^n_t$ is the demand for sector $n$'s products. In this setting, it is possible to derive how output behaves at every step of production $n$ by solving the economy upward from final demand. To provide a characterization of output at each stage of production, define the operator $\mathcal{F}^n_t[x_t]=I_n\left(\mathbb{E}_t\left[x_{t+1}+ \sum_{i=0}^{n-1}\mathcal{F}^{i}_{t+1}[x_{t+1}]- \mathcal{F}^{i}_t[x_t]\right]\right)$ with the initial condition $\mathcal{F}^0_t[x_t]=I_0(\mathbb{E}_t [x_{t+1}])$, for some argument $x_t$. Further, define $\Delta \mathcal{F}^n_t[x_t]= \mathcal{F}^n_t[x_t]-\mathcal{F}^n_{t-1}[x_{t-1}].$ Then, the following holds
\begin{lemma}[Sectoral Output]\label{thmoutputline}
    Sectoral output at production stage $n$ satisfies 
    \begin{align}
            Y^n_{t}&=D_t^0+\sum_{i=0}^n\Delta\mathcal{F}^i_t[D_t^0].
    \end{align}
    %Where $\mathcal{F}^i_t$ denotes the operator defined above. 
\end{lemma}
\begin{proof}
 See Appendix \ref{proofs}.
\end{proof}
Output at every stage of production $n$ is given by a recursion of final demand $D_t^0$ and inventory changes in all stages downstream from $n$, $\{I_{j}\}_{j=n}^0$.

In the context of this model, asking whether exogenous changes in final demand amplify upstream is effectively comparing $\frac{\partial \log Y^n_t}{\partial \log D^0_t}$ and $\frac{\partial \log Y^{n+1}_t}{\partial \log D^0_t}$: amplification occurs if there exists an $n$ such that $\frac{\partial \log Y^n_t}{\partial\log D^0_t}<\frac{\partial\log Y^{n+1}_t}{\partial\log D^0_t}$. Proposition \ref{propbullwhip} formalizes the sufficient condition for amplification in this economy.%\footnote{I generalize this result to the case in which firms have heterogeneous inventories in Proposition \ref{hetinventoriesprop} in section \ref{heteroginv} of the Online Appendix.}
\begin{prop}[Amplification in Vertically Integrated Economies]\label{propbullwhip}
The following holds
\begin{enumerate} [label=\alph*)]
    \item A change in final demand, to a first-order approximation around the non-stochastic
steady state, implies the following change in sectoral output 
\begin{align}
        {\Delta\log Y^n_t}\approx \Delta\log D^0_t\left[        1+\sum_{i=0}^n\rho I^\prime_i\prod_{j=0}^{i-1} [1+(\rho-1)I^\prime_{j}]\right].
\end{align}
\item The economy features upstream amplification of final demand shocks if the inventory function satisfies
\begin{align}
    0\leq I_i^\prime<\frac{1}{1-\rho},\:\forall i,\, \text{strictly for some } i.
    \label{Iprimecond}
\end{align}
\item Furthermore, if eq. (\ref{Iprimecond}) holds and $I^\prime_n$ is not too decreasing in $n$, amplification is larger for firms in more upstream sectors:
% \begin{align}
%     \diffp{{Y^n_t}}{{D^0_t}}-\diffp{{Y^{n-1}_t}}{{D^0_t}}=%\diffp{{\mathcal{F}^n_{t}[D_t^0]}}{{D^0_t}}>0
%     \rho I^\prime_n\prod_{j=0}^{n-1} [1+(\rho-1)I^\prime_{j}]%\diffp{{\mathcal{F}^n_{t}[D_t^0]}}{{D^0_t}}>0
% \end{align}
$\Delta\log Y^n_t-\Delta\log Y^{n-1}_t$ increases in $n$.
\end{enumerate}
\end{prop}
\begin{proof}
 See Appendix \ref{proofs}.
\end{proof}
Proposition \ref{propbullwhip} provides three important results. The first is a characterization of how a change in final demand affects output. In this vertically integrated economy, absent inventories, output would move 1-for-1 with demand. The presence of inventories implies that demand fluctuations can be absorbed or amplified as they travel through the network. When inventories are countercyclical, they absorb part of the fluctuation in demand without transmitting it to output. Vice versa, if inventories are procyclically updated, they amplify such fluctuations.  This intuition is formalized in the second result. The proposition provides a sufficient condition for amplification in eq. (\ref{Iprimecond}). The first inequality requires that the inventory function is increasing, namely that inventory changes are procyclical. This ensures that, as demand rises, so do inventories. Whenever this holds, output increases more than one-for-one with demand. This, in turn, implies that the demand change faced by the upstream firm is amplified relative to the one faced by the downstream firm. In other words, the demand shock amplifies upstream. The second inequality requires that the function is not "too increasing" relative to the persistence of the process. This inequality arises because a positive change in demand today also implies that the conditional expectation of demand tomorrow is lower than demand today due to mean-reversion. This condition ensures that the first, direct effect dominates the second one. Intuitively, as shocks become arbitrarily close to permanent, the second condition is trivially satisfied, and it is enough for inventories to be increasing in expected demand. The last result in Proposition \ref{propbullwhip} states that if all firms have procyclical inventories and such procyclicality is not smaller for more upstream firms, amplification magnifies upstream. Namely, we should find an increasing and convex relation between output and upstreamness. Intuitively, every step of production amplifies demand fluctuations; therefore, more upstream firms have a larger output response for any given change in final demand.

Proposition \ref{propbullwhip} shows that when inventories are procyclically adjusted, shocks amplify upstream. In the data, I find that this condition is verified. Using both the NBER CES Manufacturing and the Census data, I estimate the empirical derivative of inventories to sales. Figures \ref{dist_Iprime}, \ref{dist_Iprime_type} and Table \ref{npreg} show the results. I find that for most sectors in the data, $I^\prime>0$, suggesting that inventories are procyclically adjusted. As a natural consequence, I also find that output is more volatile than sales, as shown in Figure \ref{fact4plots}.
%An alternative way of summarizing the intuition is the following. In vertically integrated economies without labor and inventories, changes in final demand transmit one-to-one upstream, as no substitution is allowed across varieties. When such an economy features inventories, this result need not hold. If inventories are used to smooth production, meaning that $I(\cdot)$ is a decreasing function, shocks can be transmitted less than one-to-one as inventories partially absorb them. On the other hand, when inventories are adjusted procyclically, the economy features upstream amplification.

So far, I have established that in line networks, procyclical inventories are a sufficient condition for upstream amplification of demand shocks. Importantly, in this setting, due to production taking place on a line with only one endpoint, the network structure induces the largest possible extent of amplification. In the next sections, I study generalizations of the network that can undo this effect. I do this in two steps. First, under the maintained assumption of only one source of shocks, I allow for a general network structure to introduce the possibility of diffusion of shocks across sectors. Next, I allow for consumers in many countries and many shocks to introduce diversification forces. 
However, to solve the model, I have to specify the inventory problem.

\subsection{Network Structure and Amplification}\label{generalnetwork}

%In this section, I embed the inventory problem in a general production network framework to study how the network structure interplays with the inventory amplification mechanism. Relative to the standard network model in \cite{acemogluetal}, I introduce three main changes: i) stochastic foreign demand, ii) two types of firms as in \cite{acemoglu2020firms}, and iii) inventories. 

The general network structure introduces a horse race between the network's ability to dissipate final demand shocks and the potential amplification from inventories. As highlighted in Proposition \ref{propbullwhip}, the sufficient condition for amplification is that inventories are adjusted procylically. In this section, to retain the tractability of the model despite the general network structure, I specify the inventory choice of firms such that it replicates the procyclical behavior observed in the data while still allowing me to fully characterize the equilibrium in closed form.% I show in Appendix \ref{quant_model} that a simple dynamic model of production and inventory choices in the presence of stochastic demand naturally generates an optimally procyclical inventory policy.

\paragraph{Households.} Suppose the economy is populated by domestic and foreign consumers. Domestic consumers have \cite{golosov2007menu} preferences over a differentiated bundle $C$ and leisure. The utility is given by 
$    U_t= \log C_t - L_t
$, where $L$ is the amount of labor supplied. They maximize utility subject to the budget constraint $w_tL_t +\Pi_t=P_tC_t$, where $P_t$ is the optimal price index of the differentiated bundle and $\Pi_t$ represents rebated firm profits. The wage is the numeraire so that $w_t=1,\,\forall t$. The household maximization yields a constant expenditure on the differentiated bundle equal to 1. Foreign consumers have a stochastic demand $X_t$, which follows an $AR(1)$ process with some mean $\bar X$. The total final expenditure faced by a firm is then given by 
$
D_t=(1-\rho)(1+\bar X)+\rho D_{t-1}+\epsilon_t,\, \epsilon_t\sim F(0,\sigma^2).
$
Note that total demand inherits the stochasticity from foreign demand. I assume that the composition of the domestic and foreign consumption baskets is identical and generated through a Cobb-Douglas aggregator over varieties
$
    C_t=\prod_{s\in S} C_{s,t}^{\beta_s},
$
where $S$ is a finite number of available products. $\beta_s$ the consumption weight of good $s$ and $\sum_s\beta_s=1$. This formulation implies that the expenditure on good $s$ is $E_{s,t}=\beta_s D_t$ for $E_{s,t}$ solving the consumer expenditure minimization problem.

\paragraph{Production.} The network is characterized by an input requirement matrix $\mathcal{A}$, in which cycles and self-loops are possible.\footnote{An example of a cycle is: if tires are used to produce trucks and trucks are used to produce tires. Formally, $\exists r: [\mathcal{A}^n]_{rr} >0, n>1$. An example of a self-loop is if trucks are used in the production of trucks. Technically, such is the case if some diagonal elements of the input requirement matrix are positive, i.e., $\exists r: [\mathcal{A}]_{rr}>0$.}
I denote elements of $\mathcal{A}$ as $a_{rs}=[\mathcal{A}]_{rs}$. The network has a terminal node given by final consumption.

In each sector $s$, there are two types of firms: a fringe of competitive firms, denoted $C$ firms, and a set of firms with market power, denoted $I$ firms. Goods produced by $C$ and $I$ firms within the same sector $s$ are perfect substitutes.\footnote{See \cite{acemoglu2020firms} for a similar setup.}
All firms produce using labor and a bundle of other sectors' output. Inputs are combined through Cobb-Douglas production functions
$
   Y_{s,t}=Z_sl_{s,t}^{1-\gamma_s} M_{s,t}^{\gamma_s} ,
$
where $l_{s,t}$ is the labor used by industry $s$, $M_{s,t}$ is the input bundle and $\gamma_s$ is the input share for sector $s$. $Z_s$ is an industry-specific normalization constant. The input bundle is aggregated as
$
      M_{s,t}=\left(\sum_{r\in R} {a_{rs}}^{1/\nu}Y_{rs,t}^{\frac{\nu-1}{\nu}}\right)^{\frac{\nu}{\nu-1}},
$
where $Y_{s,t}$ is the output of sector $s$, $Y_{rs,t}$ is the output of industry $r$ used in sector $s$ production, and $\gamma_s =\sum_{r\in R}a_{rs}$ so that the production function has constant returns to scale. $\nu$ is the elasticity of substitution, and $a_{rs}$ is an input requirement, in equilibrium this will also coincide with the expenditure amount $Y_{rs,t}$ needed for every dollar of $Y_{s,t}$. $R$ is the set of industries supplying inputs to sector $s$.\footnote{In this paper, I focus on a setting where the production network is given by technology, namely the input requirement matrix $\mathcal{A}$ and study comparative statics and counterfactuals on the network structure. For recent contributions studying the role of endogenous network formation for the transmission of shocks, see \cite{lim, huneeus,acemoglu2020endogenous,taschereau,acemoglu2024macroeconomics,kopytov}. } 

The fringe of competitive $C$ firms is not subject to inventory management problems. They choose output $Y_{s,t}$ and inputs $l_{s,t},\,\{Y_{rs,t}\}_r$ to maximize profits $\pi_{s,t} = p_{s,t} Y_{s,t} - l_{s,t} -\sum_r p_{r,t} Y_{rs,t}$ subject to the production technology. The set of $I$ firms produces the same varieties with a productivity shifter $Z_s^I>Z_s$ taking input prices as given and solves 
\begin{align*}
    \max_{Y_{s,t},I_{s,t}, Q_{s,t}}\;&\mathbb{E}_t\sum_{t} \beta^t\left[p_{s,t} Q_{s,t} - c_{s,t}Y_{s,t} -\frac{\delta}{2}(I_{s,t}-\alpha Q_{s,t+1})^2 \right] \quad st\quad I_{s,t}=I_{s,t-1}+Y_{s,t}-Q_{s,t},
\end{align*}
where $Q_{s,t}$ is the quantity sold, $Y_{s,t}$ is the quantity produced, and $c_{s,t}$ is the marginal cost of the expenditure minimizing input mix: $\argmin_{l_{s,t},\,\{Y_{rs,t}\}_r} l_{s,t} +\sum_r p_{r,t} Y_{rs,t} $ subject to $ \bar{Y}_{s,t}=Z^I_sl_{s,t}^{1-\gamma_s} \left(\sum_{r\in R} {a_{rs}}^{1/\nu}Y_{rs,t}^{\frac{\nu-1}{\nu}}\right)^{\gamma_s\frac{\nu}{\nu-1}}$. Note that the model abstracts from productivity shocks. I introduce them in an extension discussed later on. $\delta$ and $\alpha>0$ govern the costs of holding inventories or facing stock-outs and backlogs.\footnote{This model of inventory choice is a simplified version of the linear-quadratic inventory model proposed by \cite{ramey1999inventories} following \cite{holt1960planning} as a second-order approximation of the full inventory problem. I discuss a fully dynamic model in which firms face breakdowns and stochastic demand in Appendix \ref{quant_model}. Proposition \ref{propendinventories} shows that inventories are optimally procyclical in that setting.} The optimal inventory policy is a function of the expected demand:
\begin{align}
    I_{s,t}=\max\{ \mathcal{I}_{s,t}+\alpha \mathbb{E}_tQ_{s,t+1},0\},
    \label{optimalinv}
\end{align}
with $\mathcal{I}_{s,t}\coloneqq (\beta-1)c_{s,t}/\delta$.\footnote{Note that $\mathcal{I}_{s,t}<0$ since, in the presence of time discounting or depreciation of inventories, the firm would ideally like to borrow output from the future and realize the sales today.} This optimal rule, known in the operations literature as \textit{days-worth of supply (DOS)}, implies that inventories are procyclically adjusted, as is corroborated by the inventory data (see Figure \ref{dist_Iprime}). This optimal policy is linear in expected sales. This fits the data extremely well, as shown in Figure \ref{lowess_graphs}, with a linear fit explaining 80\% of the variation. This formulation of the problem, where the presence of inventories is motivated directly by the structure of the firm's payoff function, is a reduced form stand-in for stock-out avoidance motives. I discuss this assumption in detail at the end of the section.\footnote{In Appendix \ref{quant_model}, I provide a simple dynamic model in which firms face stochastic demand to show that the optimal dynamic policy implies procyclical inventory changes. Secondly, note that, in this setup, procyclicality follows from the optimal target rule adopted by firms. An alternative motive for holding inventories could be production smoothing, whereby a firm holds a stock of goods to avoid swings in the cost of production between periods. In Appendix \ref{smoothingmotive}, I introduce a production smoothing motive in the form of a convex cost function and show that the firm optimally chooses procyclical inventories if the smoothing motive is not too strong. Furthermore, if the production smoothing motive were to dominate, inventories would have to be countercyclical, which is counterfactual based on the evidence discussed earlier. Finally, this version of the model abstracts from both productivity shocks and inventory depreciation. I introduce stochastic productivity in Appendix \ref{productivityshocks} and show that it reinforces the procyclical nature of inventories. I abstract from depreciation as it does not affect the results qualitatively.} Finally, since this target rule sets the optimal inventory stock at the end of the period, the production choice depends on the optimal inventory target and the current inventory stock. In particular, the firm will choose a production level that meets demand and obtains the targeted change in inventories $\Delta I_{s,t} = I^\star_{s,t}-I_{s,t-1}$. In this sense, the output choice is both forward- and backward-looking since it depends on the current inventory stock as well as expectations of future demand.

%This formulation of the problem has two important advantages. First, the optimal inventory policy implies a linear affine mapping between inventory holdings and future sales (as given by a constant target rule), which matches the data extremely well, as shown in Figure \ref{lowess_graphs} in the Online Appendix. Second, it allows for a full closed-form characterization of the equilibrium that would not be possible otherwise. I provide the equilibrium definition and characterization in Appendix  \ref{proofs}.

\paragraph{Characterization.} I set the vector of normalizing constants for $C$ firms $Z_s\coloneqq (1-\gamma_s)^{\gamma_s-1}\gamma_s^{-\gamma_s\frac{\nu}{\nu-1}}$ so that, together with the normalization $w_t=p_{0,t}=1$, they imply that the expenditure minimizing input bundles have a marginal cost $c_{s,t}=1,\,\forall s,\,t$ \citep[see][]{carvalho2019production}. The competitive $C$ firms then set prices $p_{s,t}=1,\,\forall s,\,t$. 
The set of $I$ firms price at the marginal cost of the competitive fringe and obtain a markup $\mu_s^I$ over their marginal cost. As a consequence, $I$ firms make profits that are rebated to households while $C$ firms do not produce in equilibrium.
$I$ firms then optimally choose an inventory policy $I_{s,t}=\max\{ \mathcal{I}_{s,t}+\alpha \mathbb{E}_tQ_{s,t+1},0\}$.\footnote{In what follows I disregard the possibility that the optimal inventory level is 0 since there always exists an average foreign demand $\bar X$ such that it is never optimal to hold no inventory.}
Note that in equilibrium, there is no difference between the value and quantity of output in this economy, as all prices are equal to 1.  The equilibrium definition and full characterization is provided in Appendix  \ref{proofs}.

Given these policies, it is possible to solve for equilibrium quantities analytically. The linearity of the policies implies that I can solve the problem separately at each step of production and then aggregate across stage of production within each firm. %: consider a firm selling some of its output directly to consumers and the complementary output to other firms. The part of output sold to consumers, denoted by the superscript 0, is
% $
%     Y^0_{s,t}=\beta_s[D_t+\alpha\rho\Delta_t],\,
% $ 
% with $\Delta_t=D_t-D_{t-1}$. $Y^0_{s,t}$ also represents the input expenditure of sector $s$ to its generic supplier $r$ once it is rescaled by the input requirement $\gamma_s a_{rs}$. Hence, output of producers one step of production removed from consumption obtains by summing over all final good producers $s$. Market clearing and the inventory policy imply $Y^1_{r,t}=\sum_s\gamma_s a_{rs}Y^0_{s,t}+\Delta I_{r,t}^1$ and, denoting $\omega=1+\alpha(\rho-1)$, therefore
%  $
%     Y^1_{r,t}=\sum_s\gamma_s a_{rs}\left[D_t+\alpha\rho\sum_{i=0}^1\omega^i\Delta_t \right].
% $
Denote $\gamma_sa_{rs}=\tilde{\mathcal{A}}_{rs}$, so that  $\sum_s\gamma_s a_{rs}={\sum_s\tilde{\mathcal{A}}_{rs}} $ is the weighted outdegree of a node $r$, namely the sum of the shares of expenditure of all industries $s$ coming from input $r$. Also denote $\omega=1+\alpha(\rho-1)$ and $ \tilde{\mathcal{A}}^n$ as the $n^{th}$ power of the $\tilde {\mathcal {A}}$ matrix.
Iterating upward from stage $0$ to generic stage $n$ and denoting $\Delta_t \equiv D_t-D_{t-1}$, %and noting that  $
%\sum_v\tilde{\mathcal{A}}_{kv}\sum_q\tilde{\mathcal{A}}_{vq}\hdots\sum_r\tilde{\mathcal{A}}_{or}\sum_s\tilde{\mathcal{A}}_{rs}\beta_s=\left[\tilde{\mathcal{A}}^n\right]_kB,
%$ 
I can write the value of production of industry $k$ at stage $n$ as
\begin{align}    Y^n_{k,t}=\left[\tilde{\mathcal{A}}^n\right]_kB \left[D_t+\alpha\rho\sum_{i=0}^n\omega^i \Delta_t\right].
    \label{outputnework}
\end{align}
In eq. (\ref{outputnework}), the network structure is summarized by $\left[\tilde{\mathcal{A}}^n\right]_kB$. The first term $\left[\tilde{\mathcal{A}}^n\right]_kBD_t$ incorporates both the direct and indirect exposure to contemporaneous demand, while the rest of the equation represents the inventory effect both directly and indirectly through the network.

Finally, as firms operate at multiple stages of production, total output of firm $k$ is
$
    Y_{k,t}=\sum_{n=0}^\infty Y^n_{k,t}.
$
I can now characterize sectoral output as a function of the inventory channel and the features of the network.

\begin{prop}[Sectoral Output]\label{indoutput}
The value of sectoral output for a generic industry $k$ is given by
\begin{align}
Y_{k,t}=\sum_{n=0}^\infty  \left[\tilde{\mathcal{A}}^n\right]_kB\left[D_t+\alpha\rho\sum_{i=0}^n\omega^i \Delta_t\right].
\end{align}
This can be written in matrix form as
\begin{align}
     Y_{k,t}= \tilde L_k B D_t+ \alpha\rho\left[\sum_{n=0}^\infty\tilde {\mathcal{A}}^n\sum_{i=0}^n\omega^i\right]_kB \Delta_t ,
     \label{sales_mat}
\end{align}
where $B$ is the $S\times 1$ vector of consumers' expenditure shares $\beta_s$ and $\tilde L_k$ is the $k^{th}$ row of the Leontief inverse, defined as 
$    \tilde L= [I+ \tilde{\mathcal{A}}+ \tilde{\mathcal{A}}^2+...]= [I-\tilde{\mathcal{A}}]^{-1}.
$ Where $\tilde{\mathcal{A}}\coloneqq \mathcal{A}\hat \Gamma$ and $\hat \Gamma=\diag\{\gamma_1,..., \gamma_R\}$. 
Sectoral output exists non-negative for any $\alpha,\rho$ such that $\omega\in[0,1]$.
\end{prop}
\begin{proof}
See Appendix \ref{proofs}.
\end{proof}

Proposition \ref{indoutput} generalizes Lemma \ref{thmoutputline} along the network dimension. Several features are worth discussing. The first observation is that the model collapses to the standard characterization of output in production networks when there is no inventory adjustment, as the second term in eq. (\ref{sales_mat}) vanishes to recover $Y_{k,t}=\tilde L_k B D_t$. This occurs whenever there are no inventories ($\alpha=0$),  when current shocks do not change expectations on future demand ($\rho=0$), and in the non-stochastic steady state where $\Delta_t=0,\,\forall t$. Second, output might diverge as $n\rightarrow\infty$ if $\alpha(\rho-1)>0$. However, by the assumptions made on $\tilde{\mathcal{A}}$,\footnote{In particular the fact that $\sum_k \tilde{\mathcal{A}}^{kv}<1$, i.e., the assumption that the firm labor share is positive.} and the maintained assumption that $\omega\in(0,1)$, additional distance from consumption implies ever decreasing additional output, so output converges.\footnote{ In Proposition \ref{propdag} in Appendix \ref{dagsection}, I show that restricting the network to a Directed Acyclic Graph allows existence and non-negativity even if $\tilde{\mathcal{A}}^n\sum_{i=0}^n\omega^i$ has a spectral radius outside the unit circle which is the sufficient condition used in Proposition \ref{indoutput}. } 

With the characterization in Lemma \ref{thmoutputline}, I can study the behavior of the growth rate of sectoral output in response to changes in demand. To do so, denote $\mathcal{U}_k=\sum_{n=0}^\infty\frac{1-\omega^{n+1}}{1-\omega}[\tilde{\mathcal{A}}^n B]_k \bar{D}/Y_k$ the inventory-based distance from final consumers of industry $k$. This statistic, combining properties of the network and of inventories along the supply chain, is bounded below by 1 and bounded above by the industry's upstreamness $U_k$. 
Note that $\mathcal U_k$ is equal to $U_k$ if there are no inventories and strictly larger when $\alpha,\rho>0$. Furthermore, since $\omega\in(0,1)$, additional stages of production are given decreasing weights while in $U_k$, they are given a constant weight. As a consequence, in the cross-section, $\mathcal U_k$ is concave in $U_k$.
With this definition, the following holds. 

\begin{prop}[Sectoral Output Changes]\label{propgrowthrates}
    The growth rate of sectoral output, to a first-order approximation around the non-stochastic steady state, satisfies
    \begin{align}
    \Delta\log Y_{k,t}
        &\approx\Delta\log D_t+\alpha\rho\mathcal{U}_k \Delta\log D_t.
        \label{outputgrowthsingle}
    \end{align}    
    Therefore, the inventory-based upstreamness is a sufficient statistic for the network effect on amplification. 
    Formally, comparing the output elasticity of industry $k$, $\varepsilon_{Y_k}\equiv \frac{\Delta Y_{k,t}}{Y_k}\frac{\bar D}{\Delta_t}=1+\alpha\rho\mathcal{U}_k$, to that of industry $j$ 
    \begin{align}
        \varepsilon_{Y_k}> \varepsilon_{Y_j}\Leftrightarrow \mathcal{U}_k> \mathcal{U}_j.
    \end{align}
\end{prop}
\begin{proof}
See Appendix \ref{proofs}.
\end{proof}
Proposition \ref{propgrowthrates} establishes two important results. First, it characterizes how the growth rate of sectoral output depends on the growth rate of final demand. This dependence comes in two forms. The first term, common across sectors, links directly demand and output changes one-for-one. This term states, that in the special cases where there are no inventories ($\alpha=0$) or shocks are not autocorrelated ($\rho=0$), the growth rate of output is the same as the growth rate of demand. This is a well-known property of network models without reallocation effects \cite[see][]{baqaee2019macroeconomic}. Importantly, this term also shows that the network structure does not matter absent inventories. The second term introduces heterogeneity across sectors, entirely governed by their upstreamness. More upstream sectors face higher inventory amplification along the supply chain and, as a consequence, feature higher output changes for a given change in final demand. The second part of the proposition highlights that differences in inventory-based upstreamness are a necessary and sufficient condition for differences in changes in output. 

As a further consequence of this result, Proposition \ref{corvolatility} relates differences in the position of industries in their supply chain to differences in output growth volatility.
\begin{prop}[Sectoral Output Growth Volatility]\label{corvolatility}
    The volatility of sectoral output growth $ \sigma_{\Delta \log Y_{k,t}}$ satisfies
    \begin{align}
        \sigma_{\Delta \log Y_{k,t}}&= \left(1+\alpha\rho \mathcal{U}_k\right)\sigma_{\Delta\log D_t}, 
    \end{align}
    Hence, industry $k$ is more volatile than industry $j$ if and only if it is more upstream:
     \begin{align}
        \sigma_{\Delta \log Y_{k,t}}> \sigma_{\Delta \log Y_{j,t}}\Leftrightarrow \mathcal{U}_k> \mathcal{U}_j.
    \end{align}
\end{prop}
\begin{proof}
See Appendix \ref{proofs}.
\end{proof}
The intuition behind this result is straightforward. For a given level of volatility in final demand, more upstream industries experience larger fluctuations in output and, therefore, higher output growth volatility. 

An important corollary of Propositions \ref{propgrowthrates} and \ref{corvolatility} is that the network structure matters even to a first order. In particular, the network is fully summarized by the inventory-weighted upstreamness $\mathcal{U}$.

\paragraph{Comparative Statics.} Propositions \ref{propgrowthrates} and \ref{corvolatility} establish that the inventory-weighted measure of upstreamness is a sufficient statistic to understand the cross-sectional heterogeneity in output growth fluctuations. In what follows, I present comparative statics to understand how this relates to fundamental features of the network economy. %In particular, I establish that, all else equal, i) higher inventories and ii) more persistence shocks induce larger fluctuations, iii) industries with a lower weight in consumption baskets are more volatile, and iv) that fragmenting production increases the volatility of the economy.

%To best understand the result in Proposition \ref{propgrowthrates}, note the following comparative statics
\begin{prop}[Output Growth increases in Inventories and Persistence of Shocks]\label{invcs}
    The following holds: 
    \begin{enumerate} [label=\alph*)]
        \item If firms hold larger inventories ($\alpha\uparrow)$, output elasticity $\epsilon_k$ increases $\forall k$.
        % $\omega$ increases and, therefore, so does $\mathcal{U}_k,\,\forall k$. Hence, 
        \item If shocks become more persistent ($\rho \uparrow)$, output elasticity $\epsilon_k$ increases $\forall k$.
        \item As shocks become arbitrarily close to permanent, the output elasticity is $\epsilon_k=1+\alpha{U}_k$, $\forall k$. 
    \end{enumerate}
\end{prop}
\begin{proof}
See Appendix \ref{proofs}.
\end{proof}
Proposition \ref{invcs} provides three intuitive results. First, as firms hold larger inventory stocks, as governed by $\alpha$, the elasticity of output to changes in final demand increases, driven by stronger inventory amplification along the chain. Second, more persistence in the stochastic demand process also induces larger responses. After a change in final demand, the presence of inventories generates two forces: the most direct effect is that, if $\rho>0$, an increase in demand today implies an increase in the conditional expectation of future demand. Firms respond by increasing their output elasticity from 1 to $1+\alpha\rho$. At the same time, due to mean-reversion, the conditional expectation decays at rate $\rho$. When $\rho<1$ firms do not increase inventories of the full $\alpha\rho$ amount but rather of $\alpha\rho(1+\alpha(\rho-1))$. As shocks become more persistent, the latter effect is weakened, and the procyclical adjustment of inventories generates larger output elasticities. This effect compounds through the network at every step of the production process. In the limit case of permanent shocks, $\rho\rightarrow 1$, this effect is eliminated altogether, and output grows linearly in the number of production steps between a given industry and consumers as measured by upstreamness, as stated in point c).

Next, to study the role of the network structure, I start with a useful definition adapted from \cite{la2022optimal}.

\begin{definition}[Downstream Symmetry]\label{defsymmetry}
%[Upstream, Downstream and Consumption Symmetry]
    %Industries $i$ and $j$ are upstream symmetric if $a_{ki}=a_{kj}$ for all industries $k$. They
    Industries $i$ and $j$ are downstream symmetric if $a_{ik}=a_{jk},\,\forall k$. %, and they are consumption symmetric if $\beta_i=\beta_j$.
\end{definition}

%Upstream symmetry means that two sectors share the same production technology, and, as a consequence, the network upstream from them is identical. 
Downstream symmetry implies that two industries have the same role as input suppliers to other firms. They may, however, have a different importance for final consumers. %When downstream symmetry is combined with consumption symmetry, it implies that they also have identical roles as suppliers of households. Given this definition, I can establish the following result.

\begin{prop}[Output Growth Volatility and Demand Composition]\label{propdemand}
    Suppose sectors $r$ and $s$ are downstream symmetric, but sector $r$ has a smaller consumption weight than sector $s$, $\beta_r<\beta_s$. Then sector $r$ is more volatile than sector $s$: $\Var(\Delta\log Y_{rt})>\Var(\Delta\log Y_{st})$.
\end{prop}
\begin{proof}
See Appendix \ref{proofs}.
\end{proof}

If two industries have the same role as suppliers in the production network (downstream symmetry), they are equally upstream to other sectors. Their total position in the network will only differ if they have different roles as suppliers to consumers. If a sector is more important in consumption baskets, a larger fraction of its output is sold directly to households. Therefore, a smaller fraction of the sector's production is subjected to amplification through inventories along the supply chain. As a consequence, the sector is less volatile. 

Importantly, note that, by assuming downstream symmetry, Proposition \ref{propdemand} excludes all possible differences in upstreamness stemming from different linkages with other sectors in the network. To clarify the role of sector-to-sector linkages for the level of volatility of the economy, Proposition \ref{propfragmentation} considers the case of production fragmentation.
\begin{prop}[Output Growth Volatility and Vertical Fragmentation]\label{propfragmentation}
    %Suppose a sector $i$ fragments into two new sectors $j$ and $k$, such that $j$ is upstream symmetric to $i$, $k$ is downstream and consumption symmetric to $i$, $j$ is the sole supplier of $k$, and $k$ is the sole buyer of $j$. Then, output growth volatility weakly increases for every sector in the economy. Strictly $\forall r$ such that $\ell_{ri}>0$ before the fragmentation.
    Consider the fragmentation of a sector $i$, such that all network paths from $i$ to any other sector $k$ go through an extended path $i\rightarrow j\rightarrow k,\,\forall k$. Then, output growth volatility weakly increases for every sector in the economy. Strictly $\forall r$ such that $\ell_{ri}>0$.
\end{prop}

\begin{proof}
See Appendix \ref{proofs}.
\end{proof}
Proposition \ref{propfragmentation} studies how pure vertical fragmentation alters the network structure and, as a consequence, the volatility of the economy. First, note that the case considered is one in which the only change in the network is the splitting of a sector $i$ into two industries $i$ and $j$ without any additional change upstream or downstream. The direct consequence of this shift is that any sector connected as a direct or indirect supplier to industry $i$ ($\ell_{ri}>0$) is subject to an additional step of upstream inventory amplification, while all the rest of the network is unchanged. All these sectors will, therefore, experience higher output volatility.

In summary, the economy features upstream amplification driven by firms' procyclical inventory policy. Fragmenting production, in this context, unequivocally increases the volatility of the economy as it increases the length of supply chains. This, in turn, interacts with the inventory adjustment, amplifying demand shocks upstream.

\subsection{Multiple Destinations}

To isolate the role of supply chain length, I have so far considered an economy with only one source of fluctuations, thereby eliminating any potential for diversification across space. To highlight the potential diversification effects of more complex supply chains, this section considers an extended model with multiple consumption destinations. There are $J$ destinations whose demand evolves according to AR(1) processes with mean $\bar D_j$ and variance $\sigma^2_j$. Demand is iid across space. Denote $\mathcal{U}_j^r$ the inventory-weighted upstreamness of sector $r$ relative to consumption destination $j$, where $\mathcal{U}^r=\sum_j\xi_j^r \mathcal{U}_j^r$, and $\xi_j^r$ is the fraction of output of sector $r$ consumed in country $j$, directly or indirectly, as defined in Section \ref{methodology}. Keeping the rest of the model unchanged, the following result extends Proposition \ref{propgrowthrates} to many destinations.\footnote{Throughout, I maintain that the parameter governing the optimal inventory-to-sales ratio is the same across sectors. Importantly, this implies assuming that it is not destination-specific. Prior literature has shown that firms engaging in trade tend to have higher inventory buffers. I abstract from this channel as this would imply heterogeneous inventory choices across sectors and prevent a closed-form characterization of the problem. I come back to this case later in the text and in Appendix \ref{hetinv_general_app}.}
\begin{prop}[Sectoral Output Growth]\label{growthmulti}
    The growth rate of sectoral output, to a first-order approximation around the non-stochastic steady state, satisfies
    \begin{align}
        \Delta\log Y^r_{t}\approx\eta_{t}^r+\alpha\rho\sum_j \mathcal{U}_{j}^r\xi_{j}^r\Delta\log D_{jt},
        \label{outputgrowthmulti}
    \end{align}
with $\eta_{t}^r=\sum_j\xi_j^r\Delta\log D_{jt}$. In the notation of Section \ref{methodology}, $\Delta\log D_{jt}=\eta_{jt}$.
\end{prop}
\begin{proof}
See Appendix \ref{proofs}.
\end{proof}
Proposition \ref{growthmulti} characterizes the behavior of sectoral output growth, generalizing Proposition \ref{propgrowthrates} to multiple destinations. First, note that in the first term, shocks to different final destination countries are aggregated via the shift-share structure used in Section \ref{methoddatasec}. The shares are given by the network exposure of industry $r$ to each destination $j$, $\xi_j^r$. Secondly, the intensity of upstream amplification is governed by a different shift-share term, which includes bilateral inventory-weighted upstreamness $\mathcal{U}^r_j$. This measure, defined as the bilateral version of the upstreamness index discussed above, computes the average number of production steps between sector $r$ and consumers of country $j$, appropriately accounting for the intensity of the inventory effects. 

The intuition underlying Proposition \ref{growthmulti} is as follows. Starting from a special case, suppose that industry $r$ sells only to consumers from country $j$ ($\xi_j^r=1$), then eq. (\ref{outputgrowthmulti}) collapses to (\ref{outputgrowthsingle}). Suppose instead that sector $r$ sells to two destinations $j$ and $k$ in equal proportions, but the length of the supply chains to these destinations differ. Then, changes in demand in country $j$ are subjected to upstream inventory amplification governed by the distance of sector $r$ to country $j$, as measured by the bilateral inventory-weighted upstreamness $\mathcal{U}^r_j$. Equivalently for country $k$. The total amount of amplification is a weighted average where the weights are given by the fraction of output sold in each destination, collected in the vector $\Xi^r$, such that $\xi_j^r=[\Xi^r]_j$.
Hence, an industry has a higher responsiveness to changes in demand if a large share of its output is eventually consumed in destinations that are far, in the upstreamness sense. 
Next, Proposition \ref{volatilitymulti} characterizes the behavior of output growth volatility. 
\begin{prop}[Volatility of Output Growth Rate]\label{volatilitymulti}
    The volatility of sectoral output growth satisfies
    \begin{align}\nonumber
        \Var(\Delta\log Y_{t}^r)&=\sum_j(1+\alpha\rho\mathcal{U}_j^r)^2{\xi_{j}^r}^2\Var(\eta_{jt})
    \end{align}
\end{prop}
\begin{proof}
See Appendix \ref{proofs}.
\end{proof}
Proposition \ref{volatilitymulti} generalizes Proposition \ref{corvolatility}, introducing diversification forces. The result provides one important insight: when demand shocks are uncorrelated across consumption destinations, the variance of sectoral output growth is a weighted average of the variances of demand growth rates, with weights given by the squared destination shares. These volatilities are magnified by upstream amplification, governed by $\alpha\rho\mathcal{U}_j^r$. To understand the behavior of the economy, consider the inventory-less benchmark where $\alpha=0$ and suppose that all demand shocks have the same variance $\sigma^2$. In this case, the volatility of output $\Var(\Delta\log Y_{t}^r)=\sigma^2\sum_j{\xi_{j}^r}^2$ is smaller than that of demand, $\sigma^2$, since $\sum_j{\xi_j^r}^2\leq1$, strictly if the industry sells to more than one destination. Recall that, from Proposition \ref{corvolatility}, absent inventories in a single-shock economy, the volatilities of output and demand are identical. Exposure to many destinations provides a diversification force that allows output to fluctuate less than final demand. However, if the economy features inventories $(\alpha>0)$, upstream amplification can dominate, and more upstream sectors can be more volatile. Differently from the single-shock case in Proposition \ref{corvolatility}, the presence of inventories is not a sufficient condition to generate upstream amplification since diversification forces might dominate.

In general, comparing industries at different points of the supply chain yields inconclusive results as they may differ in their position with respect to each destination $\mathcal U^r_j$ and in the weight of that destination $\xi^r_j$. To clarify the separate roles of supply chain position and diversification forces, I study them one at a time in the following comparative statics results.
\paragraph{Comparative Statics.} First, I consider the role of industries' position in supply chains in Proposition \ref{corvolatilitymulti}, then isolate the effect of diversification forces in Proposition \ref{corvolatilitymulti2}. 
\begin{prop}[Output Growth Volatility and Supply Chains Position]\label{corvolatilitymulti}
Suppose sector $r$ and sector $s$ are downstream symmetric. If $\mathcal{U}_{j}^r\geq \mathcal{U}_{j}^s,\,\forall j$, and $\exists j:\: \mathcal{U}_{j}^r>\mathcal{U}_{j}^s$ then $\Var(\Delta\log Y_{t}^r)>\Var(\Delta\log Y_{t}^s)$.
\end{prop}
\begin{proof}
See Appendix \ref{proofs}.
\end{proof}
Proposition \ref{corvolatilitymulti} spells out an intuitive special case. Suppose two sectors have similar network structures, but the composition of demand is such that one industry is more upstream than the other vis \'a vis some destinations (for example, it has a smaller weight in the consumption basket as in Proposition \ref{propdemand}). Then, such an industry has a higher output growth volatility since all destination shock variances amplify at a higher rate. This result highlights how the intuition of the single-shock economy does not carry through. In Proposition \ref{corvolatility}, lower consumption weight is sufficient to generate higher upstreamness and, therefore, higher volatility. Instead, when the economy features many shocks, composition effects are at play. Therefore, the sufficient condition needs to be significantly stronger. Namely, the sector is weakly more upstream from all destinations. Importantly, this result holds under the assumption of downstream symmetry, which implies the two industries' exposure to each destination is identical. Effectively, if the two industries are downstream symmetric, their potential for diversification across space is the same. As a consequence, higher upstreamness is a sufficient condition for higher output volatility.\footnote{Note that in this more general context, it is neither necessary nor sufficient that industry $r$ is more upstream than industry $s$ for it to have higher output growth volatility. As a counterexample, suppose $U^r>U^s$ but the composition of bilateral upstreamness is different and such that $r$ is more upstream relative to destinations with low demand volatility, while $s$ is more upstream relative to destinations with large $\Var(\eta_{jt})$, then it is possible that $\Var(\Delta\log Y_{t}^r)<\Var(\Delta\log Y_{t}^s)$.} 
Proposition \ref{corvolatilitymulti2} provides the opposite thought experiment, highlighting the role of diversification.
\begin{prop}[Output Growth Volatility and Diversification]\label{corvolatilitymulti2}
Suppose all destinations have the same demand volatility such that $\Var(\eta_{jt})=\sigma^2,\,\forall j$. Comparing sectors $r$ and $s$ such that $\mathcal{U}^r=\mathcal{U}^s$,
\begin{enumerate}[label=\alph*)]
\item $\Var(\Delta\log Y_{t}^r)>\Var(\Delta\log Y_{t}^s)$ if and only if $(1+\alpha\rho\mathcal{U}^s)\cdot\Xi^s$ second-order stochastically dominates $(1+\alpha\rho\mathcal{U}^r)\cdot\Xi^r$;
\item if further $\mathcal{U}_{j}^r=\mathcal{U}_{j}^s=\mathcal{U},\,\forall j$, then $\Var(\Delta\log Y_{t}^r)>\Var(\Delta\log Y_{t}^s)$ if and only if the distribution $\Xi^s$ second-order stochastically dominates $\Xi^r$.
\end{enumerate}
\end{prop}
\begin{proof}
See Appendix \ref{proofs}.
\end{proof}
This result shows the effect of a more diversified sales exposure by considering the special case where each country has the same demand volatility and studying two equidistant industries from each destination. Since the two industries have the same position relative to each destination market, they face the same degree of upstream amplification. As a consequence, any difference in output growth volatility is driven by different effective demand volatility. Since each destination has the same volatility of demand, differences in variance of the shift-share shocks are generated by differences in the composition of exposure. If industry $r$ has a more dispersed set of destination exposure shares, it faces lower demand volatility thanks to higher diversification across destinations. 

Taken together, the results presented in Propositions \ref{propdemand} and \ref{propfragmentation} establish that absent diversification possibilities, more upstream firms are more volatile, and, as a consequence, economies with more fragmented production structures experience larger output fluctuations. Propositions \ref{corvolatilitymulti} and \ref{corvolatilitymulti2} study the case in which diversification forces are active and show that all else equal, more bilaterally upstream industries are more volatile and that sectors with a more spread out sales distribution experience smaller fluctuations. 

The model provides a rationale for the empirical results in Section \ref{results}, as long as the inventory amplification effect dominates diversification forces.
In the remainder of the section, I test the proposed mechanism directly and provide quantitative
counterfactuals.

\paragraph{Discussion and Extensions.} %The model in this section can qualitatively rationalize the evidence on the cross-sectional distribution of output and inventories elasticities discussed in Section \ref{results}. However, it does so under a set of strong assumptions worth discussing. 
I conclude this section with a brief discussion of the model and its ingredients.

An important assumption in this framework is the linear-quadratic inventory problem, which is key to having a closed-form result. This assumption allows me to embed the inventory problem into the complex structure of the network economy. Absent this assumption, it is hard to find a recursion to solve the model without assuming specific network structures, as in Section \ref{productionline}. The linear-quadratic inventory problem can be seen as an approximation of a dynamic model in which firms face stochastic demand and stock-out risk. In Proposition \ref{propendinventories} in Appendix \ref{quant_model}, I show that, in such a model, firms optimally adjust inventories procyclically. While embedding such a model in a production network would provide interesting dynamics, it would also require restricting the space of possible networks to obtain analytical characterizations.\footnote{Fully-fledged inventory problems imply \textit{sS} policies due to the presence of ordering fixed costs \citep[see][among others]{alessandria_et_al_2010,alessandria2011us,alessandria2023aggregate}. These problems, however, do not have an analytical characterization of the input demand policy. As a consequence, these papers either assume there is no network or that it is roundabout, implying that there is no meaningful notion of network position. Even quantitatively, when the network is infinitely dimensional, due to cycles and self-loops, the lack of recursion hinders any solution.} 

Another important assumption, which partially drives the optimally procyclical nature of the inventory policy, is the absence of production-smoothing motives, which could arise if the firm had a convex cost function. In Proposition \ref{cyclicalityprop} in Appendix \ref{smoothingmotive}, I extend the model to consider this case, which, on its own, would imply optimally countercyclical inventories. When both target-rule and production-smoothing motives are present, if the latter were to dominate, we would have the counterfactual prediction of countercyclical inventories. Finally, as shown in Figure \ref{lowess_graphs}, the linear policy of inventories in sales, implied by the quadratic formulation, fits the data extremely well.\footnote{The linear regression of inventories on sales explains 80\% of the variance with industry fixed effects to absorb permanent differences.} 

In the model, I also assume that the inventory effect is symmetric across sectors and governed by $\alpha$. This assumption allows the characterization of the recursive definition of output and the comparative statics. To mitigate this shortcoming, I do two things: first, I always compare the model predictions with empirical estimates with industry fixed effects so that permanent differences are absorbed. Second,  in Appendix \ref{hetinv_general_app}, I extend the general model to allow for different inventory policies. While I cannot generalize the results with industry-specific $\alpha$s, I can characterize similar comparative statics in special cases of the network (Proposition \ref{hetinvgeneral}).

The model also abstracts from productivity shocks. Suppose otherwise that $I$ firms have a stochastic process for their productivity. As the competitive fringe anchors the price, fluctuations in productivity would be solely reflected in markup changes. As a consequence, firms would adjust their inventories even more procyclically. At high productivity, the firm would find it optimal to produce more. If the productivity shocks are mean-reverting, as productivity increases today, the conditional expectation of productivity tomorrow is lower than the current level. Consequently, the firm optimally increases today's inventories to save on marginal costs. In Proposition \ref{prodshocksprop} in Appendix \ref{productivityshocks}, I analyze this case and show that the procyclical nature of inventories is reinforced.

Lastly, the framework abstracts from modeling trade separately from domestic transactions. This choice is motivated by the observation that introducing trade costs would not change any of the results in this section. The intuition is simple: introducing destination-specific cost shifters, such as iceberg trade costs, changes the steady-state size of each sector but, in this model, does not affect the response of the economy to shocks. 

In summary, this theoretical framework encompasses a richer pattern of propagation of final demand shocks in the network and highlights the key interplay between the features of supply chains and inventories. This allows me to analyze the trends discussed in the introduction through the lens of the model. The higher fragmentation of production over time should, all else equal, increase output growth volatility. On the other hand, the decrease in the concentration of exposure shares should reduce the volatility of demand. Combining these, we should observe that, over time, output elasticities increase. However, whether these trends increase or decrease overall volatility depends on the relative strength of inventory amplification versus diversification. The remainder of the section tests the key mechanism directly and provides quantitative counterfactuals to disentangle the role of longer chains from the one of higher diversification and quantifies the total effect of these trends on economic fluctuations.

\subsection{Testing the Mechanism} \label{testingsec}
The reduced-form estimates in Section \ref{results} show that both the output and inventories elasticities increase upstream. 
The framework in Section \ref{generalnetwork} provides a precise closed-form condition linking output growth to changes in demand through inventories and upstreamness. To further corroborate the reduced-form evidence, this relation can be directly estimated in the data. Moreover, the model-implied condition provides a sharp null on what we should find if inventories were not an important driver of amplification. Proposition \ref{growthmulti} provides the key relation:
\begin{align}
    \Delta\log Y^r_{t}\approx\eta_{t}^r+\alpha\rho\sum_j \mathcal{U}_{j}^r\xi_{j}^r\eta_{j,t}\label{estimatingeqtheory}
\end{align}
The key difficulty in estimating this relation is computing the matrix of $\mathcal U_j^r$, the sufficient statistic accounting for both the structure of the network and the intensity of the inventory channel along supply chains. I show in Appendix \ref{app_model_estimates} that this metric can be directly backed out of the input-output data. I detail in the next section how I calibrate the model and recover estimates of $\alpha$ and $\rho$. I can then estimate the empirical equivalent of eq. (\ref{estimatingeqtheory}):
\begin{align}
    \Delta\log Y^r_{it}=\delta_1\hat\eta^r_{it}+\delta_2\alpha_i^r\hat\upsilon_{it}^r+\epsilon_{it}^r,
\label{estimatingeq}\end{align}
where $\alpha$ is empirical inventory-to-sales ratio, $\hat\eta^r_{it}\equiv \sum_j \xi_{ij}^r\hat\eta_{j,t}$ is the estimated demand shock discussed in Section \ref{methodology}, and $\hat\upsilon_{it}^r\equiv \sum_j \mathcal{U}_{ij}^r\xi_{ij}^r\hat\eta_{j,t}$ is the bilateral upstreamness-weighted shocks. Based on the model prediction, we should estimate $\delta_1=1$ and $\delta_2>0$.

To directly measure the inventory-to-sales ratio, I use the NBER CES Manufacturing.\footnote{As discussed earlier in the paper, WIOD contains information on the changes in the inventory stocks, which are computed as a residual in the I-O table. To recover the inventory-to-sales ratio, I need the level of inventory stock as well as the correct allocation of the industries that use these inventories. For these reasons, I use the NBER CES data, which provides reliable values for the inventory stock for U.S. manufacturing industries.} This data only covers the U.S. and manufacturing industries. Therefore, I restrict the sample of WIOD to manufacturing, and I maintain throughout the assumption that $\alpha_{i}^r=\alpha_{US}^r,\, \forall i$, namely that within an industry, all countries have the same inventory-to-sales ratios.\footnote{Note that if inventories increase in the level of frictions and these decrease with the level of a country's development, then the U.S. is likely to represent a lower bound in terms of inventory-to-sales ratios.}

Before discussing the results, note that if the model was misspecified and inventories played no role, we should expect $\hat\delta_1=1$ and $\hat\delta_2=0$. If inventories smoothed fluctuations upstream, we should have $\hat\delta_2<0$ or $\hat\delta_1<1$ since we would expect output to move less than 1-for-1 with demand. Finally, if the network dissipation role were to dominate, we should also expect $\hat\delta_2<0$ as it would capture differential responses based on the position relative to consumers, as measured by $\mathcal{U}$.

Table \ref{shocks_slope_inventories} reports the estimation results, including industry fixed effects to control for permanent differences and common trends.%\footnote{I report the estimation result including also time fixed effects in Appendix \ref{app_model_estimates}. These are very similar both qualitatively and quantitatively.} I use time 0 versions of inventories and I-O measures to avoid their contemporaneous response to the shocks.
\begin{table}[htb]
\caption{Model-Consistent Estimation of the Role of Inventories and Upstreamness}
\label{shocks_slope_inventories}
\center
\begin{threeparttable}
{\def\sym#1{\ifmmode^{#1}\else\(^{#1}\)\fi}\begin{tabular}{l*{3}{c}} \hline\hline & \multicolumn{2}{c}{{First Stage}} & Second Stage \\ \midrule
                    &\multicolumn{1}{c}{(1)}&\multicolumn{1}{c}{(2)}&\multicolumn{1}{c}{(3)}\\
                    &\multicolumn{1}{c}{ $ \eta_{it}^r $ }&\multicolumn{1}{c}{ $ \alpha_{i}^r\times \upsilon^r_{it} $ }&\multicolumn{1}{c}{ $ \Delta \ln Y _{it}^r $ }\\
\hline
 $ \tilde \eta_{it}^r $ &       0.455\sym{***}&     0.00267         &                     \\
                    &   (0.00694)         &   (0.00185)         &                     \\
[1em]
 $ \alpha_{i}^r\times \tilde\upsilon^r_{it} $ &      0.0362         &       0.452\sym{***}&                     \\
                    &    (0.0247)         &   (0.00658)         &                     \\
[1em]
 $ \eta_{it}^r $    &                     &                     &       1.009\sym{***}\\
                    &                     &                     &    (0.0461)         \\
[1em]
 $ \alpha_{i}^r\times \upsilon^r_{it} $ &                     &                     &       1.606\sym{***}\\
                    &                     &                     &     (0.169)         \\
\hline
Country-Industry FE &         Yes         &         Yes         &         Yes         \\
F-Stat              &        1720         &        1786         &        2788         \\
N                   &       11957         &       11957         &       11957         \\
 $ R^2 $            &       0.668         &       0.734         &       0.493         \\
\hline\hline
\multicolumn{4}{l}{\footnotesize Standard errors are clustered at the producing industry $\times$ country level.}\\
\multicolumn{4}{l}{\footnotesize \sym{*} \(p<0.10\), \sym{**} \(p<0.05\), \sym{***} \(p<0.01\)}\\
\end{tabular}
}
\begin{tablenotes}[flushleft]
 \item \footnotesize Note: This table shows the results of the regressions in eq. (\ref{estimatingeq}). Columns (1) and (2) report the first stage results. Standard errors are cluster-bootstrapped at the country-industry pair. 
\end{tablenotes}
\end{threeparttable}
\end{table}

Consistently with the model prediction, I find that the elasticity to the shift-share instrument $\hat\delta_1=1$ while the effect of the inventory amplification channel is positive and economically meaningful. To understand the magnitude of this effect, consider, ceteris paribus, increasing all $\mathcal{U}$ from its 25th to its 75th percentile to a sector with average inventories; this implies a 19.5\% increase in the output elasticity. Conversely, raising inventories-to-sales ratios by its interquartile range to industries at the average $\mathcal{U}$ yields a higher elasticity by 16.5\%. Quantitatively, the interplay between inventories and the industry's position in supply chains increases the average output elasticity by about 18\%.

\subsection{Model Performance and Counterfactuals} \label{quantmodelsec}
I conclude by using a calibrated version of the model to study the trends discussed in the introduction. 
In particular, I am interested in separately understanding the role of i) longer and more complex supply chains, ii) more dispersed demand composition, and iii) larger inventories. To do so, I study counterfactuals on the structure of the network and the firms' inventory policy.

\subsubsection{Calibration and Model Performance}

\paragraph{Calibration.} To take the model to the data, I use the actual WIOD input-requirement matrix $\tilde{\mathcal{A}}$ as the I-O matrix in the model. I do so using the data from 2000 and 2014. I simulate 24000 cross-sections of demand shocks for the $J=56$ countries, assuming that they follow AR(1) processes with volatility $\sigma_j$, persistence $\rho$, and average cross-country correlation $\varrho$. For each country in the dataset, I compute $\hat\sigma_j$ as the time series standard deviation of the destination-specific shocks $\eta_{jt}$ and use it directly in the data-generating process. I use $\rho=.7$ as estimated in an AR(1) on the empirical demand shocks. For the parameter $\varrho$, governing the average cross-country correlation, I use indirect inference based on the empirical relation between the shift-share volatility and upstreamness.\footnote{With longer time series, I could estimate the entire variance-covariance matrix of the J-dimensional stochastic process. The covariances are not identified since the time series length is $T<J$. Therefore, I use a single average correlation $\varrho$ across countries as the target while using the estimated country-specific volatilities $\hat\sigma_j$ directly. Online Appendix \ref{quant_model} provides additional details on the procedure.} Underlying this calibration strategy is the idea that when shocks are perfectly correlated across countries, all industries should have the same shift-share shock volatility, independently of their position in the network. If, instead, shocks were iid at the destination level, further upstream industries would have lower shift-share shock volatility as the shifters are diversified away. Finally, for the parameter governing inventories, I use the relative volatility of output growth to that of the shift-share shocks in 2001. Using an inventory-to-sales ratio of 0.18, I can match the empirical relative volatility of 1.26. These moments are reported in Table \ref{model_regression_varrho}.

To disentangle the role of supply chain length from that of diversification, I report results for two versions of the model. The first one features a unique consumption destination as in Section \ref{generalnetwork}, such that, in every period, there is only one change in final demand. In this model, any change in the network can only result in changes in supply chain length, not in demand composition and diversification. The second version instead allows for multiple destinations. In the latter case, I draw $J$ changes in final demand from the joint distribution and apply them to actual final demand data in the I-O table. For both cases, I also report the results for a model without inventories, $\alpha=0$, for comparison.

\paragraph{Model Performance.} First, the model can replicate the reduced form evidence. In particular, estimating the main specification in eq. (\ref{cardinalreg}) on the generated data, I obtain Figure \ref{simulatedcardinal}, which should be compared to Figure \ref{margins_cardinal}.
\begin{figure}[htbp]
\centering
\caption{Model Data Regression}
\label{simulatedcardinal}
\begin{subfigure}{.5\textwidth}
  \centering
  \includegraphics[width=1.05\linewidth]{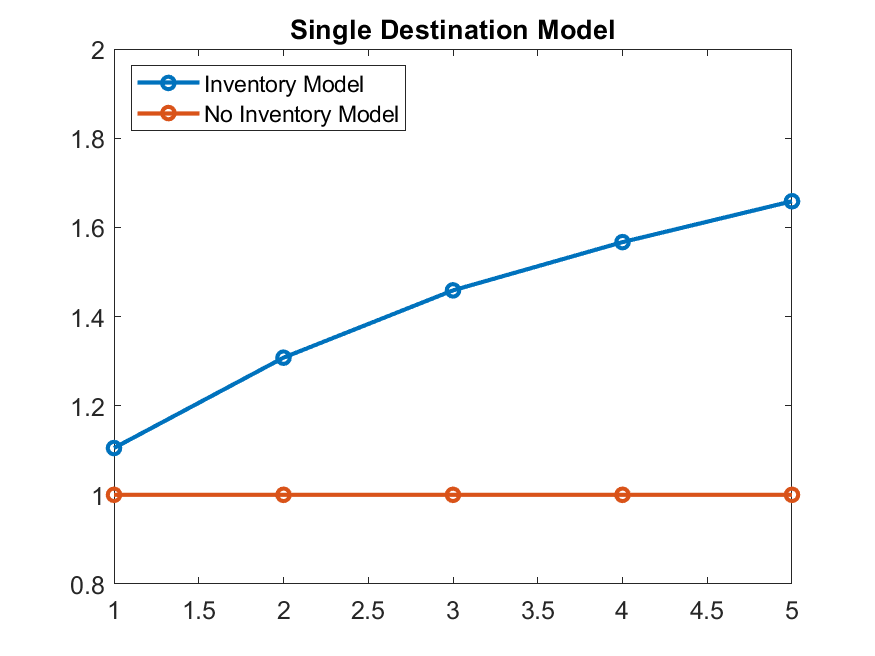}
  \caption{Single Destination}
  \label{fig:modelsub1}
\end{subfigure}%
\begin{subfigure}{.5\textwidth}
  \centering
  \includegraphics[width=1.05\linewidth]{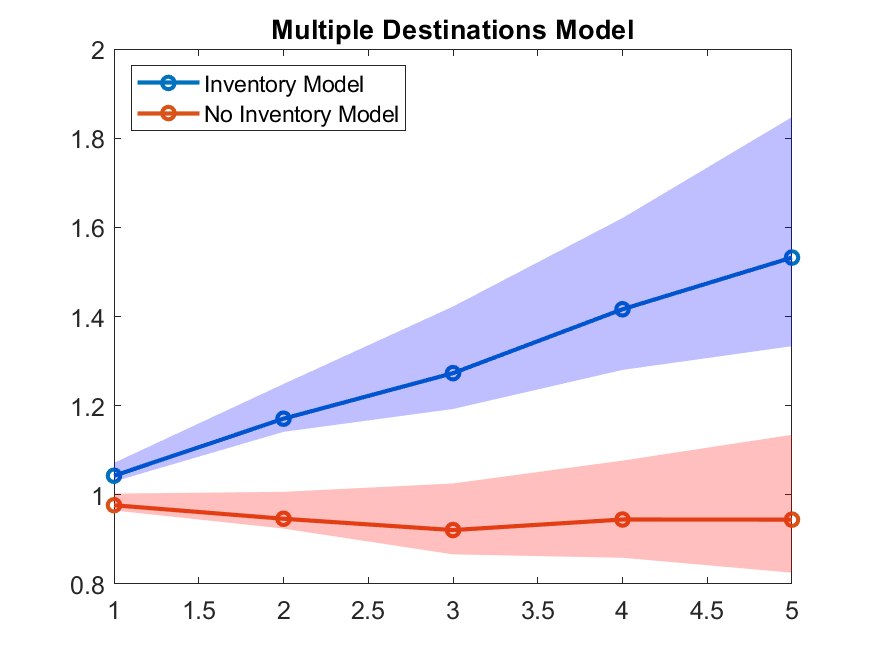}
  \caption{Multiple Destinations}
  \label{fig:modelsub2}
\end{subfigure}
\\\vspace{5pt}
\noindent\justifying
\footnotesize Note: The figures show the model equivalent of Figure \ref{margins_cardinal}. Panel (a) shows the result of regression \ref{cardinalreg} for a model with a single consumer. Panel (b) shows the same estimation for economies with multiple destinations. In the latter case, the propagation pattern is not deterministic as it matters which destination receives which shock. Therefore, I build confidence intervals by simulating the economy 24000 times and the average coefficient as well as plus and minus one standard deviation (shaded area). In both plots, the blue line represents the result in an economy with inventories, while the red line is for economies without inventories (i.e., with $\alpha$ set to 0 for all industries).
\end{figure}
I find that the inventory model can replicate the slope estimated on the data in Section \ref{results}. Namely the increasing response across different upstreamness bins. The model without inventories cannot generate the positive gradient found in Figure \ref{margins_cardinal}. 

As discussed in the previous section, the single destination model has no uncertainty around the effect of a demand shock by upstreamness. On the other hand, allowing for shocks in multiple destinations implies residual uncertainty depending on which country suffers which shock and, given the I-O matrix, which sector is affected. Therefore, this setting has the additional important feature of allowing diversification forces to operate. Quantitatively, the model matches the slope found in the empirical analysis. 

While the model correctly generates the higher output growth elasticity for more upstream industries, it could do so by generating the wrong relationship between upstreamness, demand volatility, and output volatility. As discussed in Section \ref{results}, in the data, demand volatility decreases in upstreamness while output volatility increases. While the relationship between upstreamness and demand volatility is a targeted moment, the relationship between upstreamness and output volatility is untargeted. Importantly, despite stronger diversification of demand shocks, the model shows that inventory dynamics dominate, leading to higher upstream volatility. This suggests that, quantitatively, the interplay between industry position and inventories is sufficient to generate the rising output volatility for more upstream sectors. I report these results in Figure \ref{volatilitymodeldata}.
\begin{figure}[htb]
\centering
\caption{Volatility and Upstreamness - Data vs. Model}
\label{volatilitymodeldata}
\begin{subfigure}{.5\textwidth}
  \centering
  \includegraphics[width=1\linewidth]{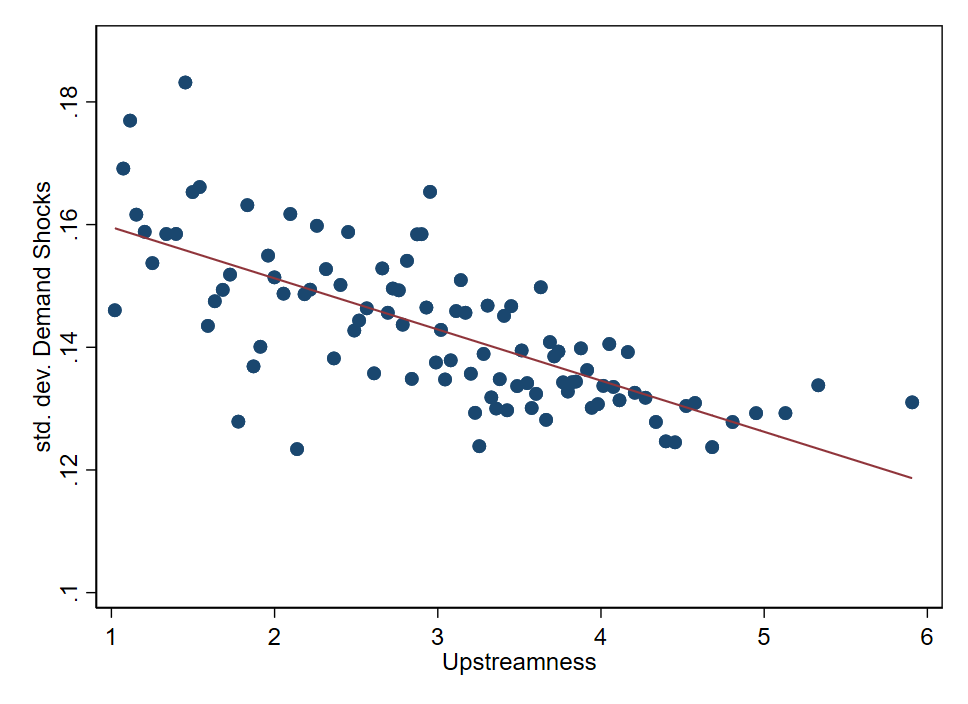}
  \caption{Demand Shock Volatility - Data}
  \label{fig:volat1}
\end{subfigure}%
\begin{subfigure}{.5\textwidth}
  \centering
  \includegraphics[width=1\linewidth]{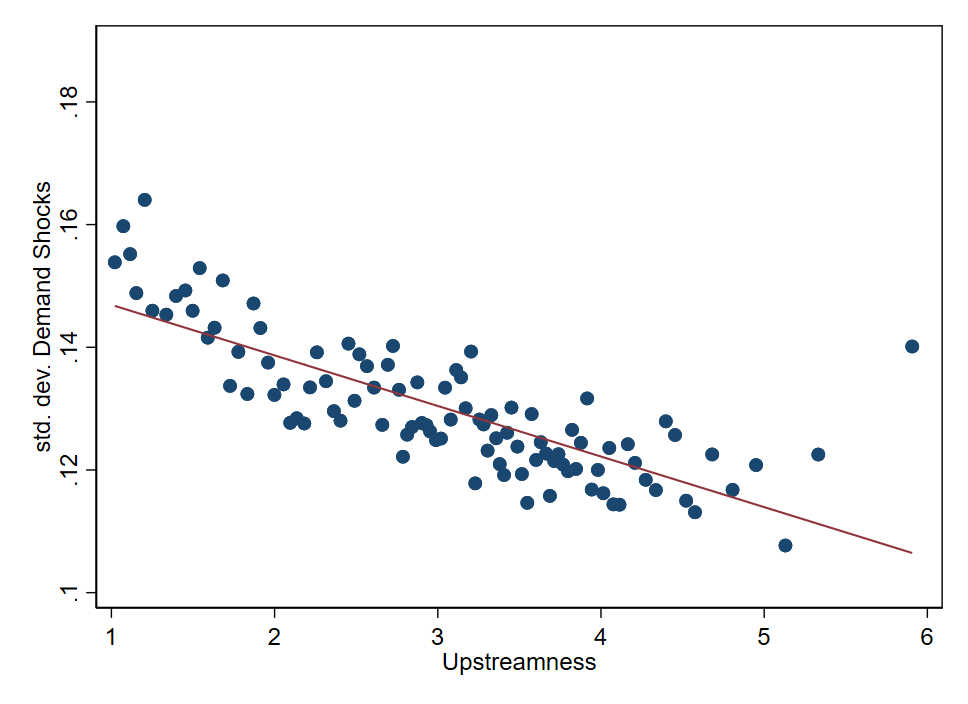}
  \caption{Demand Shock Volatility - Model}
  \label{fig:volat2}
\end{subfigure}
\\\vspace{5pt}
\begin{subfigure}{.5\textwidth}
  \centering
  \includegraphics[width=1\linewidth]{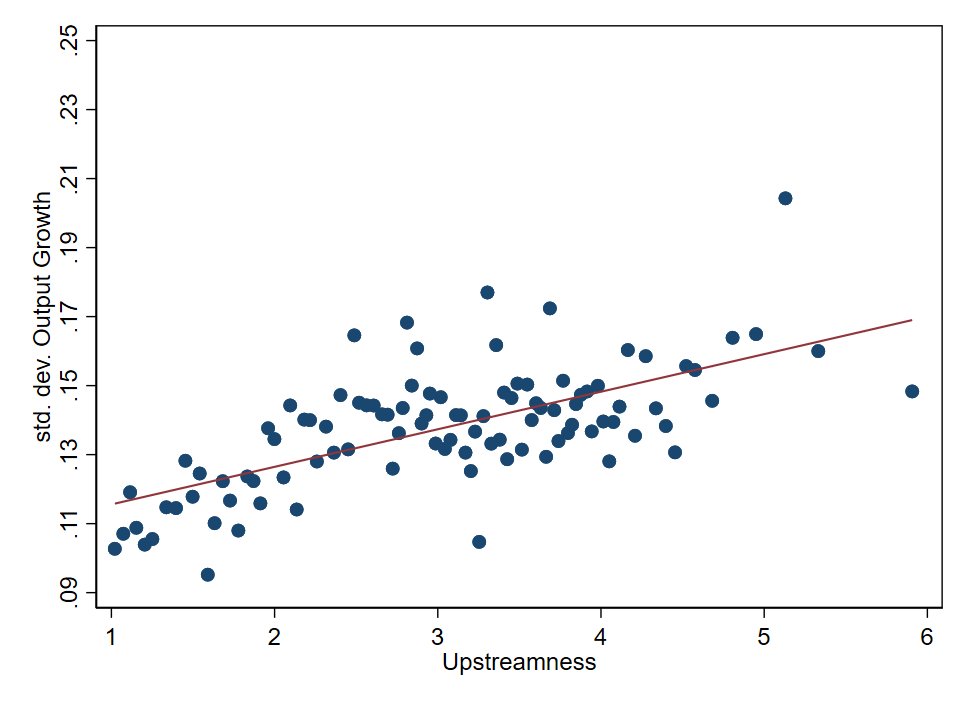}
  \caption{Output Growth Volatility - Data}
  \label{fig:volat3}
\end{subfigure}%
\begin{subfigure}{.5\textwidth}
  \centering
  \includegraphics[width=1\linewidth]{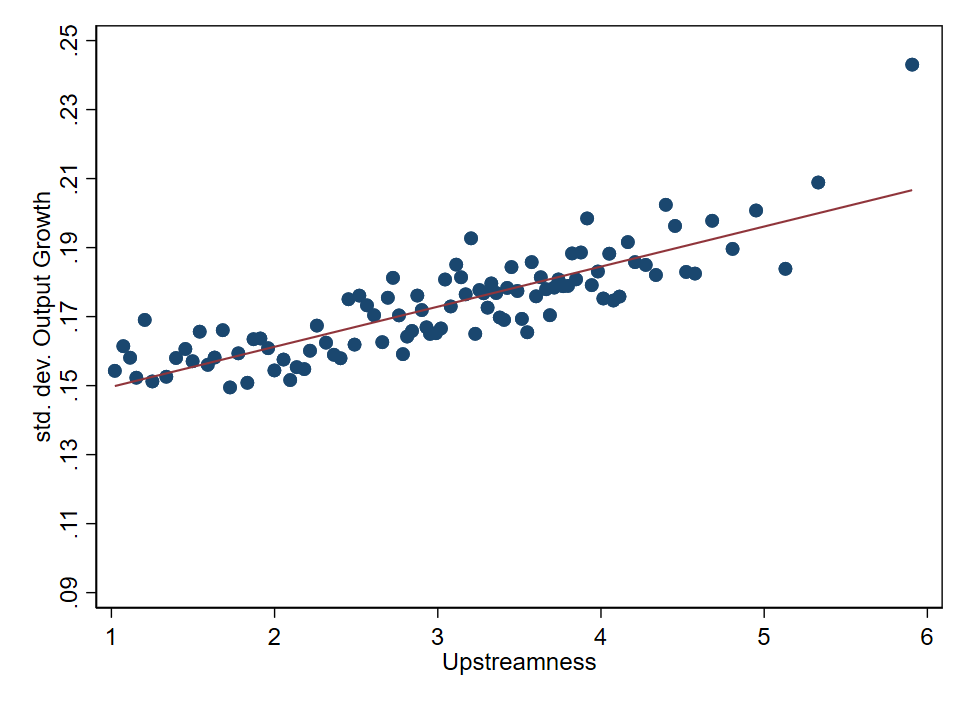}
  \caption{Output Growth Volatility - Model}
  \label{fig:volat4}
\end{subfigure}
\noindent\justifying
\footnotesize Note: The figures show the relations between demand shocks volatility $\sigma(\eta_{it}^r)= \left( \sum_t \left(\eta^r_{it}-\frac{1}{T}\sum_t \eta^r_{it}\right)^2    \right)^\frac{1}{2}$ and Upstreamness and between output growth volatility and Upstreamness. The left panels show the relation in the data, while the right panels show the estimates on the model-generated data.
\end{figure}
 The model can, therefore, rationalize both qualitatively and quantitatively the observations that i) more upstream firms, all else equal, face a lower volatility of demand but also ii) display a higher volatility of output growth. 

\subsubsection{Counterfactuals}

I use the model to study three distinct counterfactuals related to the trends discussed in the introduction. The goal is to disentangle the roles of i) increasing supply chain length, ii) stronger demand diversification, iii) rising inventories, and iv) the interplay between these trends in generating output volatility. These results speak directly to the policy debate on supply chain fragility and the potential need for reshoring. 

To isolate the effect of longer supply chains, I compare the model moments under the network in 2000 and the one in 2014. As discussed earlier, these networks have two salient differences: i) the concentration of sales shares decreased; ii) the average distance from consumers increased. Consequently, one should expect that, fixing the variance of destination-specific shocks $\eta_j$, lower sales share concentration implies more diversification and, therefore, lower variance of demand shocks $\eta^r_i$ that each industry faces. Secondly, the higher distance from consumption should reinforce the inventory amplification channel, thereby increasing the relative volatility of output to demand. I exploit the structure of the model to separate these effects. In a single destination model, there is no notion of heterogeneous demand exposure; hence, in such a context, changing the network from the 2000 to the 2014 I-O matrix only increases the length of production chains. Instead, the same counterfactual in the multiple destination model has both forces operating simultaneously. Comparing the results under these two settings allows me to separate the effect of longer chains from that of higher diversification.

The second counterfactual, motivated by the recent trends in the inventory-to-sales shown by \cite{carrerasvalle}, consists of a 25\% increase in inventories, holding fixed supply chain structure. Intuitively, this should increase the output response for a given change in demand. 

Finally, I combine these experiments and allow for both the 25\% increase in the inventory-to-sales ratio and a change in the global input-output network from the 2000 to the 2014 I-O table. This should bring about two opposing forces: i) the higher diversification reduces the volatility of the final demand each sector is exposed to; ii) for a given level of demand volatility, higher inventories and longer supply chains increase output responses.

I report the results of these counterfactual exercises in Table \ref{counterfacualtable}. The first two columns in the baseline section report targeted moments, which are calibrated to the multiple destination model. The statistics of interest are reported in the counterfactual columns.

The Table reports three key moments of the model economy, which I compute as follows. The standard deviation of demand $\sigma_\eta=\left(\sum_{t} (\eta^r_{it}-\bar\eta^r_i)^2\right)^{\frac{1}{2}}$ with $\eta^r_{it}=\sum_j\xi^r_{ij}\eta_{jt}$, where $r$ is the industry, $i$ the origin country, $j$ a destination country and $\eta_j=\Delta \log D_{jt}$. Output dispersion $\sigma_y$ is computed the same way on the growth rate of output. Finally, $\frac{\Delta \log Y^r_{it}}{\Delta \log \eta^r_{it}}$ is the ratio between the output growth of industry $i$ and the change in final demand industry $i$ is exposed to. For each simulation, I compute the dispersion measures $\sigma_\eta$ and $\sigma_y$, and the median elasticity $\frac{\Delta \log Y^r_{it}}{\Delta \log \eta^r_{it}}$. I report the average across simulations.\footnote{I take the median of $\frac{\Delta \log Y^r_{it}}{\Delta \log \eta^r_{it}}$ rather than the average because, as the denominator is at times very close to zero, the ratio can take extreme values and therefore significantly affect the average response.}

\begin{table}[htb]
    \caption{Counterfactual Moments}
    \label{counterfacualtable}
    \centering
    \begin{threeparttable}
    \begin{tabular}{l|ccc|ccc}
    \toprule
    & \multicolumn{3}{c}{Baseline}& \multicolumn{3}{c}{Counterfactual (\%)}\\[5pt]
            &$\sigma_\eta$ & $\sigma_y$   &  $\frac{\Delta \log Y^r_{it}}{\Delta \log \eta^r_{it}} $&$\sigma_\eta$ & $\sigma_y$   &  $\frac{\Delta \log Y^r_{it}}{\Delta \log \eta^r_{it}}$ \\[5pt]
         \midrule\\[-10pt]
        $\left[A\right]$ - Single Destination Model &&&&&&\\[5pt]
         
   % $\tilde {\mathcal{A}}_{2000}\rightarrow \tilde {\mathcal{A}}_{2014}$ &{0.169} & 0.224 &1.33&0.169&0.228&1.36\\
        $\tilde {\mathcal{A}}_{2000}\rightarrow \tilde {\mathcal{A}}_{2014}$ &{0.169} & 0.224 &1.33&0&+1.8\%&+2.3\%\\
     
%   $\alpha\rightarrow 1.25\alpha$ &{0.169} & 0.224 &1.33&0.169&0.236&1.41\\
   $\alpha\rightarrow 1.25\alpha$ &{0.169} & 0.224 &1.33&0&+5.4\%&+6.0\%\\
   
%    $\tilde {\mathcal{A}}_{2000}\rightarrow \tilde {\mathcal{A}}_{2014}$, $\alpha\rightarrow 1.25\alpha$ &{0.169} & 0.224 &1.33&0.169&0.241&1.44\\
    $\tilde {\mathcal{A}}_{2000}\rightarrow \tilde {\mathcal{A}}_{2014}$, $\alpha\rightarrow 1.25\alpha$ &{0.169} & 0.224 &1.33&0&+7.6\%&+8.3\%\\
  \midrule \\[-10pt]
  $\left[B\right]$ -         Multiple Destinations Model &&&&&&\\[5pt]
         
%    $\tilde {\mathcal{A}}_{2000}\rightarrow \tilde {\mathcal{A}}_{2014}$ &{0.13} &0.17 &1.31&0.125&0.165&1.34\\
    $\tilde {\mathcal{A}}_{2000}\rightarrow \tilde {\mathcal{A}}_{2014}$ &{0.13} &0.17 &1.31&-3.9\%&-2.9\%&+2.3\%\\
         
%   $\alpha\rightarrow 1.25\alpha$ &{0.13} &0.17 &1.31&0.13&0.178&1.37\\
   $\alpha\rightarrow 1.25\alpha$ &{0.13} &0.17 &1.31&0&+4.7\%&+4.6\%\\
   
   % $\tilde {\mathcal{A}}_{2000}\rightarrow \tilde {\mathcal{A}}_{2014}$, $\alpha\rightarrow 1.25\alpha$ &{0.13} &0.17 &1.31&0.125&0.174&1.4\\
     $\tilde {\mathcal{A}}_{2000}\rightarrow \tilde {\mathcal{A}}_{2014}$, $\alpha\rightarrow 1.25\alpha$ &{0.13} &0.17 &1.31&-3.9\%&+2.4\%&+6.9\%\\
    \bottomrule
    \end{tabular}

\begin{tablenotes}[flushleft]
 \item \footnotesize Note: The Table presents the results of baseline and counterfactual estimation. The first 3 columns refer to the baseline model calibrated to 2000, while the last 3 show the counterfactual results. Panel [A] shows the results for the single destination setting while Panel [B] for the multiple destination model. In each model, I perform 3 counterfactuals: i) keeping inventories constant, I use the I-O matrix of 2014 instead of the one of 2000; ii) keeping the I-O matrix constant, I increase inventories by 25\%; iii) 25\% increase of inventories and changing the I-O matrix from the one in 2000 to the one of 2014. Each counterfactual is simulated 4800 times.
\end{tablenotes}
\end{threeparttable}
\end{table}
To isolate the effect of supply chain length, consider the single destination economy. In this setting, there is no scope for diversification forces as there is only one demand shock. As a consequence, shifting the production network from the 2000 I-O matrix to the 2014 one implies changes in upstreamness, which are not coupled with changes in the ability to diversify away the shocks. The increase in the length of the average supply chain generates a higher output volatility holding fixed the volatility of demand, as shown in the first row of Table \ref{counterfacualtable}. Output growth volatility increases by 1.8\% relative to baseline. This change is generated solely by the changes in industry position relative to consumers. As a consequence, the median output growth response to changes in demand increases by 2.3\% from 1.33 to 1.36. 

To benchmark this increase, the second counterfactual shows that increasing the inventory-to-sales ratio by 25\% generates a similar increase in the change in output triggered by a change in demand from 1.33 to 1.41 and a 5.4\% increase in output growth volatility. Combining these two changes, the model predicts a reinforcing effect of the two forces as increasing chain length and inventories are complementary in generating upstream amplification. Consequently, the output response to changes in demand moves from 1.33 to 1.44, and a 7.6\% increase in output growth volatility. Note that, by construction, when changing the network in the single destination economy, we only account for the role of increased supply chain lengths, not for the increased diversification. 

To account for the increased diversification that accompanies the reshaping of the network, I turn to the multiple-destination model. The first observation is that when moving from 2000 to the 2014 network, the model predicts a decline in industries' effective demand shock volatility. As the destination exposure becomes less concentrated, for a given level of volatility of destination shocks $\eta_{jt}$, the volatility of $\eta_r^i$ declines by 3.9\% from 0.13 to 0.125. Absent other changes, this should translate directly into a 3.9\% decline in output growth volatility. However, the network changes between 2000 and 2014 also increase supply chain length so that the output elasticity increases from 1.31 to 1.34 and output growth volatility only drops by 2.9\%. Intuitively, while reshaping the network induces a higher diversification of demand shocks, whose volatility declines, it also increases the average distance from consumers and, therefore, upstream amplification by 2.3\%. The latter effect mitigates the volatility decline from higher diversification. In the second counterfactual, increasing inventories while fixing the network structure implies no change in demand exposure and a significant increase in output volatility as the output elasticity increases by 4.7\% from 1.31 to 1.37. As a consequence, output growth volatility increases by 4.6\%. Finally, the last counterfactual, allowing for both changes in the network and increasing inventories, suggests that the reduction in effective demand volatility is more than offset by the increases in supply chain length and inventories so that the volatility of output growth rises by 2.4\%. This is driven by a significant increase in the output change triggered by a change in demand of 6.9\%, from 1.31 to 1.4, driven by the rise in inventories and the longer supply chains.

These counterfactual experiments suggest that the reshaping of the network is generating opposing forces in terms of output volatility. First, the decreasing exposure to a specific destination reduces the effective volatility of final demand for each industry. At the same time, the increase in chain length would imply a higher responsiveness of output to changes in final demand. When combining these network changes with an increase in inventories from an inventory-to-sales ratio of 18\% to 22.5\%, the benefits of the changes in the network are fully undone. In particular, output volatility rises when inventories are allowed to increase. 

% To conclude, while bearing in mind that the network in this economy is efficient, we can still interpret these counterfactuals as induced by policy interventions.\footnote{Recall that the distribution of resources across sectors is driven by relative prices. In this economy, these are undistorted and induce the optimal allocation implied by the production function and preferences. The only policy that could improve the allocation is within sectors: a planner could subsidize $I$ firms to induce marginal cost pricing.} These results suggest that policy proposals aimed at shortening supply chains might be able to partially curb the propagation of demand shocks, provided that they do not reduce the dispersion of final demand. These counterfactuals also highlight that interventions targeted at reducing disruptions by increasing inventory buffers may come at the cost of higher volatility and that this effect is stronger when supply chains are long and complex.
\section{Conclusions}\label{conclusions}

Recent decades have been characterized by a significant change in how goods are produced due to the rise of global value chains. In this paper, I ask whether these trends trigger stronger or weaker propagation of final demand shocks. To answer this question, I start by asking whether we observe a higher output response to demand shocks by firms further away from consumption. Using a shift-share instrument based on global Input-Output data, I find that upstream firms respond up to twice as strongly as their downstream counterparts to the same final demand shock. I also find that this behavior also applies to inventories.

I build a theoretical framework embedding procyclical inventories in a network model to study the key features determining upstream amplification vs. dissipation patterns. I then estimate the model, and, in counterfactual exercises, I find that in the absence of the inventory amplification channel, we would observe significantly lower output responses to demand shocks. This last result becomes particularly salient in light of the recent trends of increasing inventories and lengthening production chains.

%This paper represents a first attempt at studying the interactions between the rise of global supply chains and the role of inventories in propagating shocks. As such, it ignores several elements. Two examples are the role of re-pricing as an absorption mechanism and the dependence of inventory policies on supply chain positions. This topic represents a promising avenue for both empirical and theoretical research, given the recent supply chain disruptions in the Covid-19 crisis.

\center{\addcontentsline{toc}{section}{ References} }
\bibliographystyle{aer}
{\footnotesize\bibliography{tex_spring25/bibliography}}
\begin{appendices}
\appendix
\newpage

\renewcommand{\thesection}{A.\arabic{section}}
\renewcommand{\thefigure}{A.\arabic{figure}}
\renewcommand{\thetable}{A.\arabic{table}}

\setcounter{figure}{0}
\setcounter{table}{0}
\setcounter{prop}{0}
{
\begin{center}

{\Large {Appendix} }
\end{center}}
\raggedright

%\section{Results}

\begin{table}[H]
\caption{Effect of Demand Shocks on Output Growth by Upstreamness Level}
\label{shock_iomatrix}
\center\begin{threeparttable}
\scriptsize {\def\sym#1{\ifmmode^{#1}\else\(^{#1}\)\fi}\begin{tabular}{l*{6}{c}} \hline\hline & \multicolumn{5}{c}{{First Stage}} & Second Stage \\ \midrule
                    &\multicolumn{1}{c}{(1)}&\multicolumn{1}{c}{(2)}&\multicolumn{1}{c}{(3)}&\multicolumn{1}{c}{(4)}&\multicolumn{1}{c}{(5)}&\multicolumn{1}{c}{(6)}\\
                    &\multicolumn{1}{c}{$ \eta^r_{it}\times  U \in $ [1,2]}&\multicolumn{1}{c}{$ \eta^r_{it}\times  U \in $ [2,3]}&\multicolumn{1}{c}{$ \eta^r_{it}\times  U \in $ [3,4]}&\multicolumn{1}{c}{$ \eta^r_{it}\times  U \in $ [4,5]}&\multicolumn{1}{c}{$ \eta^r_{it}\times  U \in $ [5,$ \infty $)}&\multicolumn{1}{c}{$\Delta \ln Y^r_{it}$}\\
\hline
$ \tilde\eta^r_{it}\times  U \in $ [1,2]&       0.453\sym{***}&   -0.000608         &  -0.0000978         &   -0.000119         &   -0.000113         &                     \\
                    &   (0.00212)         &   (0.00207)         &   (0.00208)         &   (0.00129)         &  (0.000738)         &                     \\
[1em]
$ \tilde\eta^r_{it}\times  U \in $ [2,3]&    0.000283         &       0.451\sym{***}&    -0.00222         &   -0.000516         &   -0.000190         &                     \\
                    &   (0.00210)         &   (0.00205)         &   (0.00206)         &   (0.00127)         &  (0.000730)         &                     \\
[1em]
$ \tilde\eta^r_{it}\times  U \in $ [3,4]&   0.0000375         &    0.000714         &       0.476\sym{***}&   -0.000634         &   -0.000354         &                     \\
\end{tabular}
}
\begin{tablenotes}[flushleft]
 \item \footnotesize Note: The Table shows the results of the regression in equation \ref{cardinalreg}. In particular, I regress the growth rate of output on the instrumented demand shocks interacted with dummies taking value 1 if upstreamness is in the $[1,2]$ bin, $[2,3]$ bin, and so on. Observations with upstreamness above 6 are included in the $[5,\infty)$ bin. All regressions include producing industry-country fixed effects. Columns 1 to 5 report the first-stage results for each of the endogenous variables. Column 6 reports the second-stage results. Standard errors are clustered at the producing industry-country level.
\end{tablenotes}
\end{threeparttable}
\end{table}

\begin{table}[H]
\caption{Effect of Demand Shocks on Inventory Changes by Upstreamness Level}
\label{cardinal_inv_table}
\center\begin{threeparttable}
\scriptsize{\def\sym#1{\ifmmode^{#1}\else\(^{#1}\)\fi}\begin{tabular}{l*{6}{c}} \hline\hline & \multicolumn{5}{c}{{First Stage}} & Second Stage \\ \midrule
                    &\multicolumn{1}{c}{(1)}&\multicolumn{1}{c}{(2)}&\multicolumn{1}{c}{(3)}&\multicolumn{1}{c}{(4)}&\multicolumn{1}{c}{(5)}&\multicolumn{1}{c}{(6)}\\
                    &\multicolumn{1}{c}{$ \eta^r_{it}\times  U \in $ [1,2]}&\multicolumn{1}{c}{$ \eta^r_{it}\times  U \in $ [2,3]}&\multicolumn{1}{c}{$ \eta^r_{it}\times  U \in $ [3,4]}&\multicolumn{1}{c}{$ \eta^r_{it}\times  U \in $ [4,5]}&\multicolumn{1}{c}{$ \eta^r_{it}\times  U \in $ [5,$ \infty $)}&\multicolumn{1}{c}{ $\frac{\Delta I^r_{it}}{Y^r_{it}} $ }\\
\hline
$ \tilde\eta^r_{it}\times  U \in $ [1,2]&       0.451\sym{***}&   -0.000540         &  -0.0000634         &   -0.000124         &   -0.000116         &                     \\
                    &   (0.00215)         &   (0.00210)         &   (0.00210)         &   (0.00130)         &  (0.000741)         &                     \\
[1em]
$ \tilde\eta^r_{it}\times  U \in $ [2,3]&    0.000246         &       0.449\sym{***}&    -0.00225         &   -0.000529         &   -0.000193         &                     \\
                    &   (0.00213)         &   (0.00208)         &   (0.00208)         &   (0.00129)         &  (0.000733)         &                     \\
[1em]
$ \tilde\eta^r_{it}\times  U \in $ [3,4]&   0.0000396         &    0.000738         &       0.477\sym{***}&   -0.000634         &   -0.000269         &                     \\
\end{tabular}
}
\begin{tablenotes}[flushleft]
 \item \footnotesize Note: The Table shows the results of the regression in equation \ref{cardinalreg} with inventory changes as the dependent variable. In particular, I regress the inventory changes over output on the instrumented demand shocks interacted with dummies taking value 1 if upstreamness is in the $[1,2]$ bin, $[2,3]$ bin, and so on. Observations with upstreamness above 6 are included in the $[5,\infty)$ bin. All regressions include producing industry-country fixed effects. Columns 1 to 5 report the first-stage results for each of the endogenous variables. Column 6 reports the second-stage results. Standard errors are clustered at the producing industry-country level.
\end{tablenotes}
\end{threeparttable}
\end{table}

% \begin{proof}[Proof of Example \ref{exemp1}]
% Using Equation (\ref{networkbullwhip2}), together with the assumption $\sum_{q\in Q}\tilde{\mathcal{A}}^{vq}=\sum_{q\in Q'}\tilde{\mathcal{A}}^{kq}$ allows to rewrite the effect of the marginal change in positioning on the responsiveness of output as  
% \begin{align*}
%      \Delta_\beta\left(\frac{\partial Y_{k,t}}{\partial D_{t}} \right)= {(1+\alpha\rho)^n}{\sum_q\tilde{\mathcal{A}}^{vq}\sum_p \tilde{\mathcal{A}}^{qp}...\sum_r \tilde{\mathcal{A}}^{mr}\tilde{\mathcal{A}}^{rs}}\left[{(1+\alpha\rho)\sum_v\tilde{\mathcal{A}}^{kv}} -1\right].
% \end{align*}
% The sign of this change is determined by the sign of the bracket for $\rho>0$.
% Since, with positively autocorrelated shocks ($\rho>0$), the first term in the bracket is always weakly larger than one, this equation is negative, implying increasing dissipation along the network, only if ${\sum_v\tilde{\mathcal{A}}^{kv}}<1$. The change is positive, implying amplification (or increasing dissipation), if the outdegree of the node is larger than 1 or if the inventory effect is strong enough to overcome the network dissipation effect.
% \end{proof}

% \begin{proof}[Proof of Example \ref{exemp2}]
% Using Equation (\ref{networkbullwhip2}), implies  
% \begin{align*}
%      \Delta_{\tilde L}\left(\frac{\partial Y_{k,t}}{\partial D_{t}} \right)= {(1+\alpha\rho)^n}\chi^{n-1}_{k}\left[{(1+\alpha\rho)\sum_v\tilde{\mathcal{A}}^{kv}} -1\right].
% \end{align*}
% Which is positive if $(1+\alpha\rho)\sum_v\tilde{\mathcal{A}}^{kv} -1>0$.\\
% \end{proof}

\end{appendices}
\begin{appendices}

\newpage

\setcounter{page}{1}
\setcounter{figure}{0}
\setcounter{table}{0}
\setcounter{prop}{0}
\setcounter{section}{0}
\setcounter{definition}{0}
\setcounter{lemma}{0}
\setcounter{equation}{0}
% \newcounter{myc}
% \setcounter{myc}{0}
% \renewcommand{\thesection}{\arabic{myc}} % use the new counter for sections
% \let\osection\section % take a "snapshot" of the current state of section
%\renewenvironment{section}{\stepcounter{myc}\osection}
\newcounter{mysection}
\makeatletter
\@addtoreset{section}{mysection}
\makeatother

\renewcommand{\thesection}{\Alph{section}}
\renewcommand{\thefigure}{OA.\arabic{figure}}
\renewcommand{\thetable}{OA.\arabic{table}}
\renewcommand{\theprop}{OA.\arabic{prop}}
\renewcommand{\thedefinition}{OA.\arabic{definition}}
\renewcommand{\thelemma}{OA.\arabic{lemma}}
\renewcommand{\theequation}{OA.\arabic{equation}}
\renewcommand{\theappprop}{OA.\arabic{appprop}}

%\setcounter{prop}{0}\renewcommand{\theprop}{B.\arabic{prop}}
%\setcounter{section}{0}
%\section{Additional Material}

\addcontentsline{toc}{section}{Appendix} % Add the appendix text to the document TOC
\part{\Large{Online Appendix for ``Inventories, Demand Shocks Propagation, and Amplification in Supply Chains''}} % Start the appendix part
\setcounter{parttocdepth}{3}
\parttoc % Insert the appendix TOC

\newpage
\justify \section{Motivating Evidence}\label{motivatingevidence}
\normalsize

\justifying\subsubsection*{Production chains have increased in length}
\addcontentsline{toc}{subsection}{Production chains have increased in length}

\paragraph{Total Length of Supply Chains}\phantom{a}
\begin{figure}[H]
\caption{Dynamics of Supply Chains Length}
\label{F1length}\centering
\includegraphics[width=.5\textwidth]{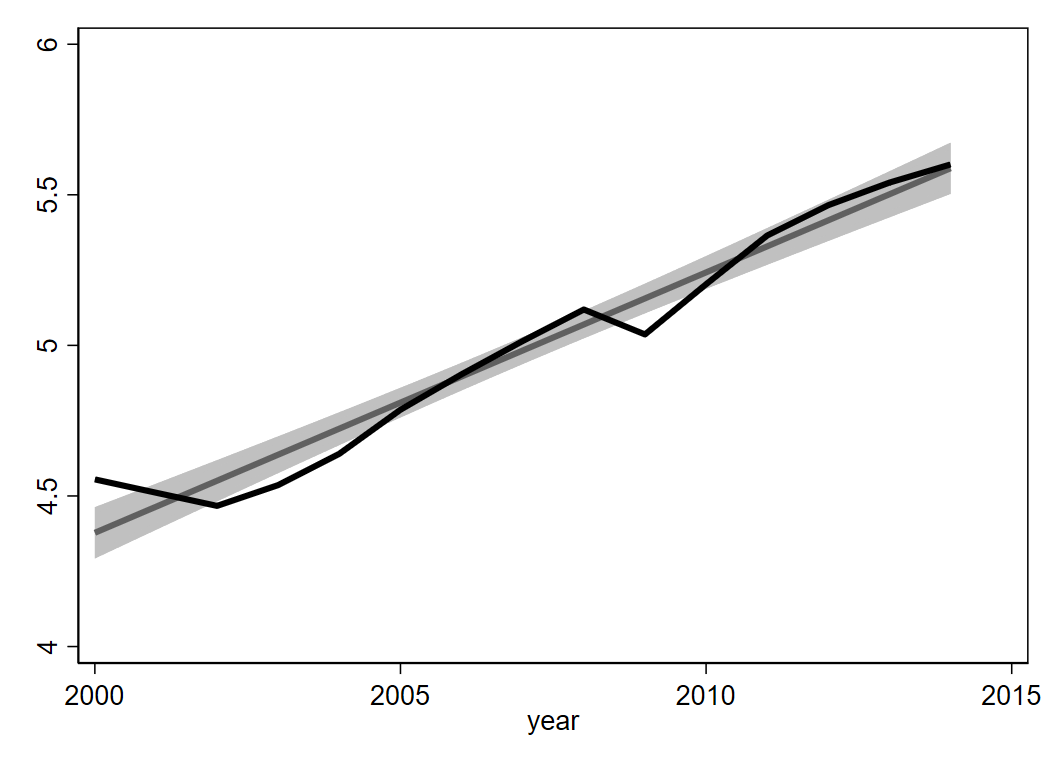} 
\\\vspace{5pt}
\noindent\justifying
 \footnotesize Note: Note: The figure shows the dynamics of the weighted length of chains measure computed as $L_t=\frac{\sum_i\sum_r y_{it}^r L^r_{it}}{\sum_i\sum_r y_{it}^r }$, here $L_{it}^r\coloneqq U_{it}^r+D_{it}^r$, namely the sum of upstreamness and downstreamness to count the total amount of steps embodied in a chain from pure value added to final consumption. The figure shows the average over time and it includes the estimated linear trend and the 95\% confidence interval around the estimate. 
\end{figure}
\justifying\subsubsection*{Sales shares are becoming less concentrated}
\addcontentsline{toc}{subsection}{Fact 2: Sales shares are becoming less concentrated}
\begin{figure}[H]
\caption{Herfindahl Index of Sales Shares}\label{hhi_trends}
\centering\hspace*{-.5cm}\subcaptionbox{Simple Average HHI}
{\includegraphics[width=.5\textwidth]{input/HHI_trend.png}}
\hspace*{.25cm}\subcaptionbox{Weighted Average HHI}
{\includegraphics[width=.5\textwidth]{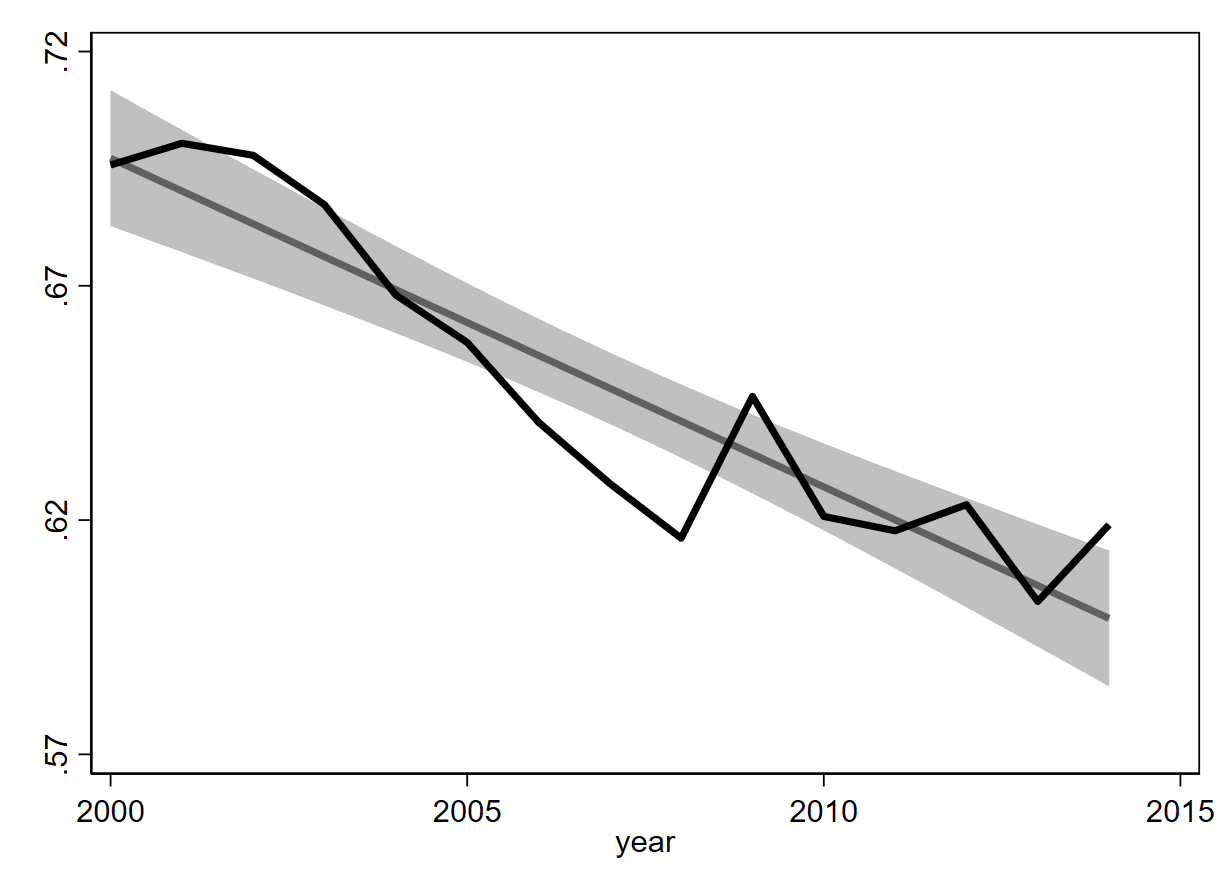}}\\
\noindent\justifying
 \footnotesize Note: The figure shows the behavior of the Herfindahl Index of destination shares over time. Destination shares are described as in equation \ref{shareseq} in Section \ref{methodology}. The Herfindahl Index is computed at the industry level as $HHI^r_t=\sum_j {\xi^r_{ij}}^2$. The left panel shows the simple average across industry, i.e. $HHI_t=R^{-1}\sum_r HHI^r_t$. The right panel shows the weighted average using industry shares as weights: $HHI_t=\sum_r \frac{Y^r_t}{Y_t}HHI^r_t$. The plots include the estimated linear trend and the 95\% confidence interval around the estimate.\end{figure}

\justifying\subsubsection*{Inventories-to-sales ratios are increasing
}
\addcontentsline{toc}{subsection}{Inventories-to-sales ratios are increasing}

\begin{figure}[H]
\caption{Trends in Inventory-to-Sales ratios}\label{fact5plots}

\centering\hspace*{-0.5cm}\subcaptionbox{Yearly NBER}
{\includegraphics[width=.5\textwidth]{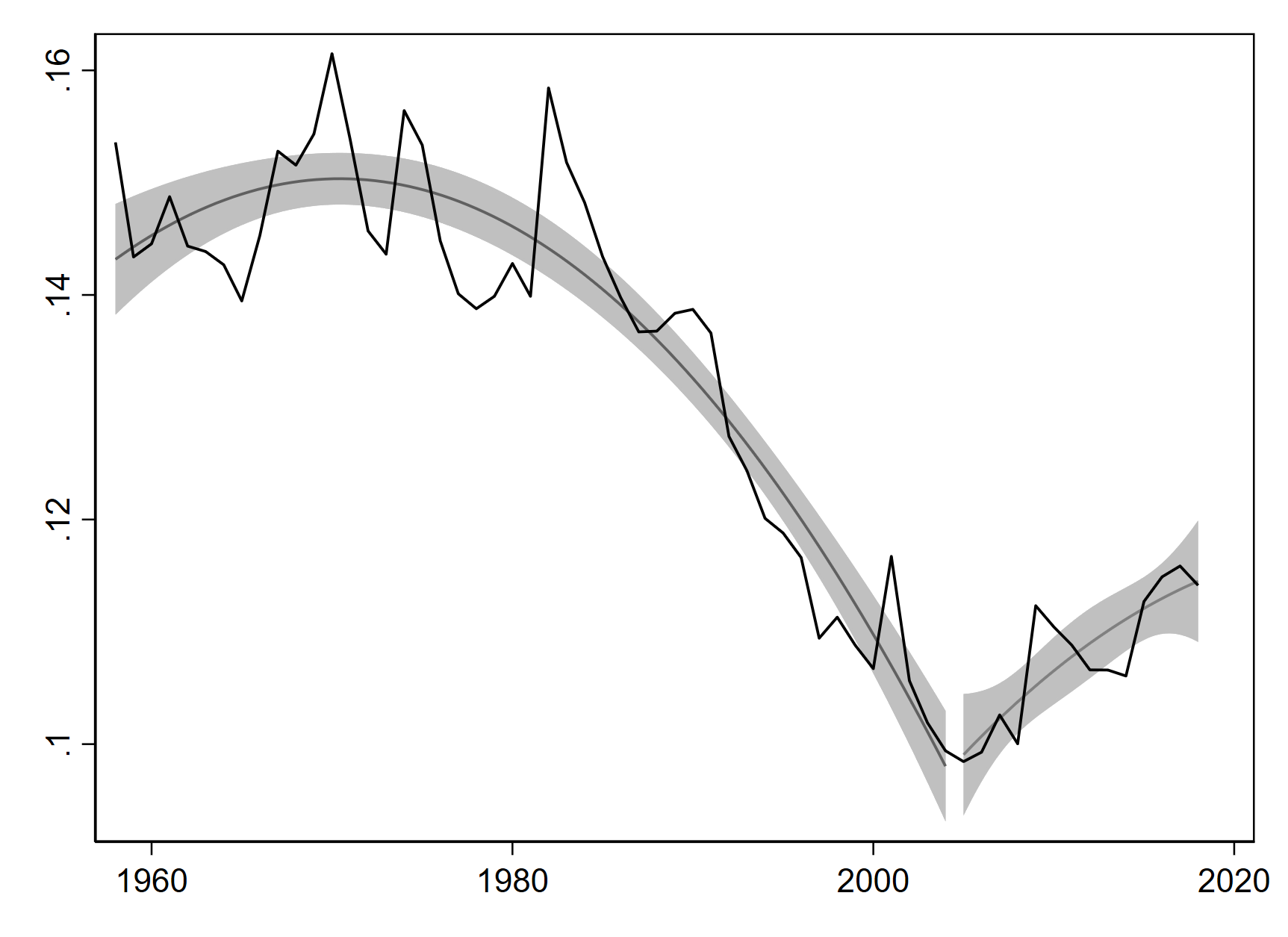}}
\hspace*{.25cm}\subcaptionbox{Monthly Census}
{\includegraphics[width=.5\textwidth]{input/inv_trend_month.png}}\\

\noindent\justifying
 \footnotesize Note: The graphs replicate the key finding in \cite{carrerasvalle}. Panel (a) shows the inventory-to-sales ratio from 1958 to 2018 from the NBER CES Manufacturing Database. Panel (b) reports the same statistic from the Census data from Jan-1992 to Dec-2018. Both graphs include non-linear trends before and after 2005. I estimate separate trends as \cite{carrerasvalle} suggests that 2005 is when the trend reversal occurs.\end{figure}
 \newpage
\section{Data}\label{app:data}
\subsection*{WIOD Coverage}\addcontentsline{toc}{subsection}{WIOD Coverage}\phantom{a}

\begin{figure}[H]
\caption{World Input Output Table }
\label{wiod}\includegraphics[width=1\textwidth]{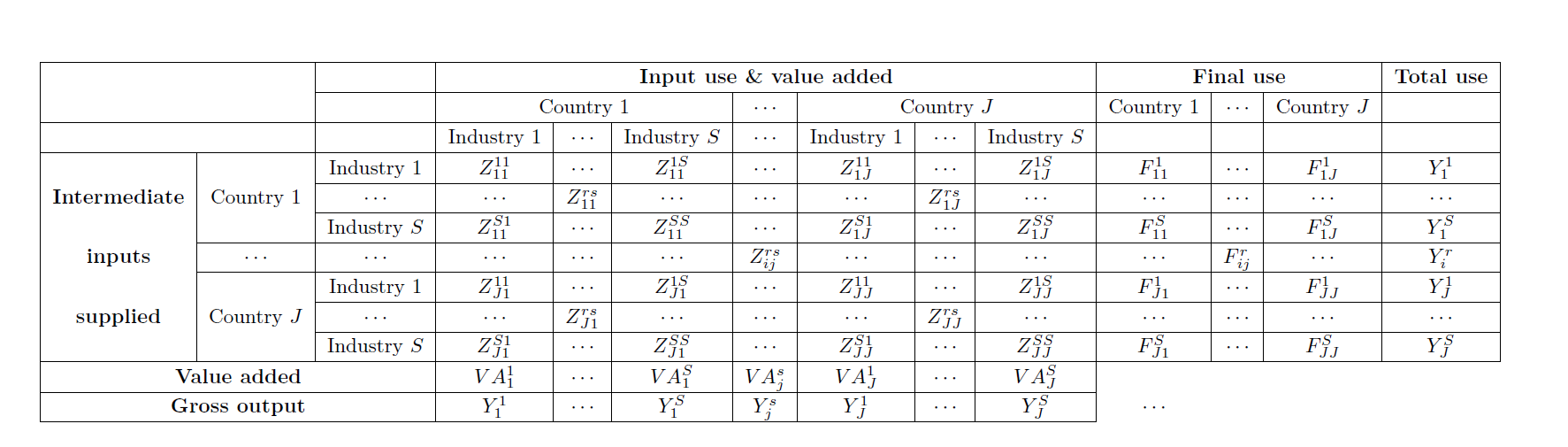} 
\end{figure}
\noindent\justifying
 \footnotesize Note: The figure shows a schematic of the structure of the World Input-Output Database.
 
\begin{table}[htp!]
\caption{Countries}
\label{country}
\center\scalebox{1}{% \begin{tabular}{ll} \hline
% country   \\ \hline
% Australia & Ireland  \\
% Austria & Italy \\
% Belgium & Japan \\
% Bulgaria & Republic of Korea \\
% Brazil & Lithuania \\
% Canada & Luxembourg \\
% Switzerland & Latvia \\
% China & Mexico \\
% Cyprus & Malta \\
% Czech Republic & Netherlands \\
% Germany & Norway \\
% Denmark & Poland \\
% Spain & Portugal \\
% Estonia & Romania \\
% Finland &  Russian Federation  \\
% France & Slovakia\\
% United Kingdom & Slovenia  \\
% Greece & Sweden \\
% Croatia & Turkey \\
% Hungary & Taiwan \\
% Indonesia & United States  \\
% India &  Rest of the World\\
%  \hline
% \end{tabular}

\begin{tabular}{llll} \hline
Country   \\ \hline
Australia &Denmark & Ireland & Poland\\
Austria &Spain & Italy &Portugal\\
Belgium &Estonia & Japan&Romania \\
Bulgaria &Finland & Republic of Korea&Russian Federation  \\
Brazil &France & Lithuania& Slovakia\\
Canada &United Kingdom & Luxembourg&Slovenia \\
Switzerland &Greece & Latvia &Sweden\\
China &Croatia & Mexico &Turkey\\
Cyprus &Hungary & Malta &Taiwan\\
Czech Republic &Indonesia & Netherlands&United States \\
Germany &India & Norway &Rest of the World\\
 \hline
\end{tabular}}
\end{table}

\begin{table}[H]
\caption{Industries}
\label{industry}
\hspace{-1cm}{\scalebox{.7}{\begin{tabular}{ll}
\hline{Industry}&{Industry}\\\hline
Crop and animal production&Wholesale trade\\
Forestry and logging&Retail trade\\
Fishing and aquaculture&Land transport and transport via pipelines\\
Mining and quarrying&Water transport\\
Manufacture of food products&Air transport\\
Manufacture of textiles&Warehousing and support activities for transportation\\
Manufacture of wood and of products of wood and cork&Postal and courier activities\\
Manufacture of paper and paper products&Accommodation and food service activities\\
Printing and reproduction of recorded media&Publishing activities\\
Manufacture of coke and refined petroleum products&Motion picture\\
Manufacture of chemicals and chemical products&Telecommunications\\
Manufacture of basic pharmaceutical products and pharmaceutical preparations&Computer programming\\
Manufacture of rubber and plastic products&Financial service activities\\
Manufacture of other non-metallic mineral products&Insurance\\
Manufacture of basic metals&Activities auxiliary to financial services and insurance activities\\
Manufacture of fabricated metal products&Real estate activities\\
Manufacture of computer&Legal and accounting activities\\
Manufacture of electrical equipment&Architectural and engineering activities\\
Manufacture of machinery and equipment n.e.c.&Scientific research and development\\
Manufacture of motor vehicles&Advertising and market research\\
Manufacture of other transport equipment&Other professiona activities\\
Manufacture of furniture&Administrative and support service activities\\
Repair and installation of machinery and equipment&Public administration and defence\\
Electricity&Education\\
Water collection&Human health and social work activities\\
Sewerage&Other service activities\\
Construction&Activities of households as employers\\
Wholesale and retail trade and repair of motor vehicles and motorcycles&Activities of extraterritorial organizations and bodies\\
\hline\end{tabular}}}
\end{table}

% \section{Test of Uncorrelatedness of Instruments}\label{orthogonalitysection}
% \raggedright As discussed in the main text, the identifying assumption for the validity of the shift-share design is conditional independence of shocks and potential outcomes. Since this assumption cannot be tested, I provide evidence that the shares and the shocks are uncorrelated to alleviate endogeneity concerns. I test the conditional correlation by regressing the shares on future shocks and industry fixed effects. Formally
% \begin{align*}
%     \xi^r_{ijt}=\beta \hat\eta_{jt+1}(i)+\gamma_{it}^r+\epsilon_{ijt}^r.
% \end{align*}
% This estimation results reported in Table \ref{orthogonality} suggest that the two are uncorrelated.
% \begin{table}[H]
% \caption{Test of Uncorrelatedness of Instruments}
% \center
% \label{orthogonality}
% \input{input/orthogonality_test.tex}
% \end{table}

\subsection*{WIOD Inventory Adjustment}\addcontentsline{toc}{subsection}{Inventory Adjustment}\phantom{a}
\normalsize
\cite{antras_et_al} define the measure of upstreamness based on the Input-Output tables. This measure implicitly assumes the contemporaneity between production and use of output. This is often not the case in empirical applications since firms may buy inputs and store them to use them in subsequent periods. This implies that, before computing the upstreamness measure, one has to correct for this possible time mismatch.

The WIOD data provides two categories of use for these instances: net changes in capital and net changes in inventories. These categories are treated like final consumption, meaning that the data reports which country but not which industry within that country absorbs this share of output. 

The WIOD data reports as $Z_{ijt}^{rs}$ the set of inputs used in $t$ by sector $s$ in country $j$ from sector $r$ in country $i$, independently of whether they were bought at $t$ or in previous periods.
\\ Furthermore, output in the WIOD data includes the part that is stored, namely
\begin{align}
Y_{it}^r=\sum_s\sum_j Z_{ijt}^{rs}+\sum_jF_{ijt}^{r}+\sum_j \Delta N_{ijt}^r.
\label{yinv}
\end{align}
As discussed above, the variables reporting net changes in inventories and capital are not broken down by industry, i.e., the data contains $\Delta N_{ijt}^r$, not $\Delta N_{ijt}^{rs}$.

This characteristic of the data poses a set of problems, particularly when computing bilateral upstreamness. First and foremost, including net changes in inventories in the final consumption variables may result in negative final consumption whenever the net change is negative and large. This cannot happen since it would imply that there are negative elements of the $F$ vector when computing
\begin{align*}
U=\hat Y^{-1}[I-\mathcal{A}]^{-2}F.
\end{align*}

However, simply removing the net changes from the $F$ vector implies that the tables are no longer balanced, which is also problematic. By the definition of output in equation \ref{yinv}, it may be the case that the sum of inputs is larger than output.
When this is the case, $\sum_i \sum_r a_{ij}^{rs}>1$, which is a necessary condition for the convergence result, as discussed in the Methodology section.

To solve these problems, I apply the inventory adjustment suggested by \cite{antras_et_al}. It boils down to reducing output by the change of inventories. This procedure, however, assumes inventory use. In particular, as stated above, the data reports $\Delta N_{ijt}^r$ but not $\Delta N_{ijt}^{rs}$. For this reason, the latter is imputed via a proportionality assumption. Namely, if sector $s$ in country $j$ uses half of the output that industry $r$ in country $i$ sells to country $j$ for input usages, then half of the net changes in inventories will be assumed to have been used by industry $s$. Formally:
\begin{align*}
\Delta N_{ijt}^{rs}=\frac{Z_{ijt}^{rs}}{\sum_sZ_{ijt}^{rs}}\Delta N_{ijt}^r.
\end{align*}

Given the inputed vector of $\Delta N_{ijt}^{rs}$, the output of industries is corrected as
\begin{align*}
\tilde Y_{ijt}^{rs}=Y_{ijt}^{rs}-\Delta N_{ijt}^{rs}.
\end{align*}

Finally, whenever necessary, Value Added is also adjusted so that the columns of the I-O tables still sum to the corrected gross output.

These corrections ensure that the necessary conditions for the matrix convergence are always satisfied. I apply these corrections to compute network measures while I use output as reported when used as an outcome.

\subsection*{Descriptive Statistics}{Descriptive Statistics}\addcontentsline{toc}{subsection}{Descriptive Statistics}\phantom{a}
This section provides additional descriptive statistics on the World Input-Output Database (WIOD) data.

\raggedright\paragraph*{Upstreamness}\addcontentsline{toc}{subsubsection}{Upstreamness}\phantom{a}

\centering
\begin{threeparttable}
[H]
 \caption{Highest and Lowest Upstreamness Industries}\label{udescriptive}
\begin{tabular}{lc}
\hline{Industry}&{Upstreamness}\\\hline
Activities of extraterritorial organizations and bodies&1\\
Human health and social work activities&1.14\\
Activities of households as employers&1.16\\
Education&1.22\\
Public administration and defence&1.22\\
Accommodation and food service activities&1.66\\
$\vdots$&$\vdots$\\
Construction&3.96\\
Manufacture of wood and of products of wood and cork, except furniture&4.22\\
Manufacture of fabricated metal products, except machinery and equipment&4.27\\
Manufacture of machinery and equipment n.e.c.&4.28\\
Manufacture of other non-metallic mineral products&4.39\\
Mining and quarrying&4.52\\
Manufacture of basic metals&5.13\\
\hline\end{tabular}
\begin{tablenotes}[flushleft]
\item \footnotesize Note: The table displays the top and bottom of the upstreamness distribution. This is computed by averaging within industry, across country, and time: $U^r=\frac{1}{T}\frac{1}{I}\sum_t^T\sum^I_i U_{it}^r$.
\end{tablenotes}
\end{threeparttable}

% \begin{figure}[H]
% \caption{Within Sector Distribution of Upstreamness}
% \label{upstreamness_within}\centering
% \includegraphics[width=.75\textwidth]{input/U_dist_sectors.png} 
% \\\vspace{5pt}
% \noindent\justifying
%  \footnotesize Note: The graph shows the distribution of upstreamness within each industry across country and time. The label of sectors follows the order of Table \ref{industry}.
% \end{figure}

\justify
\paragraph*{Destination Shares}\addcontentsline{toc}{subsubsection}{Destination Shares}

The distribution of destination shares is computed as described in the methodology section. Table \ref{shares_descriptive} reports the summary statistics of the destination shares for all industries and all periods. Importantly, the distribution is very skewed and dominated by the domestic share. On average 61\% of sales are consumed locally. Importantly, the median export share is .16\% and the $99^{th}$ percentile is 12\%. These statistics suggest that there is limited scope for diversification across destinations.

\begin{table}[htp]
\caption{Destination Shares Summary Statistics}
\label{shares_descriptive}
\hspace{-1cm}{\begin{threeparttable}
\center\scalebox{.9}{{
\def\sym#1{\ifmmode^{#1}\else\(^{#1}\)\fi}
\begin{tabular}{l*{1}{cccccccccc}}
\hline\hline

                    &       count&        mean&          sd&         min&         max&         p25&         p50&         p90&         p95&         p99\\
\hline
destination share     &     1522475&    .0227&     .1029&    3.33e-13&    .9999&    .0004&    .0017&    .026&    .0659&    .7143\\

domestic share     &       34632&    .6146&    .2744&    .0001&    .9999&    .4176&    .6674&    .9442&    .9793&    .9974\\

foreign destination share     &     1487843&    .0089&    .0273&    3.33e-13&     .962&    .0003&    .0016&    .0199&    .0418&    .1224\\
\hline\hline
\end{tabular}
}}
\begin{tablenotes}[flushleft]
\item \footnotesize Note: The table displays the summary statistics of the sales destination shares. Shares equal to 0 and 1 have been excluded. \\The latter have been excluded because they arise whenever an industry has 0 output. No industry has an actual share of 1.
\end{tablenotes}
\end{threeparttable}}
\end{table}

\paragraph{Inventories} \label{invappendix}\addcontentsline{toc}{subsubsection}{Inventories}
In the model presented in this paper, the potential amplification is driven by procyclical inventory adjustment. The WIOD data does not provide industry-specific inventory stock or change, eliminating the possibility of a direct test of the mechanism. 

To provide partial evidence of the behavior of inventories, I use NBER CES Manufacturing Industry and the US Census data. I compute the parameter $\alpha= I_t/\mathbb{E}_t D_{t+1}$  as $\alpha_t=I_t/D_{t+1}$, which produces the distributions shown in Figure \ref{alpha_desc}. \footnote{I drop outliers of the $\alpha$ distribution, specifically observation with an annual $\alpha>3$ and a monthly $\alpha>5$, which are both in the 99th percentile of the distribution. }

\begin{figure}[htb]
\caption{Distribution of $\alpha$}
\label{alpha_desc}
\begin{minipage}[b]{.5\linewidth}
\centering\includegraphics[width=1\textwidth]{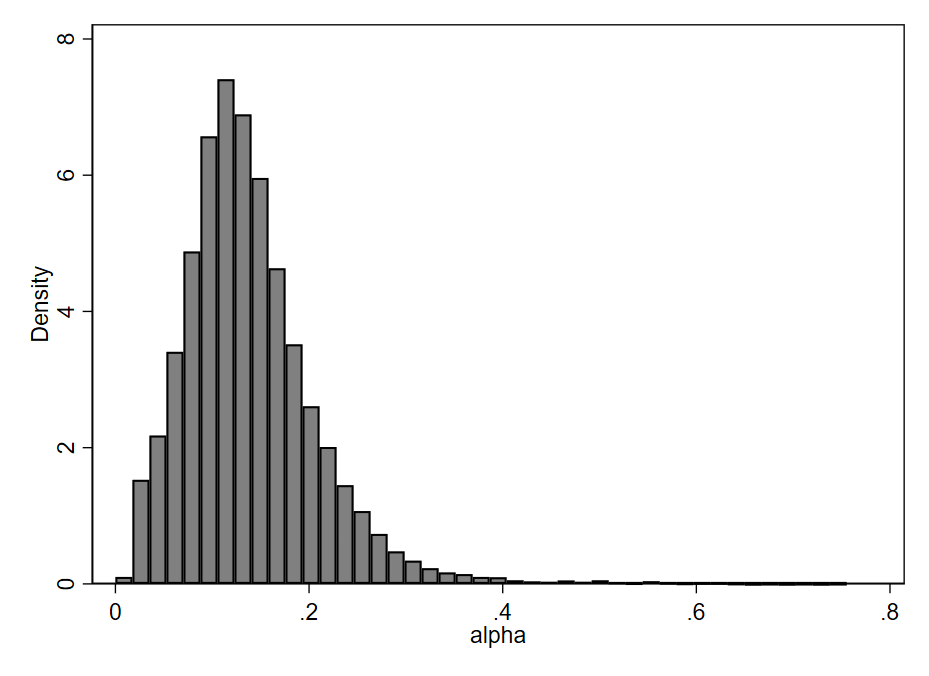}  \subcaption{Distribution of $\alpha$ - NBER CES}
\end{minipage}%
\hspace*{.15cm}\begin{minipage}[b]{.5\linewidth}
\centering\includegraphics[scale=0.26]{{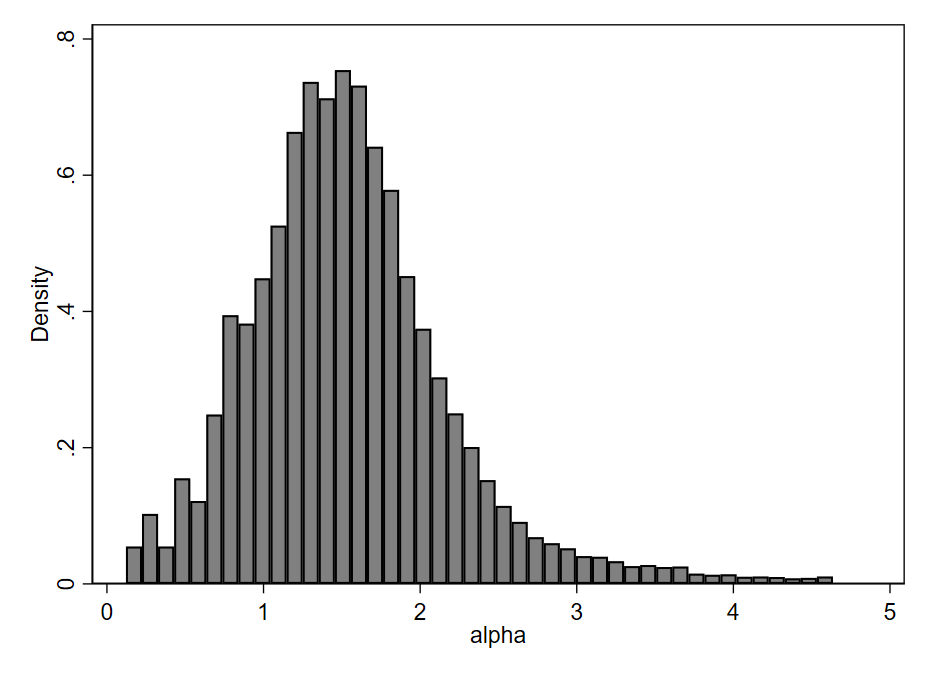}}\subcaption{Distribution of $\alpha$ - US Census} 
\end{minipage}
\\
\noindent\justifying
 \footnotesize Note: The graph shows the distribution of $\alpha_t=I_t/Y_{t+1}$ across the 54 years and 473 industries in the NBER CES Manufacturing Industry data in the left panel and across the Naics 3-digit industries since January 1992 from the US Census.
\end{figure}

\paragraph{Inventories Cyclicality}
\phantom{a}
\begin{figure}[htb]
\caption{Distribution of estimated $I^\prime$}\label{dist_Iprime}
\centering\hspace*{-.5cm}\subcaptionbox{NBER}
{\includegraphics[width=.5\textwidth]{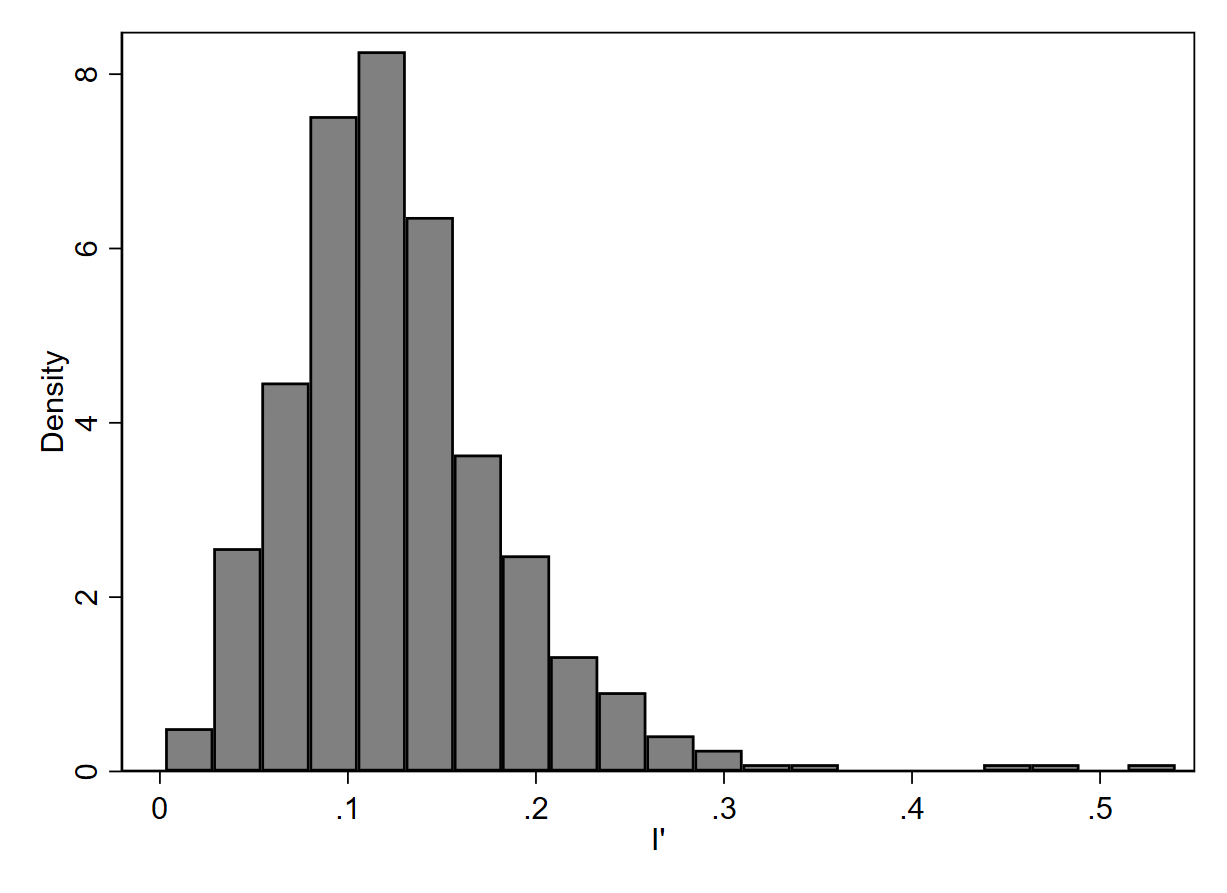}}
\hspace*{.25cm}\subcaptionbox{Census}
{\includegraphics[width=.5\textwidth]{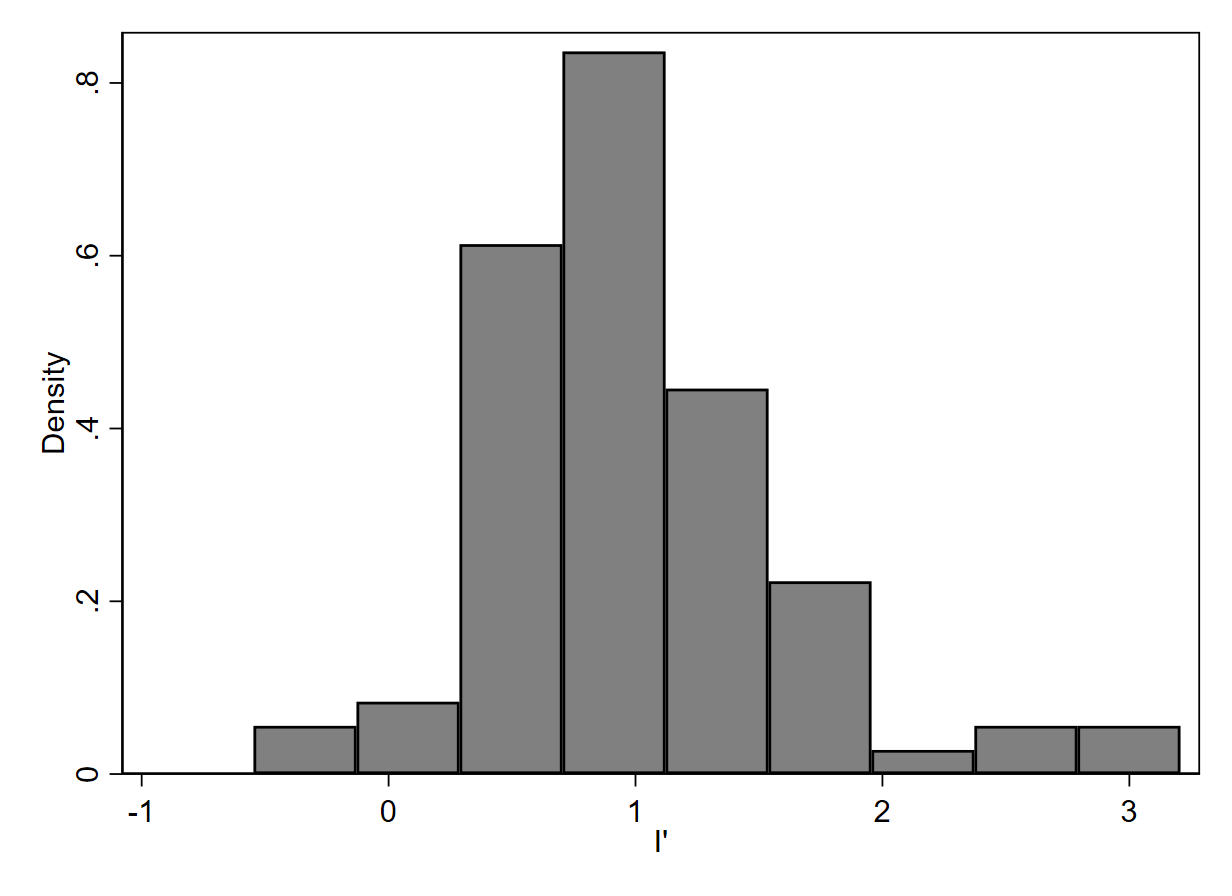}}\\
\noindent\justifying
\footnotesize Note: The graph shows the distribution of estimated $I^\prime(\cdot)$, namely the derivative of the empirical inventory function with respect to sales. The sample is the full NBER CES sample of 473 manufacturing industries. The estimation is carried out sector by sector using time variation. The graph shows the sector-specific estimated coefficient. The left panel shows the same statistics based on the monthly data from the Manufacturing \& Trade Inventories \& Sales data of the US Census.
\end{figure}

\begin{landscape}
\begin{table}[H]
\caption{Estimation of $I'(\cdot)$ }
\label{npreg}
\begin{center}
    
\begin{threeparttable}
{
\def\sym#1{\ifmmode^{#1}\else\(^{#1}\)\fi}
\begin{tabular}{l*{8}{c}}
\hline\hline
                 &\multicolumn{1}{c}{(1)}&\multicolumn{1}{c}{(2)}&\multicolumn{1}{c}{(3)}&\multicolumn{1}{c}{(4)}&\multicolumn{1}{c}{(5)}&\multicolumn{1}{c}{(6)}&\multicolumn{1}{c}{(7)}&\multicolumn{1}{c}{(8)}\\
                    {$y^i_t$} &\multicolumn{1}{c}{$\tilde I^i_t$}&\multicolumn{1}{c}{$\tilde \alpha^i_t$ }&\multicolumn{1}{c}{$\tilde I^i_t$}&\multicolumn{1}{c}{$\tilde \alpha^i_t$}&\multicolumn{1}{c}{$I^i_t$}&\multicolumn{1}{c}{$\alpha^i_t$}&\multicolumn{1}{c}{$I^i_t$}&\multicolumn{1}{c}{$\alpha^i_t$}\\
\toprule
\multicolumn{9}{l}{\textbf{NBER CES Sample - Annual Data}}\\
\hline
 Mean of $y^i_t$   &      -11.36\sym{*}  &   -0.00000270                      &      -11.20\sym{***}&       -0.0000200                &       966.8\sym{***}         & 0.127\sym{***}                    &       982.4\sym{***}&                0.127\sym{***}     \\
                    &     (5.939)         &    (0.000170)                 &     (4.154)         &  (0.000409)                    &                    (33.34)   &   (0.000891)                   &     (16.43)         &               (0.000570)        \\
[1em] \hline
$\frac{\partial\, y_t^i}{\partial\, Sales^i_t}$&      0.0815\sym{***}& -0.00000218\sym{***}&      0.0745\sym{***}& -0.00000251\sym{***}&       0.106\sym{***}         & -0.00000129\sym{***}&       0.101\sym{***}&-0.000000641\sym{***}\\
                    &    (0.0107)         &(0.000000412)         &   (0.00994)         &(0.000000205)         &       (.0013713)              &  (6.53e-08)         &   (0.00555)         &  (9.71e-08)         \\
[1em]

\hline
Industry FE         &         YES         &         YES         &          NO         &          NO         &         YES         &         YES         &          NO         &          NO         \\
\end{tabular}
}
\begin{tablenotes}[flushleft]
 \item \footnotesize Note: This table displays the results of the non-parametric kernel estimations of $I'(\cdot)$ and $\alpha^\prime(\cdot)$, the derivative of the inventory function and the inventory to sales ratio function with respect to current sales. The estimation is based on the data of the NBER CES Manufacturing Industries Dataset for 2000-11 for the top Panels and on the US Census Manufacturing \& Trade Inventories \& Sales data for the bottom Panels. Standard errors are bootstrapped. Variables with $\sim$ denote HP-filtered data. Columns (1), (2), (5), and (6) include industry fixed effects. 
\end{tablenotes}
\end{threeparttable}
\end{center}\end{table}

\vfill
\end{landscape}

\paragraph{Inventory Procyclicality by Inventory Type}\label{inventorytypes} \phantom{a}
\begin{figure}[htb]
\caption{Distribution of estimated $I^\prime$ by inventory type}\label{dist_Iprime_type}
\centering\hspace*{-.5cm}\subcaptionbox{Final Goods Inventories}
{\includegraphics[width=.4\textwidth]{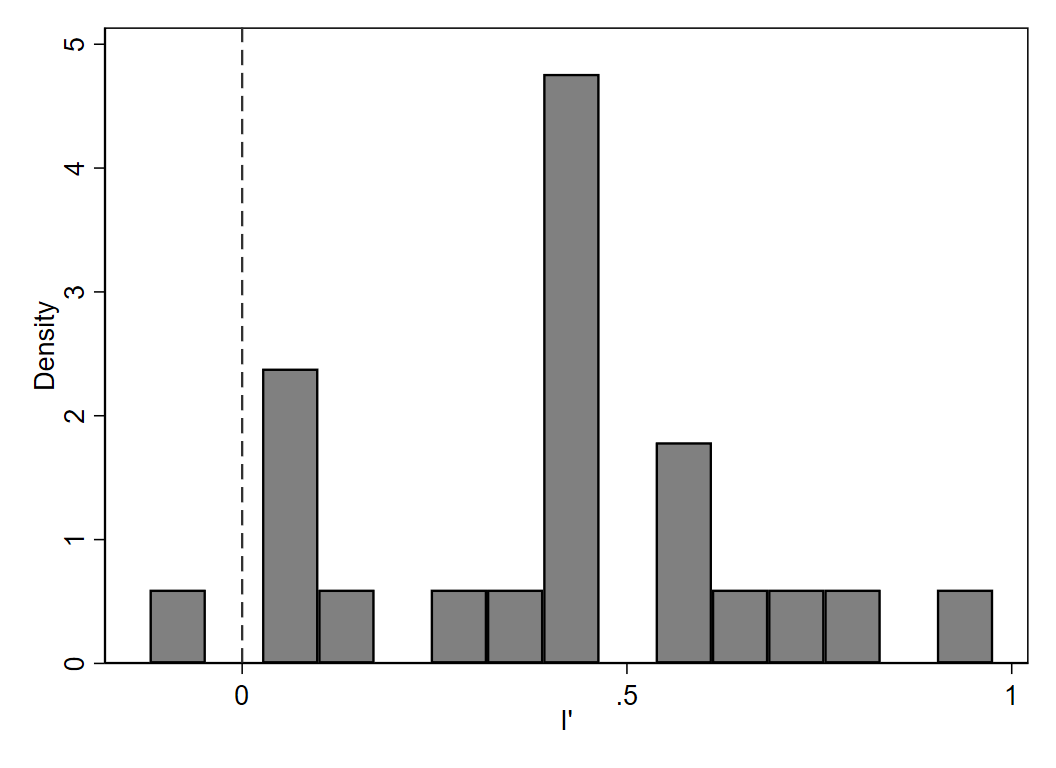}}
\hspace*{.25cm}\subcaptionbox{Materials Inventories}
{\includegraphics[width=.4\textwidth]{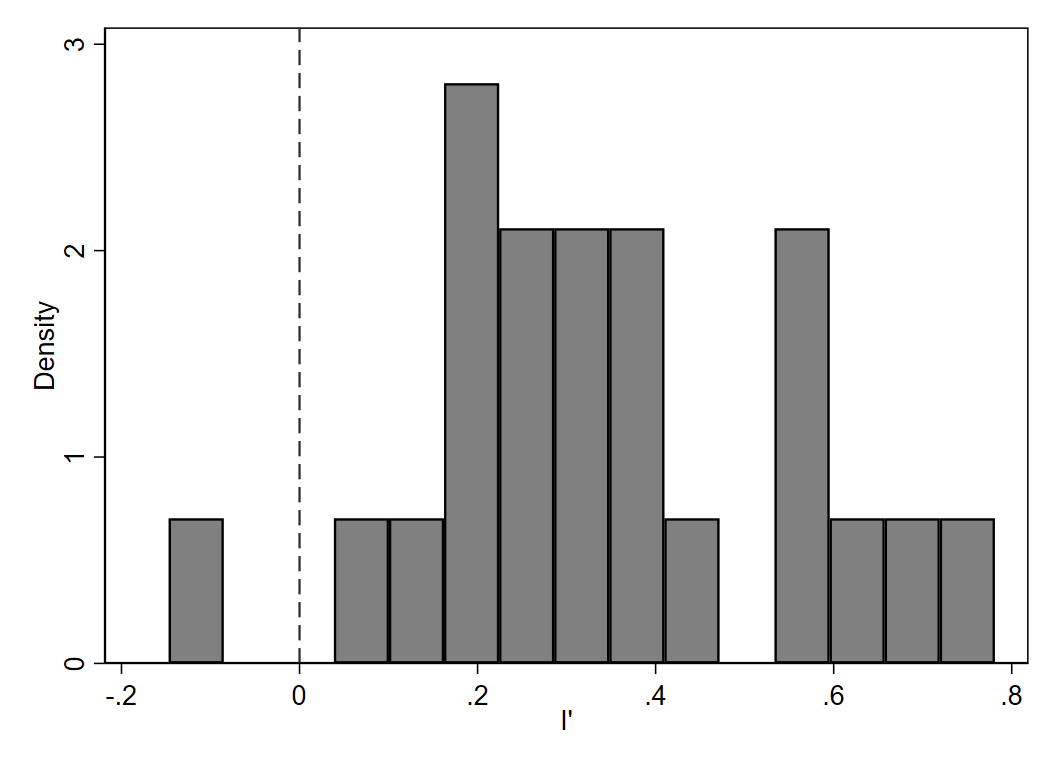}}\\
\noindent\justifying
\footnotesize Note: The graph shows the distribution of estimated $I^\prime(\cdot)$, namely the derivative of the empirical inventory function with respect to sales. The left panel shows the same statistics based on the final goods monthly inventories data from the Manufacturing \& Trade Inventories \& Sales data of the US Census while the right panel shows the estimates using materials inventories.
\end{figure}
%
% \begin{table}[H]
% \caption{Estimation of $I'(\cdot)$}
% \label{npreg}
% \center\begin{threeparttable}
% \input{np.tex}
% \begin{tablenotes}[flushleft]
%  \item \footnotesize Note: this table displays the results of the non-parametric kernel estimation of $I'(\cdot)$, the derivative of the inventory function with respect to current sales. The estimation is based on the data of the NBER CES Manufacturing Industries Dataset for the years 2000-2011. Both series are HP-filtered. Standard errors are bootstrapped. 
% \end{tablenotes}
% \end{threeparttable}
% \end{table}

\paragraph{Output and Sales Volatility}\phantom{a}

\begin{figure}[H]
	\caption{Relative Volatility of Output and Sales}
	\label{fact4plots}
	\begin{subfigure}[b]{0.3\textwidth}
	   \centering
	   \includegraphics[width=1.1\textwidth]{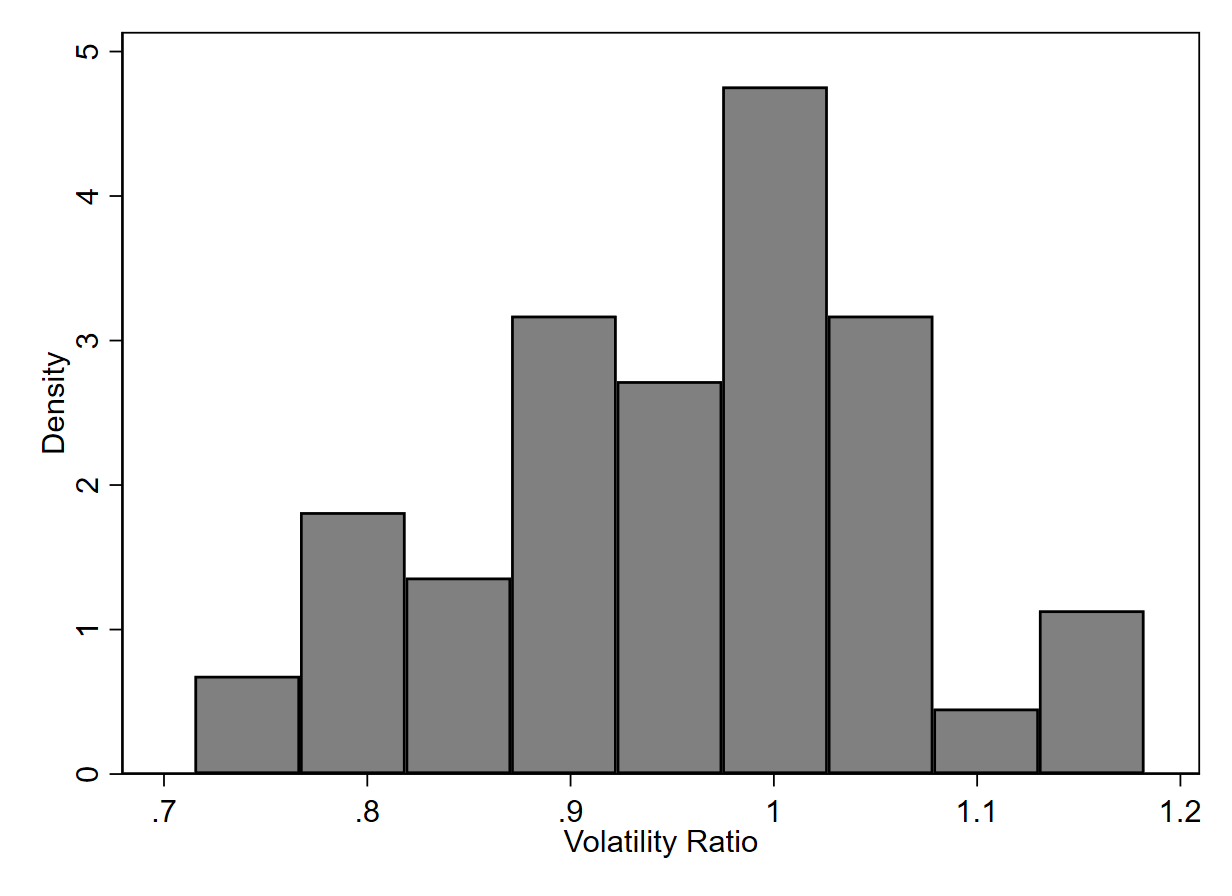}
    	\caption{Monthly}\end{subfigure}
	\hspace{3.5cm}
	\begin{subfigure}[b]{0.3\textwidth}
	   \centering
	   \includegraphics[width=1.1\textwidth]{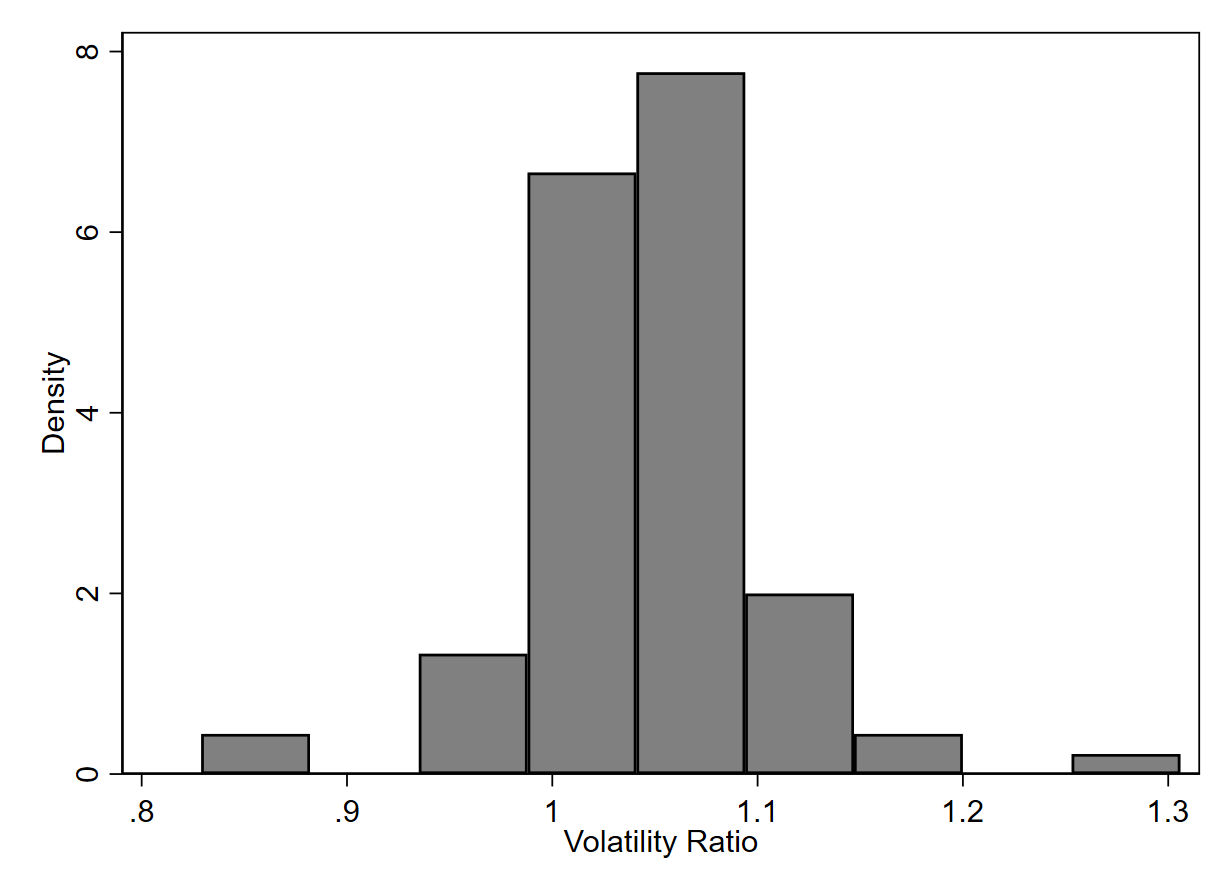}
	\caption{Quarterly}\end{subfigure}	\\
	\begin{subfigure}[b]{0.3\textwidth}
	   \centering
	   \includegraphics[width=1.1\textwidth]{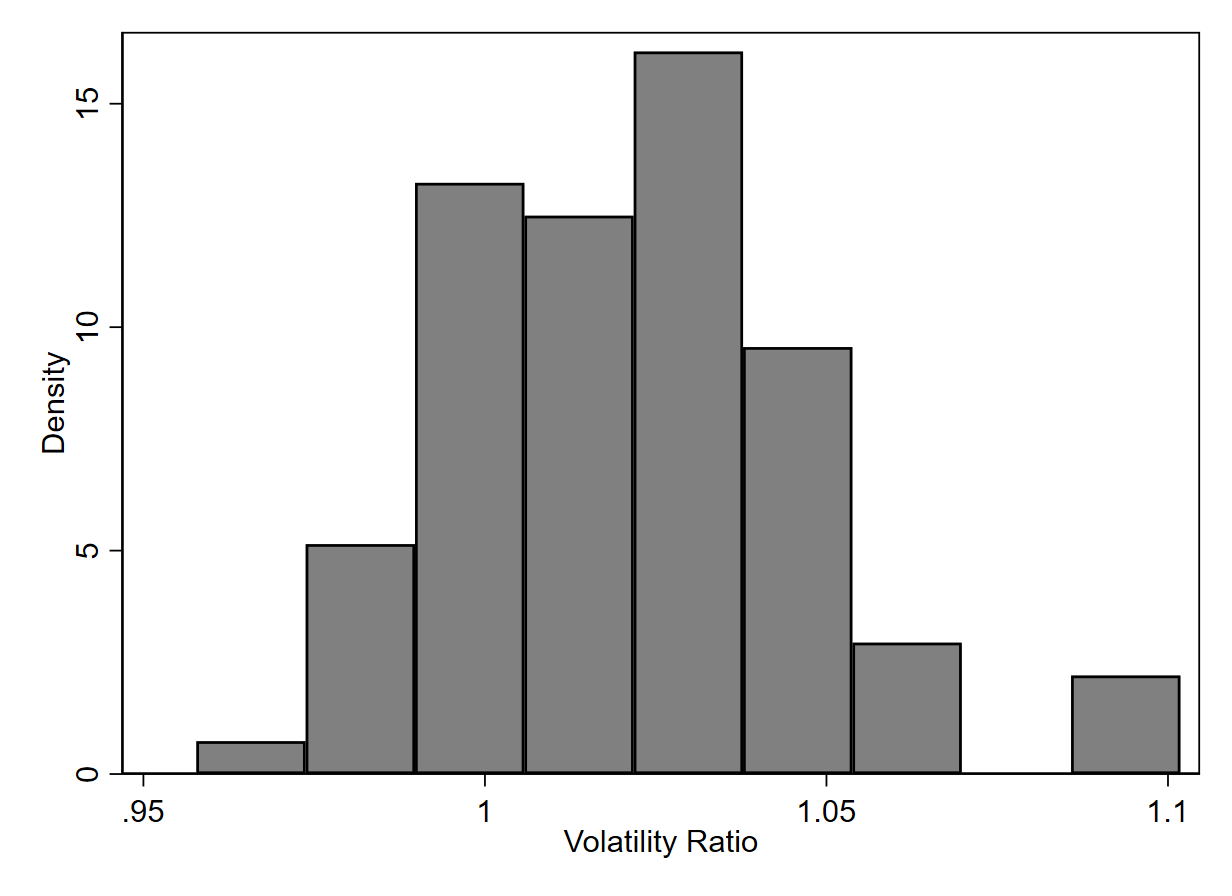}
    	\caption{Yearly}\end{subfigure}
	\hspace{3.5cm}
	\begin{subfigure}[b]{0.3\textwidth}
	   \centering
	   \includegraphics[width=1.1\textwidth]{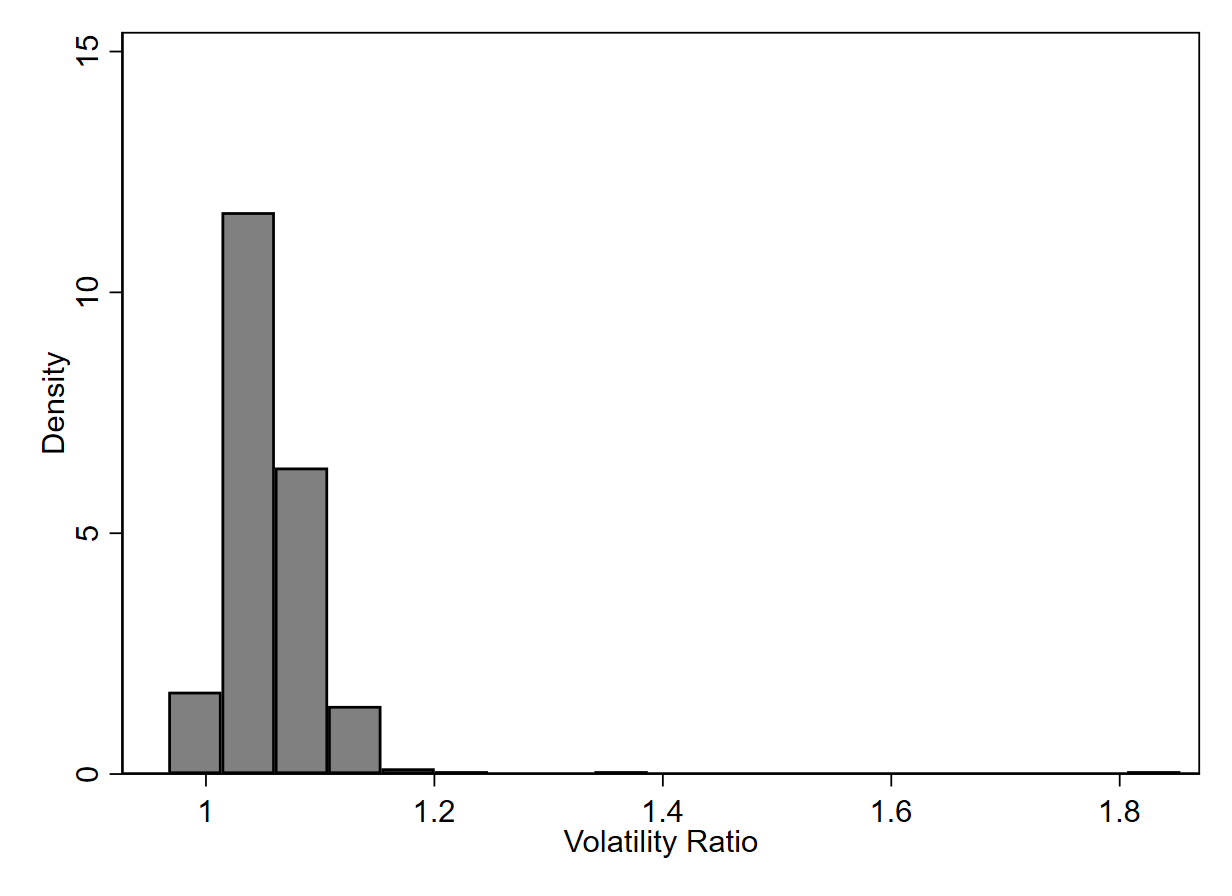}
	\caption{Yearly NBER}\end{subfigure}\\

\noindent\justifying
\footnotesize Note: The graphs show the distribution of the ratio of the volatility of HP-filtered output to HP-filtered sales across sectors. Panels (a), (b) and (c) represent data from the Manufacturing \& Trade Inventories \& Sales data of the US Census, while Panel (d) shows data from the NBER CES Manufacturing data. Both sources are described in the data Section. Panels (a) and (d) have no aggregation, while for Panels (b) and (c) I sum monthly output and sales to get quarterly and yearly output and sales.
\end{figure}

\paragraph{Inventory Policy} 
The definition of the $I$ firms problem in the model implies a constant inventory-to-future sales ratio governed by $\alpha$. This suggests that inventories are a linear function of sales. Figure \ref{lowess_graphs} shows the augmented component-plus-residual plot of the end-of-period stock of inventories as a function of current sales (the same picture arises for next-period sales). The underlying regression includes time and sector fixed effects. The graph is useful for detecting deviations from linearity in the relationship.
 
\begin{figure}[H]
\caption{Inventories and Sales}
\label{lowess_graphs}
\begin{minipage}[b]{.5\linewidth}
\centering\includegraphics[width=.9\textwidth]{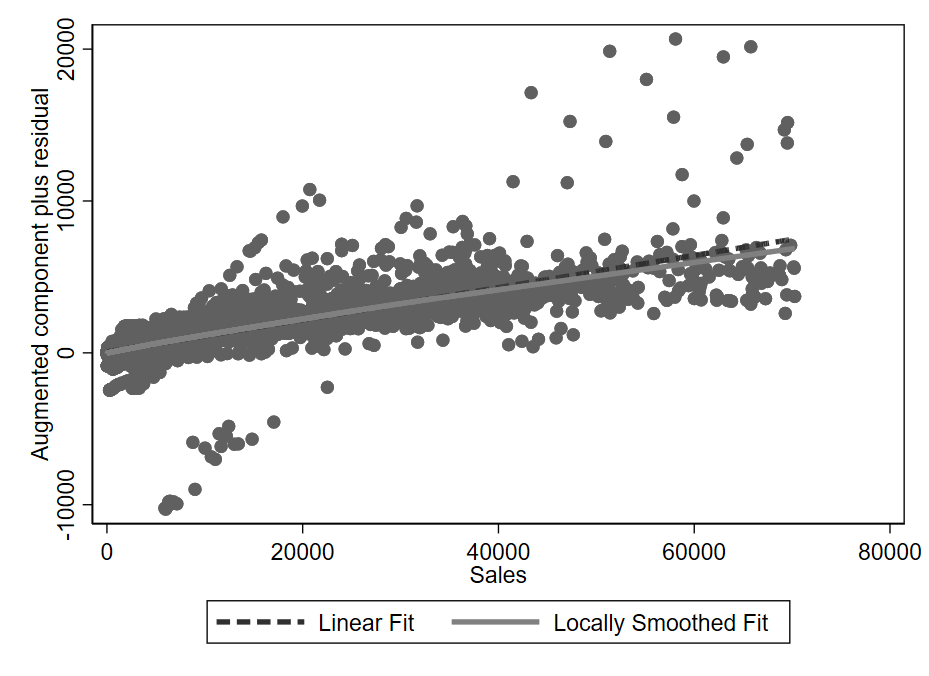}  \subcaption{Inventories and Sales - NBER CES}
\end{minipage}%
\hspace*{.15cm}\begin{minipage}[b]{.5\linewidth}
\centering\includegraphics[scale=0.234]{{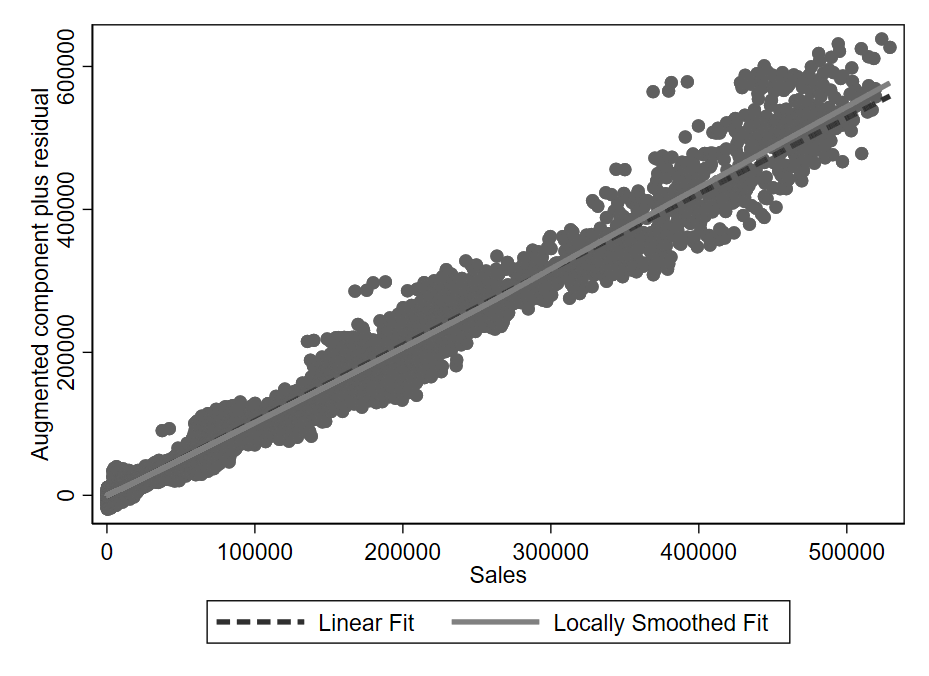}}\subcaption{Inventories and Sales - US Census} 
\end{minipage}
\\
\noindent\justifying
 \footnotesize Note: The figure depicts the augmented component-plus-residual plot of the regression of inventories over sales, including time and industry fixed effects. The black line represents the linear fit of the model. The grey line is a locally weighed smoothing fit. If the data presented significant deviations from linearity the two lines would be very different. 
\end{figure}

Figure \ref{lowess_graphs} suggests that the linearity assumption is very close to the data. The function deviates from linearity only at high sales deviations (recall the estimated model has two-way fixed effects). This suggests that the inventory-to-sales ratio is mostly constant, other than in particularly high sales periods when it starts to decline. 

Table \ref{alphaU} provides the correlation between sector position and inventory sales ratios. The two measures are positively correlated, which suggests that, in the data, more upstream sectors tend to hold a larger fraction of future sales as inventories. 

\begin{table}[H]
\caption{Inventories and Upstreamness}
\label{alphaU}
\center\begin{threeparttable}
\scriptsize{
\def\sym#1{\ifmmode^{#1}\else\(^{#1}\)\fi}
\begin{tabular}{l*{2}{c}}
\hline\hline
                    &\multicolumn{1}{c}{(1)}&\multicolumn{1}{c}{(2)}\\
                    &\multicolumn{1}{c}{ $\alpha_{i,t}^r $  }&\multicolumn{1}{c}{ $\alpha_{i,t}^r $  }\\
\hline
 $ U_{i,t}^r $      &      0.0121\sym{***}&      0.0241\sym{**} \\
                    &   (0.00355)         &    (0.0106)         \\

Constant            &      0.0783\sym{***}&      0.0428         \\
                    &    (0.0108)         &    (0.0313)         \\
\hline
Industry FE         &          No         &         Yes         \\
N                   &         210         &         210         \\
 $ R^2 $            &      0.0524         &       0.863         \\
\hline\hline
\multicolumn{3}{l}{\footnotesize Standard errors in parentheses}\\
\multicolumn{3}{l}{\footnotesize \sym{*} \(p<0.10\), \sym{**} \(p<0.05\), \sym{***} \(p<0.01\)}\\
\end{tabular}
}
\begin{tablenotes}[flushleft]
\item \footnotesize Note: This table shows the results of the estimation of $\alpha$ against upstreamness. Column (1) reports the result of the OLS estimate while Column (2) includes industry fixed effects. change in sales.
\end{tablenotes}
\end{threeparttable}
\end{table}

\subsection*{Test of Uncorrelatedness of Instruments}
\justify As discussed in the main text, the identifying assumption for the validity of the shift share design is conditional independence of shocks and potential outcomes. Since this assumption cannot be tested, I provide evidence that the shares and the shocks are uncorrelated to alleviate endogeneity concerns. I test the conditional correlation by regressing the shares on future shocks and industry fixed effects. Formally
\begin{align*}
    \xi^r_{ijt}=\beta \hat\eta_{jt+1}+\gamma_{it}^r+\epsilon_{ijt}^r.
\end{align*}
The estimation results reported in Table \ref{orthogonality_test} suggest that the two are uncorrelated.
\begin{table}[H]
\caption{Test of Uncorrelatedness of Instruments}
\center
\label{orthogonality_test}
\scriptsize{
\def\sym#1{\ifmmode^{#1}\else\(^{#1}\)\fi}
\begin{tabular}{l*{1}{c}}
\hline\hline
                    &\multicolumn{1}{c}{$\xi^r_{ijt}$}\\
\hline
$\hat\eta_{jt+1}(i)$             &     -0.0121         \\
                    &   (0.00762)         \\
\hline
N                   &     1517824         \\
 $ R^2 $            &     0.00284         \\
\hline\hline
\multicolumn{2}{l}{\footnotesize  Clustered standard errors in parentheses}\\
\multicolumn{2}{l}{\footnotesize \sym{*} \(p<0.10\), \sym{**} \(p<0.05\), \sym{***} \(p<0.01\)}\\
\end{tabular}
}
\end{table}

\newpage

\section{Additional Results - Section \ref{results}}\label{app_empirics_s4}

% \begin{table}[H]
% \caption{Industry Output Growth, Price Indices, and Demand Shocks}
% \label{growthvariance_price}
% \center\begin{threeparttable}
% \input{input/growthvariance_inst.tex}
% \begin{tablenotes}[flushleft]
% \item \footnotesize Note: The table shows the regressions of the growth rate of industry output and the changes in the sectoral price index on the weighted demand shocks that the industry receives. Columns 1 and 2 regress output growth rates on demand shocks with industry and time fixed effects. Columns 3 and 4 show the same regression with the change in the deflator as the outcome. This is computed by taking the ratio of the I-O tables at current and previous year's prices to obtain the growth rate of the deflator from year to year.  
% \end{tablenotes}
% \end{threeparttable}
% \end{table}

\subsection*{\cite{adao} Inference}\addcontentsline{toc}{subsection}{Ad\~ao et al. (2019) Inference}
I follow \cite{adao} (AKM) to compute the standard errors in the main regression \ref{cardinalreg}. The practical difficulty in doing so in this context is threefold: i) I am interested in interaction coefficients between dummies that split the sample and the shift-share regressor; ii) my main specification includes fixed effects; and iii) I use multiple instruments for the interactions of interest. These features of the problem imply that I cannot rely directly on commonly used routines to compute AKM inference. 

I solve this problem with the following steps:
\begin{enumerate}
    \item I use 2SLS instead of the common IV routines. This implies that the first step is to compute the prediction $\hat\eta^{rj}_{it}$ for the Upstreamness bin $j$ of the regressions
    \begin{align}
        \mathbbm{1}\{U_{it-1}^r\in[j,j+1]\}\eta_{it}^r = \sum_{k=1}^5 \gamma_k \mathbbm{1}\{U_{it-1}^r\in[k,k+1]\}+ \beta_k \mathbbm{1}\{U_{it-1}^r\in[k,k+1]\}\tilde\eta_{it}^r+ \delta_{i}^r+ \epsilon_{it}^{rj}
    \end{align}
    This step generates the 2SLS for each interacted instrument with the appropriate fixed effects.
    \item Demean both the residual and the outcome $\Delta\log Y^r_{it}$ by subtracting the time-series average within each country-industry to obtain the equivalent fixed effect as the main estimation
    \begin{align}
\bar X^{r}_{it}=X^{r}_{it}-\frac1{T}\sum_t X^{r}_{it},
    \end{align}
where $X^{r}_{it}=\{\Delta\log Y^r_{it}, \{\hat\eta^{rj}_{it}\}_{j=1,...,5}\}$.
\item Finally, estimate separately for each $j=1,...,5$
\begin{align*}
\bar{\Delta\log Y^{r}_{it}}=\chi_j\bar{\hat{ \eta}}_{it}^{rj}+\upsilon_{it}^{rj} \quad \textbf{ if } U_{it-1}^r\in[j,j+1]
\end{align*}
This allows me to recover the original point estimates but apply the AKM standard error correction in the last step. 
\end{enumerate}

\begin{figure}[H]
\caption{AKM estimates}
\centering
\label{akm_fig}
  \includegraphics[width=.6\linewidth]{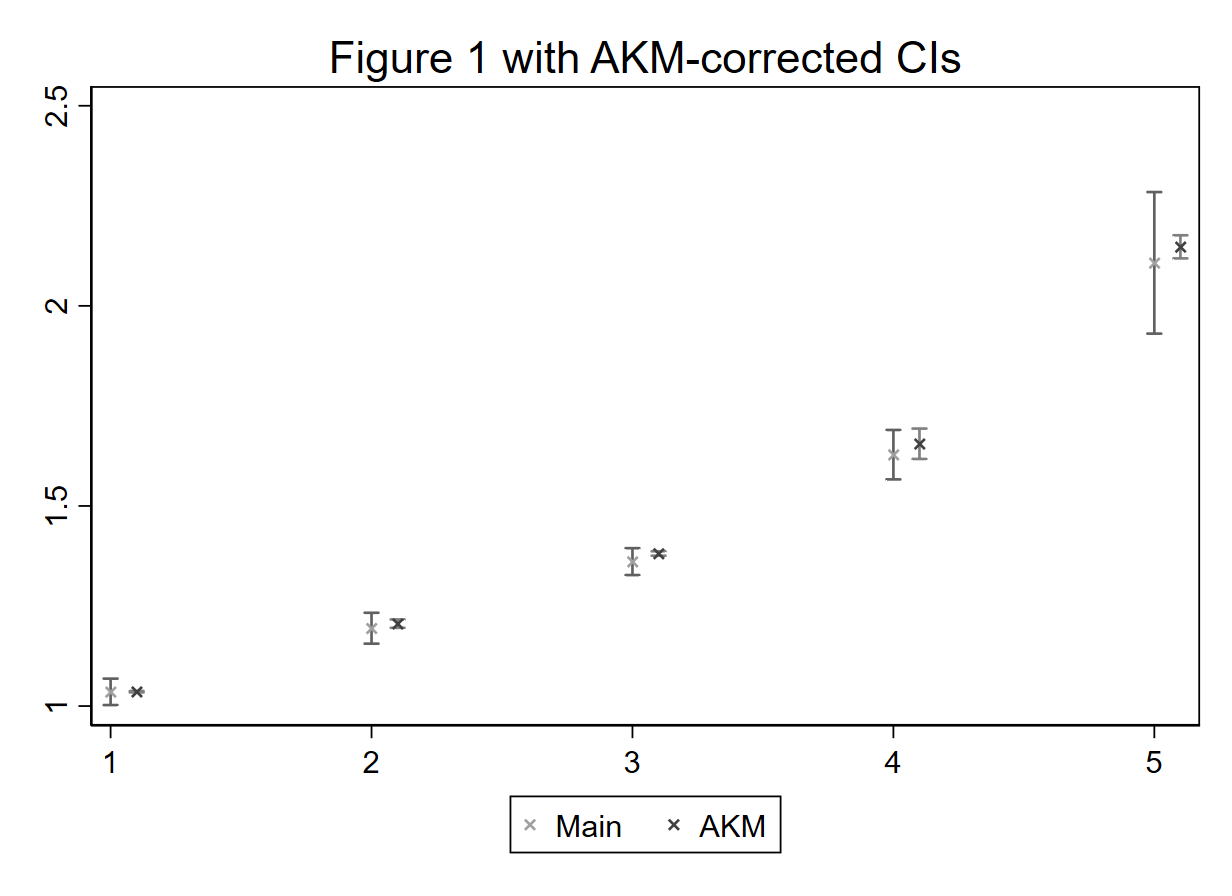}
\\\vspace{5pt}
\noindent\justifying
\footnotesize Note: the figure shows the AKM estimates for the main result. 
\end{figure}

\subsection*{Alternative Demand Shocks}\addcontentsline{toc}{subsection}{Government Consumption Shocks}
\paragraph*{Instrumenting Demand Shocks with Government Consumption} In the empirical specification in Section \ref{resultsindustrysec}, I use the variation arising from destination-time specific changes in foreign aggregate demand. One might worry that if two countries $i$ and $j$ trade intensely, this measure could be plagued by reverse causality, such that if large sectors grow in country $i$ they could affect demand in country $j$. To further alleviate these concerns, in this section, I use foreign government consumption as an instrument. 
\begin{figure}[htb]
\centering
\caption{Effect of Demand Shocks on Output Growth by Upstreamness Level - Government Consumption}
\label{controlfunction}
  \centering
  \includegraphics[width=.5\linewidth]{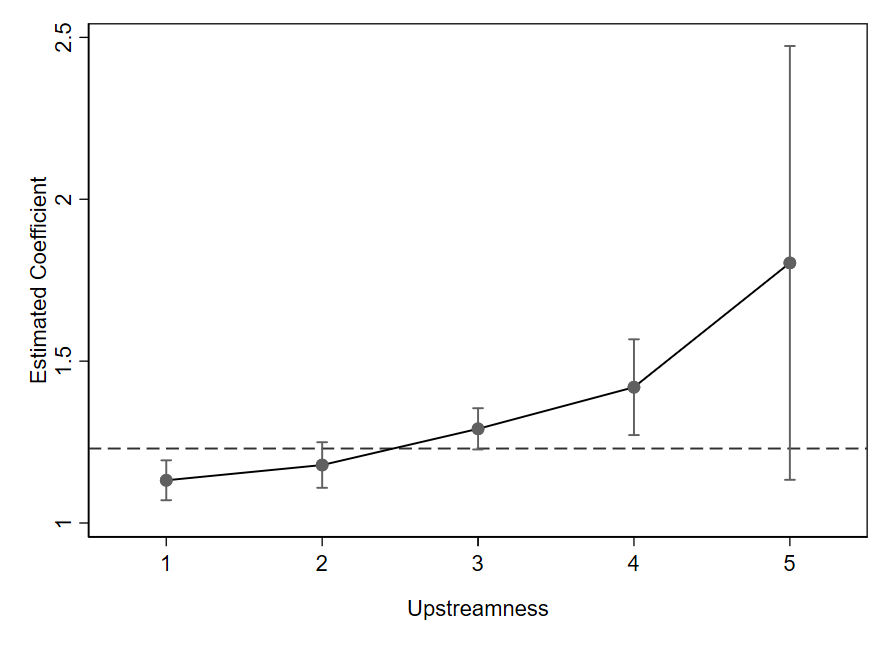}
\\\vspace{5pt}
\noindent\justifying
\footnotesize Note: The figure shows the marginal effect of demand shocks on industry output changes by industry upstreamness level using the government consumption instrument. The dashed horizontal line represents the average coefficient. The vertical bands illustrate the 95\% confidence intervals around the estimates. The regression includes country-industry fixed effects, and the standard errors are cluster-bootstrapped at the country-industry level. Note that due to relatively few observations above 6, all values above 6 have been included in the $U\in[5,\infty)$ category. The full regression results are reported in Table \ref{controlfunctiontable}.
\end{figure}
More specifically, the WIOD data contains information about the value of purchases of the government of country $j$ of goods of industry $r$ from country $i$ in period $t$. Denote this $G^r_{ijt}$. I apply the same steps as in Section \ref{methodology}, replacing consumer purchases with government ones. This procedure allows me to create a measure $\hat\eta^{rG}_{it}=\sum_j\xi^{r}_{ij}\hat\eta^G_{jt}$. With the understanding that, as in the main results, the estimated destination-time shifter $\hat\eta^G_{jt}$ is estimated by excluding all purchases of goods from country $i$ or industry $r$.
I then follow the same steps as in the main specification in eq. (\ref{cardinalreg}) with the alternative instrument. Figure \ref{controlfunction} shows the estimated coefficients of interest. The full estimate is reported in Table \ref{controlfunctiontable}. I confirm the main results both qualitatively and quantitatively.
\begin{table}[H]
\caption{Effect of Demand Shocks on Output Growth by Upstreamness Level - Government Consumption Instrument}
\label{controlfunctiontable}
\center\begin{threeparttable}
\scriptsize{\def\sym#1{\ifmmode^{#1}\else\(^{#1}\)\fi}\begin{tabular}{l*{6}{c}} \hline\hline & \multicolumn{5}{c}{{First Stage}} & Second Stage \\ \midrule
                    &\multicolumn{1}{c}{(1)}&\multicolumn{1}{c}{(2)}&\multicolumn{1}{c}{(3)}&\multicolumn{1}{c}{(4)}&\multicolumn{1}{c}{(5)}&\multicolumn{1}{c}{(6)}\\
                    &\multicolumn{1}{c}{$ \eta^r_{it}\times  U \in $ [1,2]}&\multicolumn{1}{c}{$ \eta^r_{it}\times  U \in $ [2,3]}&\multicolumn{1}{c}{$ \eta^r_{it}\times  U \in $ [3,4]}&\multicolumn{1}{c}{$ \eta^r_{it}\times  U \in $ [4,5]}&\multicolumn{1}{c}{$ \eta^r_{it}\times  U \in $ [5,$ \infty $)}&\multicolumn{1}{c}{$\Delta \ln Y^r_{it}$}\\
\hline
$ \tilde\eta^r_{it}\times  U \in $ [1,2]&       0.170\sym{***}&    0.000220         &    0.000107         &  -0.0000152         &  -0.0000648         &                     \\
                    &   (0.00205)         &   (0.00203)         &   (0.00223)         &   (0.00141)         &  (0.000690)         &                     \\
[1em]
$ \tilde\eta^r_{it}\times  U \in $ [2,3]&    0.000758         &       0.159\sym{***}&   -0.000212         &   0.0000298         &   -0.000146         &                     \\
                    &   (0.00194)         &   (0.00192)         &   (0.00210)         &   (0.00133)         &  (0.000651)         &                     \\
[1em]
$ \tilde\eta^r_{it}\times  U \in $ [3,4]&    0.000106         &    0.000650         &       0.151\sym{***}&   -0.000348         &   -0.000145         &                     \\
\end{tabular}
}
\begin{tablenotes}[flushleft]
 \item \footnotesize Note: The Table shows the results of the regression in equation \ref{cardinalreg} with output changes as the dependent variable and using the government consumption instrument. In particular, I regress the output changes over output on the instrumented demand shocks interacted with dummies taking value 1 if upstreamness is in the $[1,2]$ bin, $[2,3]$ bin, and so on. Observations with upstreamness above 6 are included in the $[5,\infty)$ bin. All regressions include producing industry-country fixed effects. Columns 1 to 5 report the first-stage results for each of the endogenous variables. Column 6 reports the second-stage results. Standard errors are clustered at the producing industry-country level.
\end{tablenotes}
\end{threeparttable}
\end{table}

\newpage

\section{Robustness Checks}\label{robustnessappendix}

In this section, I provide a set of robustness checks. First, I apply the correction proposed by \cite{borusyak2020non} to correct for potential omitted variable bias. To compare my estimates with existing ones in the recent literature I re-estimate the reduced form result after discretizing the network data. I also reproduce the main result in the ordinal, rather than cardinal, binning, and under alternative fixed effects models to estimate the final demand shifters. I conclude by showing that the results are robust to estimation on deflated data, to account for potential price effects, using time-varying aggregation shares and controlling for past output as suggested by \cite{aak}.

\subsection*{Re-centered Instrument}
\addcontentsline{toc}{subsection}{Re-centered Instrument}
In a recent paper \cite{borusyak2020non} show that when using shift-share design, there is a risk of omitted variable bias arising from potentially non-random shares. They also suggest to re-center the instrument to prevent such bias by using the average counterfactual shock.

I apply this methodology by permuting N=1000 times, within each year, the distribution of destination shocks $\hat\eta_{jt}$. After the permutation, I compute the average for each treated unit and demean the original demand shock to create $\tilde \eta_{it}^r=\hat \eta_{it}^r-\mu_{it}^r$. Where $\mu_{it}^r\equiv \frac{1}{N}\sum_n \sum_j \xi_{ij}^r\tilde\eta_{jt}$ and $\tilde\eta_{it}^r$ is the permuted shock. I then re-estimate the main specification in equation \ref{cardinalreg} with the re-centered shocks. The results, shown in Table \ref{recenteredtable}, are unchanged both qualitatively and quantitatively. 

\begin{table}[H]
     \caption{Re-centered Instrument Estimation}
         \label{recenteredtable}
   \centering\begin{threeparttable}
    {\def\sym#1{\ifmmode^{#1}\else\(^{#1}\)\fi}\begin{tabular}{l*{6}{c}} \hline\hline & Baseline & Recentered Instrument \\ \midrule
                    &\multicolumn{1}{c}{(1)}&\multicolumn{1}{c}{(2)}\\
                    &\multicolumn{1}{c}{$\Delta \ln Y^r_{it}$}&\multicolumn{1}{c}{$\Delta \ln Y^r_{it}$}\\
\hline
$ \eta^r_{it}\times  U \in $ [1,2]&       1.035\sym{***}&       1.036\sym{***}\\
                    &    (0.0168)         &    (0.0169)         \\
[1em]
$ \eta^r_{it}\times  U \in $ [2,3]&       1.194\sym{***}&       1.195\sym{***}\\
                    &    (0.0197)         &    (0.0198)         \\
[1em]
$ \eta^r_{it}\times  U \in $ [3,4]&       1.361\sym{***}&       1.362\sym{***}\\
                    &    (0.0171)         &    (0.0172)         \\
[1em]
$ \eta^r_{it}\times  U \in $ [4,5]&       1.628\sym{***}&       1.628\sym{***}\\
                    &    (0.0315)         &    (0.0316)         \\
[1em]
$ \eta^r_{it}\times  U \in $ [5,$ \infty $)&       2.107\sym{***}&       2.092\sym{***}\\
                    &    (0.0902)         &    (0.0877)         \\
\hline
Country-Industry FE &         Yes         &         Yes         \\
U Bin FE            &         Yes         &         Yes         \\
F-Stat              &                     &                     \\
N                   &       32251         &       32251         \\
 $ R^2 $            &       0.498         &       0.498         \\
\hline\hline
\multicolumn{3}{l}{\footnotesize Standard errors are clustered at the producing industry $\times$ country level.}\\
\multicolumn{3}{l}{\footnotesize \sym{*} \(p<0.10\), \sym{**} \(p<0.05\), \sym{***} \(p<0.01\)}\\
\end{tabular}
}
    \begin{tablenotes}[flushleft]
\item \footnotesize Note: The Table shows the results of the regression in equation \ref{cardinalreg} in column 1 and the re-centered instrument approach in \cite{borusyak2020non}. The latter is done by computing permutations of the shocks to demean the shift-share shock. Standard errors are clustered at the producing industry-country level.
\end{tablenotes}
\end{threeparttable}
\end{table}

\subsection*{Alternative Shifter Estimation}
\addcontentsline{toc}{subsection}{Alternative Shifter Estimation}
\normalsize
In section \ref{methoddatasec}, I used the fixed effect model to gauge the idiosyncratic demand shocks. Such a model may confound other sources of variation such as supply shocks, along with the object of interest. To investigate this possibility I use an alternative econometric model to extract the demand shocks. 
Following \cite{kramarz2020volatility} more closely, I include producer fixed effects: $\gamma_{it}^r$ is the fixed effect for the producing industry $r$ in country $i$ at time $t$, namely
\begin{align}
\Delta f_{kjt}^s=\eta_{jt}(i,r)+\gamma_{jt}^s+\nu_{kjt}^s\quad k\neq i, s\neq r.
\label{fixedeffectssupply}
\end{align}
Where the conditions $k\neq i, s\neq r$ ensure that domestically produced goods used for final consumption are not included in the estimation and neither are the goods within the same sector.  The result for the main specification (equation \ref{cardinalreg}) is presented in Table \ref{fe_robustness}.
The findings confirm the main results in Section \ref{resultsindustrysec} both qualitatively and quantitatively. 
\begin{table}[H]
\caption{Effect of Demand shocks by level of Upstreamness - Alternative Shifter Estimation}
\label{fe_robustness}
\center\begin{threeparttable}
  \scriptsize{\def\sym#1{\ifmmode^{#1}\else\(^{#1}\)\fi}\begin{tabular}{l*{6}{c}} \hline\hline & \multicolumn{5}{c}{{First Stage}} & Second Stage \\ \midrule
                    &\multicolumn{1}{c}{(1)}&\multicolumn{1}{c}{(2)}&\multicolumn{1}{c}{(3)}&\multicolumn{1}{c}{(4)}&\multicolumn{1}{c}{(5)}&\multicolumn{1}{c}{(6)}\\
                    &\multicolumn{1}{c}{$ \eta^r_{it}\times  U \in $ [1,2]}&\multicolumn{1}{c}{$ \eta^r_{it}\times  U \in $ [2,3]}&\multicolumn{1}{c}{$ \eta^r_{it}\times  U \in $ [3,4]}&\multicolumn{1}{c}{$ \eta^r_{it}\times  U \in $ [4,5]}&\multicolumn{1}{c}{$ \eta^r_{it}\times  U \in $ [5,$ \infty $)}&\multicolumn{1}{c}{$\Delta \ln Y^r_{it}$}\\
\hline
$ \tilde\eta^r_{it}\times  U \in $ [1,2]&       0.452\sym{***}&   -0.000645         &  -0.0000966         &   -0.000122         &   -0.000117         &                     \\
                    &   (0.00212)         &   (0.00207)         &   (0.00208)         &   (0.00129)         &  (0.000739)         &                     \\
[1em]
$ \tilde\eta^r_{it}\times  U \in $ [2,3]&    0.000283         &       0.451\sym{***}&    -0.00225         &   -0.000530         &   -0.000197         &                     \\
                    &   (0.00210)         &   (0.00205)         &   (0.00206)         &   (0.00128)         &  (0.000732)         &                     \\
[1em]
$ \tilde\eta^r_{it}\times  U \in $ [3,4]&   0.0000455         &    0.000704         &       0.477\sym{***}&   -0.000657         &   -0.000372         &                     \\
\end{tabular}
}
\begin{tablenotes}[flushleft]
\item \footnotesize Note: The Table shows the results of the regression in equation \ref{cardinalreg}. In particular, I regress the growth rate of output on the instrumented demand shocks interacted with dummies taking value 1 if upstreamness is in the $[1,2]$ bin, $[2,3]$ bin, and so on. Observations with upstreamness above 6 are included in the $[5,\infty)$ bin. All regressions include producing industry-country fixed effects. Columns 1 to 5 report the first-stage results for each of the endogenous variables. Column 6 reports the second-stage results. Standard errors are clustered at the producing industry-country level.
\end{tablenotes}
\end{threeparttable}
\end{table}
\newpage
\raggedright\section{Omitted Proofs}\label{proofs}
\normalsize
\justifying

\begin{proof}[Proof of Lemma \ref{thmoutputline}]
    Define the operator $\mathcal{F}^n_t[x_t]=I_n(\mathbb{E}_t[x_{t+1}+ \sum_{i=0}^{n-1}\mathcal{F}^{i}_{t+1}[x_{t+1}]- \mathcal{F}^{i}_t[x_t]])$ with the initial condition $\mathcal{F}^0_t[x_t]=I_0(\mathbb{E}_t [x_{t+1}])$. Further, define $\Delta \mathcal{F}^n_t[x_t]\equiv \mathcal{F}^n_t[x_t]-\mathcal{F}^n_{t-1}[x_{t-1}].$ 

Then, sectoral output at each stage is given by
\begin{align*}\nonumber
    Y^0_t&=D_t^0+\mathcal{F}^0_t[D_t^0]-\mathcal{F}^0_{t-1}[D_t^0]\\\nonumber
    Y^1_t&=D_t^0+\mathcal{F}^0_t[D_t^0]-\mathcal{F}^0_{t-1}[D_t^0]+\mathcal{F}^1_t[D_t^0]-\mathcal{F}^1_{t-1}[D_t^0]\\\nonumber
    Y^2_t&=D_t^0+\mathcal{F}^0_t[D_t^0]-\mathcal{F}^0_{t-1}[D_t^0]+\mathcal{F}^1_t[D_t^0]-\mathcal{F}^1_{t-1}[D_t^0] +\mathcal{F}^2_t[D_t^0]-\mathcal{F}^2_{t-1}[D_t^0]\\\nonumber
    &\vdots\\\nonumber
    Y^n_{t}&=D_t^0+\sum_{i=0}^n\mathcal{F}^i_t[D_t^0]-\mathcal{F}^i_{t-1}[D_t^0]=\\
    Y^n_{t}&=D_t^0+\sum_{i=0}^n\Delta\mathcal{F}^i_t[D_t^0].
\end{align*}
\end{proof}

\begin{proof}[Proof of Proposition \ref{propbullwhip}]
I prove the proposition in three steps
\begin{enumerate}[label=\alph*)]
\item From Proposition \ref{thmoutputline}, evaluating the derivative at stage 0
\begin{align*}
    \diffp{{Y^0_t}}{{D_t^0}}=1+\rho I^\prime_0.
\end{align*}
At stage 1
\begin{align*}
    \diffp{{Y^1_t}}{{D_t^0}}&=1+\rho I^\prime_0+\frac{\partial }{\partial D^0_t}\left[\mathcal{F}^1_t \right]=\\
    &=1+\rho I^\prime_0+\frac{\partial }{\partial D^0_t}\left[I_1(\mathbb{E}_t[D^0_{t+1}+\mathcal{F}^0_{t+1}-\mathcal{F}^0_t]) \right]=\\
    &=1+\rho I^\prime_0+I^\prime_1(\rho+\rho^2I^\prime_0-\rho I^\prime_0)=\\
    &=1+\rho I^\prime_0+\rho I^\prime_1(1+\rho I^\prime_0-I^\prime_0)=1+\rho I^\prime_0+\rho I^\prime_1(1+(\rho-1) I^\prime_0).
\end{align*}
At stage 2
\begin{align*}
    \diffp{{Y^2_t}}{{D_t^0}}&=1+\rho I^\prime_0+\rho I^\prime_1(1+\rho I^\prime_0-I^\prime_0)+ \frac{\partial }{\partial D^0_t}\left[\mathcal{F}^2_t \right]=\\
    &=1+\rho I^\prime_0+\rho I^\prime_1(1+\rho I^\prime_0-I^\prime_0)+\frac{\partial }{\partial D^0_t}\left\{ I_2\left(\mathbb{E}_t\left[D^0_{t+1}+\mathcal{F}^1_{t+1}-\mathcal{F}^1_t+\mathcal{F}^0_{t+1}-\mathcal{F}^0_t\right ] \right)  \right\}=\\
   &= 1+\rho I^\prime_0+\rho I^\prime_1(1+\rho I^\prime_0-I^\prime_0)+ I^\prime_2\left[\rho +\diffp{{\mathbb{E}_t\mathcal{F}^1_{t+1}}}{{D_t^0}} - \diffp{{\mathcal{F}^1_{t}}}{{D_t^0}}+ \diffp{{\mathbb{E}_t\mathcal{F}^0_{t+1}}}{{D_t^0}}- \diffp{{\mathcal{F}^0_{t+1}}}{{D_t^0}}\right]=\\
   &=1+\rho I^\prime_0+\rho I^\prime_1(1+\rho I^\prime_0-I^\prime_0)+ I^\prime_2\left[\rho +\diffp{{\mathbb{E}_t\mathcal{F}^1_{t+1}}}{{D_t^0}} - \rho I^\prime_1(1+(\rho-1) I^\prime_0)+ \diffp{{\mathbb{E}_t\mathcal{F}^0_{t+1}}}{{D_t^0}}- \rho I^\prime_0\right].
\end{align*}
Focusing on $\diffp{{\mathbb{E}_t\mathcal{F}^0_{t+1}}}{{D_t^0}}$
\begin{align*}
    \diffp{{\mathbb{E}_t\mathcal{F}^0_{t+1}}}{{D_t^0}}&=\diffp{{}}{{D_t^0}}\mathbb{E}_t[I_0(\mathbb{E}_{t+1}D_{t+2}^0]=\diffp{{}}{{D_t^0}}\mathbb{E}_t[I_0\left(\bar D(1-\rho)+\rho D^0_{t+1}\right)]=\\
    &=\diffp{{}}{{D_t^0}}\int I_0\left((1-\rho)\bar{D}+\rho D_{t+1}^0\right)f(\epsilon_{t+1})d \epsilon_{t+1}=\\
    &=\diffp{{}}{{D_t^0}}\int I_0\left((1-\rho)\bar{D}+\rho ((1-\rho)\bar{D}+\rho D^0_t+\epsilon_{t+1})\right)f(\epsilon_{t+1})d \epsilon_{t+1}=\\
    &=\int\diffp{{}}{{D_t^0}} I_0\left((1-\rho)\bar{D}+\rho ((1-\rho)\bar{D}+\rho D^0_t+\epsilon_{t+1})\right)f(\epsilon_{t+1})d \epsilon_{t+1}=\\
    &=\int I^\prime_0\rho^2 f(\epsilon_{t+1})d \epsilon_{t+1}=I^\prime_0\rho^2\int  f(\epsilon_{t+1})d \epsilon_{t+1}=I^\prime_0\rho^2.
\end{align*}
Where in the fourth step I make use of the bounded support assumption and the Leibniz rule. Identical steps yield $\diffp{{\mathbb{E}_t\mathcal{F}^1_{t+1}}}{{D_t^0}}=\rho^2 I^\prime_1(1+\rho I^\prime_0(\rho-1))$. Plugging in and rearranging
\begin{align*}
     \diffp{{Y^2_t}}{{D_t^0}}&=1+\rho I^\prime_0+\rho I^\prime_1(1+(\rho-1) I^\prime_0)+I^\prime_2(1+(\rho-1) I^\prime_0)(1+(\rho-1) I^\prime_0)
\end{align*}
Similarly at higher stages. By forward induction, this can be compactly written as 
\begin{align*}
    1+\sum_{i=0}^n\rho I^\prime_i\prod_{j=0}^{i-1} [1+(\rho-1)I^\prime_{j}].
\end{align*}
Hence, dividing through by $D^0$, to a first-order approximation
\begin{align*}
    \frac{\Delta Y^n_t}{D^0}=\frac{\Delta^0_t}{D^0}\left[1+\sum_{i=0}^n\rho I^\prime_i\prod_{j=0}^{i-1} [1+(\rho-1)I^\prime_{j}]\right],
\end{align*}
Where $\Delta^0_t=D_t-D_t-1$. Noting that, in the non-stochastic steady state $D^0=Y^n,\,\forall n$ and replacing with log changes
\begin{align*}
    \Delta \log  Y^n_t=\Delta\log D^0_t\left[1+\sum_{i=0}^n\rho I^\prime_i\prod_{j=0}^{i-1} [1+(\rho-1)I^\prime_{j}]\right].
\end{align*}
\item Given point a), if $0<I^\prime_i<\frac{1}{1-\rho}$, $\forall\,i$, every term in the product is positive, hence $ \Delta \log  Y^n_t>\ \Delta \log  Y^{n-1}_t$, $\forall n,t$. 

\item Finally, note that to a first order $\frac{\Delta \log  Y^n_t}{\Delta \log  D^0_t}-\frac{\Delta \log  Y^{n-1}_t}{\Delta \log  D^0_t}=\rho I^\prime_n\prod_{j=0}^{n-1} [1+(\rho-1)I_{j}^\prime$. I want to show under which condition this increases in $n$. Hence, when $\frac{\Delta \log  Y^n_t}{\Delta \log  D^0_t}-\frac{\Delta \log  Y^{n-1}_t}{\Delta \log  D^0_t}-\frac{\Delta \log  Y^{n-1}_t}{\Delta \log  D^0_t}+\frac{\Delta \log  Y^{n-2}_t}{\Delta \log  D^0_t}>0$. This is equal to 
\begin{align*}
    &\rho I^\prime_n\prod_{j=0}^{n-1} [1+(\rho-1)I_{j}^\prime]-\rho I^\prime_{n-1}\prod_{j=0}^{n-2} [1+(\rho-1)I_{j}^\prime]=\\
    &\rho (I^\prime_n-I^\prime_{n-1})\prod_{j=0}^{n-2} [1+(\rho-1)I_{j}^\prime]+\rho I^\prime_n (1+(\rho-1)I^\prime_{n-1}).
\end{align*}
If the condition in equation (\ref{Iprimecond}) is satisfied, then the second term is positive, while the first depends on the sign of $I^\prime_n-I^\prime_{n-1}$. Hence if $I^\prime_n$ is not too decreasing in $n$, the second term dominates, the expression is positive and, therefore, $\frac{\Delta \log  Y^n_t}{\Delta \log  D^0_t}-\frac{\Delta \log  Y^{n-1}_t}{\Delta \log  D^0_t}$ increases in $n$.  \end{enumerate}
 \end{proof}

\subsection*{Equilibrium Definition and Characterization}

\begin{definition}[Equilibrium]
An equilibrium in this economy is given by a sequence of i) household consumption and labor policies $L_t, C_{s,t},\, \forall s$ that, given prices, maximize their utility; ii) a sourcing policy of $C$ firms $Y_{sr,t}, \forall s,r$ that, given input prices, maximize their profits; iii) pricing, inventory and output policies of $I$ firms so that, given input prices, they maximize their profits and iv) market clearing conditions for all differentiated varieties and labor.
\end{definition}

\paragraph*{Households} The household expenditure problem is given by
\begin{align*}
    \max_{C,L} U=\log C-L\quad \text{st} \quad wL +\Pi=PC
\end{align*}
defining $\lambda$ the multiplier on the budget constraint, this implies
\begin{align*}
\lambda w&=1,\\
C^{-1}&=\lambda P.
\end{align*}
With the numeraire condition $w=1$, this implies $PC=1$.

The expenditure minimization on differentiated varieties is given by 
\begin{align*}
    \min_{C_s} \sum_{s\in S} C_s p_s \quad\text{st} \quad \bar U=\Pi_{s\in S} C_s^{\beta_s}
\end{align*}
which implies the optimal domestic expenditure on variety $s$, denoted $E^D_s= \beta_s$.
The total expenditure, given by the sum of domestic and foreign expenditures, is  $D_t=PC+X_t=1+X_t$. Given the stochasatic process for $X_t=(1-\rho)\bar X+\rho X_{t-1}+\epsilon_t$, the total expenditure evolves according to $D_t=\rho D_{t-1}+(1-\rho)\bar D+\epsilon_t$, where $\bar D\coloneqq 1+\bar X$.

\paragraph*{Production} $C$ firms operate like in the standard network model without productivity shocks. Formally, taking prices as given, they solve 
\begin{align*}
    \max_{l_s, Y_{rs,t}} p_{s,t} Z_s l_{s,t}^{1-\gamma_s}\left(\sum_{r\in R} {a_{rs}}^{1/\nu}Y_{rs,t}^{\frac{\nu-1}{\nu}}\right)^{\gamma_s\frac{\nu}{\nu-1}} -w l_{s,t} - \sum_{r\in R} p_r Y_{rs,t}.
\end{align*}
With $Z_s=(1-\gamma_s)^{\gamma_s-1}\gamma_s^{-\gamma_s\frac{\nu}{\nu-1}}$. Given $w=1$, the optimal amount of labor costs is $l_{s,t}=(1-\gamma_s)p_{s,t} Y_{s,t}$, while the total expenditure on intermediates is the $\gamma_s p_{s,t} Y_{s,t}$. Within intermediate expenditure, firms in sector $s$ spend on input $r$ $p_{r,t} Y_{rs,t}=(p_{r,t}/\mathcal{C}_{s,t})^{\nu-1} a_{rs}\gamma_s p_{s,t} Y_{s,t} $, where $\mathcal{C}_{s,t}$ is the ideal intermediate goods cost index of the firm.

At the optimum the firm therefore uses
\begin{align*}
    Y_{rs,t}=\gamma_s(1-\gamma_s)a_{rs} p_{s,t}Y_{s,t}p_{r,t}^{-\nu}\left(\sum_{q\in R} a_{qs} p_{q,t}^{1-\nu}\right)^{-1},
\end{align*}
units of input $r$. Plugging back into the production function and using the normalization constant $Z_s$ implies the pricing equation
\begin{align*}
    p_{s,t}=w^{1-\gamma_s}\left(\gamma_s^{-1}\sum_{q\in R}a_{qs}p_{q,t}^{1-\nu}\right)^{\frac{\gamma_s}{1-\nu}}.
\end{align*}
Taking logs and noting that  $\log w=0$ given the numeraire condition 
\begin{align*}
    \log(p_{s,t})= \frac{\gamma_s}{1-\nu}\log\left(\gamma_s^{-1}\sum_{q\in R}a_{qs}p_{q,t}^{1-\nu}\right).
\end{align*}
This condition implies a system of equations solved by $p_s=1,\,\forall s$. To see this, conjecture the solution and note that $\sum_q{a_{qs}=\gamma_s}$ by constant returns to scale.\footnote{An alternative way to obtain this result is to write the system in relative prices $\log(p_s/w)$, this is solved by $p_s=w,\,\forall s$ and then imposing the numeraire condition $w=1$.} 

$I$ firms have the same expenditure minimization problem, meaning that their optimal expenditure shares will be identical to the ones of $C$ firms. However, they have a different pricing and inventory problem. 

I assume that $I$ firms compete \`a la Bertrand and therefore price at the marginal cost of the $C$ firms. As a consequence, the pricing problem of $I$ firms is solved by a vector $p^I_s=1,\,\forall s$. Note that this is an equilibrium since $Z_s^I>Z_s$ so that $I$ firms charge a markup $\mu_s^I>1$ and obtain positive profits. The quantity problem is then given by solving the dynamic problem 
\begin{align*}
    \max_{Y_{s,t},I_{s,t}, Q_{s,t}}\;&\mathbb{E}_t\sum_{t} \beta^t\left[p_{s,t} Q_{s,t} - c_{s,t}Y_{s,t} -\frac{\delta}{2}(I_{s,t}-\alpha Q_{s,t+1})^2 \right] \quad st\quad I_{s,t}=I_{s,t-1}+Y_{s,t}-Q_{s,t},
\end{align*}
noting that $p_s=1$ and therefore $Q_{s,t}$ is equal to the demanded quantity at $p_s=1$, which is given the expenditure on good $s$ is given by consumer expenditure $D_{s,t}=\beta_s D_t$ and the expenditure from other firms. Denote total demand on good $s$ by final consumers and other firms $\mathcal{D}_{s,t}$. Note that, as the vector of prices is a constant and, given the absence of productivity shocks, the marginal cost of all firms is also constant. To fix ideas, suppose $I$ firms in sector $s$ have productivity $Z_s^I=Z_s/\zeta_s $, with $\zeta_s<1$, then their marginal cost given by $\zeta_s$. Then the problem becomes
\begin{align*}
    \max_{Y_{s,t},I_{s,t}}\;&\mathbb{E}_t\sum_{t} \beta^t\left[\mathcal{D}_{s,t} - \zeta_s Y_{s,t} -\frac{\delta}{2}(I_{s,t}-\alpha \mathcal{D}_{s,t+1})^2 \right] \quad st\quad I_{s,t}=I_{s,t-1}+Y_{s,t}-\mathcal{D}_{s,t},
\end{align*}
This implies the optimal inventory rule $I_{s,t}=\frac{(\beta-1)\zeta_s }{\delta}+\alpha\mathbb{E}_t \mathcal{D}_{s,t+1}$. I disregard the case in which $I_{s,t}<0$ at the optimum since I can always choose $\bar X$ large enough so that $\alpha\mathbb{E}_t \mathcal{D}_{s,t+1}>\frac{(\beta-1)\zeta_s }{\delta}$ and therefore $I_{s,t}>0$.

At this point, the remaining problem is the definition of $\mathcal{D}_{s,t}$. Towards a resolution, note that in this production network model firms sell part of their output directly to consumers. Denote this output $Y_{s,t}^0$ as it is at $0$ distance from consumption. The linearity of the inventory policy implies that I can characterize the stage-specific problem and aggregate ex-post. For the part of output sold directly to final consumers, $\mathcal{D}_{s,t}^0=D_{s,t}=\beta_s D_t$. Production at this stage is defined as sales plus the change in inventories. As discussed above, sales are given by the total demand for good $s$ at $p_s=1$, which is equal to $\beta_s D_t$. The change in inventories $\Delta I_{s,t}^0=I_{s,t}^0-I_{s,t-1}^0= \alpha(\mathbb{E}_t \mathcal{D}_{s,t+1}^0-\mathbb{E}_t\mathcal{D}_{s,t}^0)$. Using $\mathcal {D}_{s,t}^0=\beta_s D_t$ and the properties of the stochastic process of $D_t$ discussed above, $\Delta I_{s,t}^0=\alpha\beta_s\rho \Delta_t$, where $\Delta_t=D_t-D_{t-1}$. Hence output of $I$ firms in sector $s$ at distance $0$ is $Y_{s,t}^0=\beta_s[D_t+\alpha\rho \Delta_t]$. 

Next note that to produce output $Y_{s,t}^0$ firms demand from input suppliers in sector $r$ an amount $\mathcal{D}^1_{rs,t}=a_{rs}\gamma_s Y_{s,t}^0$. I denote these with $1$ since for this part of their production they operate at distance 1 from consumers. Given this demand, the implied production is $Y_{rs,t}^1=a_{rs}\gamma_s Y_{s,t}^0+\Delta I_{rs,t}^1$. Summing over all possible final producers $s$, implies a total production of $I$ firms in sector $r$ at distance 1 equal to $Y_{r,t}^1=\sum_s a_{rs}\gamma_s Y_{s,t}^0+\Delta I_{rs,t}^1$. I can solve the rest of the model by forward induction and finally aggregating at the firm level over all possible distances from consumers so that output of $I$ firms in sector $k$ is $Y_{k,t}=\sum_{n=0}^\infty Y_{k,t}^n$. Proposition \ref{indoutput} provides the exact closed-form solution for production in the network.

\begin{proof}[Proof of Proposition \ref{indoutput}]
The first part of the Proposition follows immediately from the definition of output at a specific stage $n$ and total sectoral output as the sum over stage-specific output. 
The proof of the second part requires the following steps: first, using the definition of $\left[\tilde{\mathcal{A}}^n\right]_kB$ and denoting $\omega=1+\alpha(\rho-1)$, rewrite total output as
\begin{align*}
Y_{k,t}&=\sum_{n=0}^\infty  \left[\tilde{\mathcal{A}}^n\right]_kB\left[D_t+\alpha\rho\sum_{i=0}^n\omega^i \Delta_t\right]=\left[\tilde{\mathcal{A}}^0 +\tilde{\mathcal{A}}^1+\ldots\right]_kBD_t+\alpha\rho \left[\tilde{\mathcal{A}}^0\omega^0+\tilde{\mathcal{A}}^1(\omega^0+\omega^1)+\ldots\right]_kB\Delta_t\\
     &= \tilde L_k B D_t+ \alpha\rho\left[\sum_{n=0}^\infty\tilde{\mathcal{A}}^n\sum_{i=0}^n\omega^i\right]_kB \Delta_t.
\end{align*}
 The equality between the two lines follows from the convergence of a Neumann series of matrices satisfying the Brauer-Solow condition.
 To show that $Y_{k,t}$ exists non-negative for $\omega-1=\alpha(\rho-1)\in[-1,0]$, I characterize the problem at the bounds $\omega-1=-1$ and $\omega-1=0$ and exploit monotonicity inbetwee. Note that if $\omega-1=-1$ then $\omega=0$, the second term collapses to $\tilde L$, and existence and non-negativity follow from $\tilde L$ finite and non-negative. If $\omega-1=0$, then $\omega=1$ and
 \begin{align*}
Y_{k,t}&= \tilde L_k B D_t+ \alpha\rho\left[\sum_{n=0}^\infty(n+1)\tilde{\mathcal{A}}^n\right]_kB \Delta_t= \tilde L_k B D_t+ \alpha\rho\left[\tilde{\mathcal{A}}^0+2\tilde{\mathcal{A}}^1+3\tilde{\mathcal{A}}^2+\ldots\right]_kB \Delta_t\\
 &=\tilde L_k B D_t+ \alpha\rho\tilde L^2_kB \Delta_t,
\end{align*}
where the last equality follows from $\sum^\infty_{i=0} (i+1)\mathcal{A}^i=[I-\mathcal{A}]^{-2}$ if $\mathcal{A}$ satisfies the Brauer-Solow condition. Existence and non-negativity follow from the existence and non-negativity of $[I-\tilde{\mathcal{A}}]^{-2}$. 

If $\omega-1\in(-1,0)$, then $\omega\in(0,1)$. As this term is powered up in the second summation and as it is strictly smaller than $1$, it is bounded above by $n+1$. This implies that the whole second term 
$
    \sum_{n=0}^\infty\tilde{\mathcal{A}}^n\sum_{i=0}^n\omega^i<\sum_{n=0}^\infty (n+1)\tilde{\mathcal{A}}^n=\tilde L^2<\infty.
$ Alternatively, note that the second summation is strictly increasing in $\omega$, as $\omega\leq 1$ the summation is bounded above by $n+1$.
Which completes the proof.
\end{proof}

\begin{proof}[Proof of Proposition \ref{propgrowthrates}]
    From Proposition \ref{indoutput} we have that, differentiating sectoral output with respect to final demand 
    \begin{align*}
        \frac{\partial Y_{k,t}}{\partial D_t}=\tilde L_k B + \alpha\rho\left[\sum_{n=0}^\infty\tilde{\mathcal{A}}^n\sum_{i=0}^n\omega^i\right]_kB
    \end{align*}
    Dividing by steady-state sales $\tilde L_k B \bar D$ and using $\delta D_t=\Delta_t$ and the definition of $\mathcal{U}_k$, the first order change is given by
\begin{align*}
        \frac{\Delta Y_{k,t}}{\tilde L_k B \bar D}
        &\approx\tilde L_k B\frac{\Delta_t}{\tilde L_k B \bar D}+\alpha\rho \left[\sum_{n=0}^\infty(n+1)\tilde{\mathcal{A}}^n \right]_kB \frac{\Delta_t}{\tilde L_k B \bar D}=\frac{\Delta_t}{\bar D}+\alpha\rho\ \mathcal{U}_k \frac{\Delta}{\bar D}.
    \end{align*}    
    Finally, noting that around the non-stochastic steady state $Y_{k}=\tilde L_k B \bar D$, the following holds
    \begin{align*}
        \frac{\Delta Y_{k,t}}{Y_k}
        &=\frac{\Delta_t}{\bar D}+\alpha\rho \mathcal{U}_k \frac{\Delta_t}{\bar D}.
    \end{align*} 
\end{proof}

\begin{proof}[Proof of Proposition \ref{corvolatility}]
Immediate from Proposition \ref{propgrowthrates} since all the terms multiplying $\Delta_t/\bar D$ are constant. The proof of the second part of the Proposition follows from the fact that there is a single demand destination, therefore $\sigma_{\Delta \log D_t}$ is common across industries. Hence the volatility of output growth can only be larger for more upstream industries   
\end{proof}
\begin{proof}[Proof of Proposition \ref{invcs}]
    Immediate from Proposition \ref{propgrowthrates} and the observation that output growth is increasing in $\omega$, which itself is increasing in $\alpha$ and $\rho$, under the maintained assumption that $\omega\in(0,1)$.

    For part c), from Proposition \ref{propgrowthrates}, note that if $\rho\rightarrow 1$, then $\omega=1+\alpha(\rho-1)\rightarrow 1$. Hence, $\mathcal{U}_k=\sum_{n=0}^\infty\frac{1-\omega^{n+1}}{1-\omega}\tilde{\mathcal{A}}^n B_k \bar{D}/Y_k = \sum_{n=0}^\infty(n+1)\tilde{\mathcal{A}}^n/Y_k$. The Nuemann series converges to $[I-\mathcal{A}]^{-2}$ and therefore $\mathcal{U}_k\rightarrow U_k$. 
\end{proof}

\begin{proof}[Proof of Proposition \ref{propdemand}]
    I show that if two industries $r$ and $s$ are downstream symmetric and $\beta_r<\beta_s$ then $\mathcal{U}_r>\mathcal{U}_s$ and, therefore, by Proposition \ref{corvolatility}, $\Var(\Delta\log Y_{rt})>\Var(\Delta\log Y_{st})$.

    The fraction of output sold directly to households is $\beta_r D/Y_r$. This share is increasing in $\beta_r$ since 
    \begin{align}
        \diffp{{\beta_r D/Y_r}}{{\beta_r}}\propto Y_r-\ell_{rr}\beta_r\geq 0.
    \end{align}
    If two industries are downstream symmetric, all the higher-order terms are identical since $a_{ri}=a_{si},\,\forall i$. Since the fraction of output sold to consumers is higher for $s$ than for $r$ the complement fraction sold indirectly is smaller for $s$ than for $r$. As a consequence, $\mathcal{U}_s<\mathcal{U}_r$ and, therefore,  $\Var(\Delta\log Y_{rt})>\Var(\Delta\log Y_{st})$.
\end{proof}

\begin{proof}[Proof of Proposition \ref{propfragmentation}]
    First, note that the case studied in the proposition is one of pure vertical fragmentation: a sector $s$ splits into two sectors $i$ and $j$. This split implies that i) $i$ remains upstream symmetric to itself; ii) $k$ remains downstream symmetric to itself. The consequence of this split is that the network upstream of the fragmented sector $i$ is unchanged the network downstream of the fragmented sector $j$ is unchanged. Consequently, any network path going through $i$ now remains unchanged but has to go through $j$ as well, adding a step in the supply chain. Therefore, any sector $r$ connected to consumers through $i$, ($\ell_{ri}>0$), becomes more upstream and therefore volatile, all else equal. 
\end{proof}

\begin{proof}[Proof of Proposition \ref{growthmulti} ]    Denote by $B_j$ the $R\times1$ vector with a typical element $\beta_j^r$ and by $B$ the $R\times J$ matrix with $B_j$ in the $j^{th}$ column. Also, note that now $\bar D$ is a $J\times 1$ vector with typical element $\bar D_j$. As before $\tilde L_r$ is the $r^{th}$ row of $\tilde L$, which is an $R\times R$ matrix. Sectoral output of a generic sector $r$ is given by 
    \begin{align*}
        Y_{t}^r=\sum_j\tilde{L}_r B_{j}D_{jt}+\alpha\rho\left[\sum_{n=0}^\infty\tilde{\mathcal{A}}^n\frac{1-\omega^{n+1}}{1-\omega}\right]_r B_{j}\Delta_{jt}.
    \end{align*}
To a first order, the change in output is given by
\begin{align*}
    \Delta Y_{t}^r\approx\sum_j\tilde{L}_r B_{j}\Delta_{jt}+\alpha\rho\sum_j\left[\sum_{n=0}^\infty\tilde{\mathcal{A}}^n\frac{1-\omega^{n+1}}{1-\omega}\right]_r B_{j}\Delta_{jt}.
\end{align*}
Dividing by steady-state output 
\begin{align*}
    \frac{\Delta Y_{t}^r}{\tilde{L}_r B \bar D}\approx\frac{\sum_j\tilde{L}_r B_{j}\Delta_{jt}}{\tilde{L}_r B \bar D}+\alpha\rho\sum_j\left[\sum_{n=0}^\infty\tilde{\mathcal{A}}^n\frac{1-\omega^{n+1}}{1-\omega}\right]_r \frac{B_{j}\Delta_{jt}}{\tilde{L}_r B \bar D}.
\end{align*}
Dividing and multiplying by $\bar D_j$ inside the summations
\begin{align*}
    \frac{\Delta Y_{t}^r}{\tilde{L}_r B \bar D}\approx\frac{\sum_j\tilde{L}_r B_{j}\bar D_j}{\tilde{L}_r B \bar D}\frac{\Delta_{jt}}{\bar D_j}+\alpha\rho\sum_j\left[\sum_{n=0}^\infty\tilde{\mathcal{A}}^n\frac{1-\omega^{n+1}}{1-\omega}\right]_r \frac{B_{j}\bar D_j}{\tilde{L}_r B \bar D}\frac{\Delta_{jt}}{\bar D_j}.
\end{align*}
Using the definition of $\mathcal{U}_j^r$
\begin{align*}
    \frac{\Delta Y_{t}^r}{\tilde{L}_r B \bar D}&\approx\sum_j\underbrace{\frac{\tilde{L}_r B_{j}\bar D_j}{\tilde{L}_r B \bar D}}_{\xi_{j}^r}\underbrace{\frac{\Delta_{jt}}{\bar D_j}}_{\eta{jt}}+\alpha\rho \sum_j\frac{\sum_{n=0}^\infty \tilde{\mathcal{A}}^n\frac{1-\omega^{n+1}}{1-\omega}B_j\bar D_j}{\tilde{L}_r B \bar D}\frac{\Delta_{jt}}{\bar D_j}\\
    &=\sum_j\xi_{j}^r\eta_{jt}+\alpha\rho \sum_j\frac{\sum_{n=0}^\infty \tilde{\mathcal{A}}^n\frac{1-\omega^{n+1}}{1-\omega}B_j\bar D_j}{\tilde{L}_r B \bar D}\frac{\Delta_{jt}}{\bar D_j}
    \end{align*}
    Dividing and multiplying the second term inside the summation by $\frac{\tilde{L}_r B_{j} \bar D_j}{\tilde{L}_r B_{j} \bar D_j}$
    \begin{align*}
    \frac{\Delta Y_{t}^r}{\tilde{L}_r B \bar D}&\approx    \sum_j\xi_{j}^r\eta_{jt}+\alpha\rho \sum_j\underbrace{\sum_{n=0}^\infty \frac{\tilde{\mathcal{A}}^n\frac{1-\omega^{n+1}}{1-\omega} B_{j}\bar D_j}{\tilde{L}_r B_{j} \bar D_j}}_{\mathcal{U}^r_j}\underbrace{\frac{\tilde{L}_r B_{j} \bar D_j}{\tilde{L}_r B \bar D}}_{\xi^r_j}\underbrace{\frac{\Delta_{jt}}{\bar D_j}}_{\eta_{jt}}\\
    &=\sum_j\xi_{j}^r\eta_{jt}+\alpha\rho \sum_j\mathcal{U}^r_{j} \xi^r_{j}\eta_{jt}=\eta_{t}^r+\alpha\rho \sum_j\mathcal{U}^r_{j} \xi^r_{j}\eta_{jt}.
    \end{align*}
        Finally, noting that, in steady state, output and sales coincide
\begin{align*}
    \frac{\Delta Y^r_{t}}{Y^r}=\eta_{t}^r+\alpha\rho \sum_j\mathcal{U}^r_{j} \xi^r_{j}\eta_{jt}.
\end{align*}
\end{proof}
\begin{proof}[Proof of Proposition \ref{volatilitymulti}]
    Starting from Proposition \ref{growthmulti}
    \begin{align*}
    \frac{\Delta Y_{t}^r}{Y^r}=\eta_{it}^r+\alpha\rho \sum_j\mathcal{U}_{j}^r \xi^r_{j}\eta_{jt}
    \end{align*}
    To simplify notation, denote $\upsilon_{it}^r\equiv\sum_j\mathcal{U}_{j}^r \xi^r_{j}\eta_{jt}$. 
    Then, the variance of output growth is given by
    \begin{align*}
        \Var(\Delta\log Y_{t}^r)=\Var(\eta_{t}^r)+(\alpha\rho)^2\Var(\upsilon_{t}^r)+2\alpha\rho \Cov(\eta_{t}^r,\upsilon_{t}^r).
    \end{align*}
    Next, I study each term in turn. Using the definition of $\eta_{t}^r$ and iid-ness of $\eta_{jt}$ across $j$, the variance of $\eta^r_{t}$ is 
    \begin{align*}
        \Var(\eta^r_{t})&=\sum_t \left(\eta^r_{t}-\frac{1}{T}\sum_t \eta^r_{t}\right)^2=\sum_t \left(\sum_j\xi^r_{j}\eta_{jt}-\frac{1}{T}\sum_t\sum_j \xi^r_{j}\eta_{jt}\right)^2\\
        &=\sum_t\left( \sum_j \xi_{j}^r(\eta_{jt}-\bar{\eta}_j)\right)^2 \overset{iid}{=}\sum_j {\xi_{j}^r}^2 \Var(\eta_{jt}).
    \end{align*}
    Next, note that the variance of $\upsilon_{it}^r$, by the same arguments, is
    \begin{align*}
        \Var(\upsilon_{t}^r)=\sum_j {U_{j}^r}^2{\xi_{j}^r}^2\Var(\eta_{jt}).
    \end{align*}
    Finally, the covariance term is given by
    \begin{align*}\nonumber \Cov(\eta_{t}^r,\upsilon_{t}^r)&=\Cov\left(\sum_j\xi_{j}^r \eta_{jt}, \sum_j \mathcal{U}_{j}^r\xi_{j}^r\eta_{jt}\right)\\\nonumber
    &=\sum_j\sum_k \xi_{j}^r \mathcal{U}_{k}^r\xi_{k}^r \Cov(\eta_{jt},\eta_{kt})
    \overset{iid}{=} \sum_j \mathcal{U}_{j}^r{\xi_{j}^r}^2\Var(\eta_{jt}).
    \end{align*}
    Where the last equality follows by iid-ness. 
    Combining the three terms
    \begin{align*}\nonumber
        \Var(\Delta\log Y_{t}^r)&=\sum_j {U_{j}^r}^2{\xi_{j}^r}^2\Var(\eta_{jt})+(\alpha\rho)^2\sum_j {\mathcal{U}_{j}^r}^2{\xi_{j}^r}^2\Var(\eta_{jt})+2\alpha\rho \sum_j \mathcal{U}_{j}^r{\xi_{j}^r}^2\Var(\eta_{jt})=\\\nonumber
        &=\sum_j {\xi_{j}^r}^2\Var(\eta_{jt})\left[1+ (\alpha\rho)^2{\mathcal{U}_{j}^r}^2 +2\alpha\rho \mathcal{U}_{j}^r \right]=\sum_j {\xi_{j}^r}^2\Var(\eta_{jt})\left[ 1+ \alpha\rho{\mathcal{U}_{j}^r} \right]^2.
    \end{align*}
\end{proof}

\begin{proof}[Proof of Proposition \ref{corvolatilitymulti}]
        The difference in variances is
\begin{align*}\nonumber
    \Var(\Delta\log Y_{t}^r)-\Var(\Delta\log Y_{t}^s)&=\sigma_\eta\sum_j\xi_{j}^r(\mathcal{U}_{j}^r-\mathcal{U}_{j}^s)
    \end{align*}
    which is positive if $\exists j$ such that $\mathcal{U}_{j}^r>\mathcal{U}_{j}^s$ and $\nexists j$ such that $\mathcal{U}_{j}^r<\mathcal{U}_{j}^s$.
\end{proof}

\begin{proof}[Proof of Proposition \ref{corvolatilitymulti2}]
    First, note that in the special case $\Var(\eta_{jt})=\sigma_\eta,\,\forall j$, the variance of output is given by
    \begin{align*}
        \Var(\Delta\log Y_{t}^r)=\sigma_\eta\sum_j {\xi_{j}^r}^2{(1+\alpha\rho \mathcal{U}_{j}^r)}^2,
    \end{align*}
By assumption $\mathcal{U}^r=\mathcal{U}^s$. First note that
\begin{align*}
    \sum_j \xi^r_j(1+\alpha\rho \mathcal{U}^r_j)=1+\alpha\rho\sum_j \xi^r_j U_j^r= 1+\alpha\rho \mathcal{U}^r,
    \end{align*}
    Where the second equality follows from $\sum_j\xi_j^r=1$ and the third from the observation that $\mathcal{U}^r=\sum_j\xi^r_j \mathcal{U}^r_j$. Hence the assumption that $\mathcal{U}^r=\mathcal{U}^s$ implies that $\sum_j \xi^r_j(1+\alpha\rho\mathcal{U}^r_j)=\sum_j \xi^s_j(1+\alpha\rho\mathcal{U}^s_j)$.
    Next, note that the difference in output growth volatility between the two industries is given by
\begin{align*}\nonumber
    \Var(\Delta\log Y_{t}^r)-\Var(\Delta\log Y_{t}^s)&=\sigma_\eta\left[\sum_j{(1+\alpha\rho\mathcal{U}^r_j)}^2 {\xi_{j}^r}^2- \sum_j {(1+\alpha\rho\mathcal{U}^s_j)}^2{\xi_{j}^s}^2\right]
\end{align*}
The statement in part $a)$ follows immediately by the observation that, since the two distributions have the same mean, the expression is positive if and only if $\Var((1+\alpha\rho\underline{\mathcal{U}}^r)\Xi^r)>\Var((1+\alpha\rho\underline{\mathcal{U}}^s)\Xi^s)$. Since the two distributions have the same means, this is equivalent to $\Psi^s\Xi^s$ second-order stochastic dominating $\Psi^r\Xi^r$.

The proof of point $b)$ follows from the observation that, in the special case $\mathcal{U}^r_j=\mathcal{U}^s_j=\mathcal{U}\,\forall j$
\begin{align*}\nonumber
    \Var(\Delta\log Y_{t}^r)-\Var(\Delta\log Y_{t}^s)&=\sigma_\eta{(1+\alpha\rho\mathcal{U})}^2 \left[\sum_j{\xi_{j}^r}^2- \sum_j {\xi_{j}^s}^2\right].
\end{align*}
Further, since $\sum_j \xi^r_j=\sum_j \xi^s_j=1$, the expression is positive if and only if $\Xi^r$ has a higher variance than $\Xi^s$. Which, in the case of equal means, is equivalent to $\Xi^s$ second-order stochastic dominating $\Xi^r$.
\end{proof}

\newpage

\section{Model Extensions and Additional Theoretical Results}

\subsection*{A Dynamic Model of Optimal Procyclical Inventories}\label{quant_model}\addcontentsline{toc}{subsection}{A Dynamic Model of Optimal Procyclical Inventories}

In this section, I show that optimally procyclical inventory policy obtains as the policy for a firm subject to production breakdowns. Consider a price-taking firm facing some stochastic demand $q(A)$ where $A$ follows some cdf $\Phi$. The firm produces at marginal cost $c$ and, with probability $\chi>0$, is unable to produce in a given period. The problem of the firm is described by the value functions for the ``good'' state where it can produce, and the ``bad'' state where production is halted. The firm can store inventories $I$ between periods. Inventories follow the law of motion $I^\prime=I+y-q(A)$, where $y$ is output and $q(A)$ is, by market clearing, total sales. Suppose further that firms do not face consecutive periods of halted production.
\begin{align*}
V^G(I,A)&=\max_{I^\prime} p q(A)-c y +\beta\mathbb{E}_{A^\prime|A}\left[\chi V^B(I^\prime,A^\prime)+(1-\chi) V^G(I^\prime,A^\prime)\right],\\
V^B(I,A)&=p \min\{q(A),I\}+\beta \mathbb{E}_{A^\prime|A}V^G(I^\prime,A^\prime).
\end{align*}
The first order condition for next period inventories is then given by
\begin{align*}
\frac{1-\beta(1-\chi)}{\beta\chi}c=\frac{\partial \mathbb{E}_{A^\prime|A} V^B(I^\prime,A^\prime)}{\partial I^\prime}.
\end{align*}
Note that the LHS is a positive constant and represents the marginal cost of producing more today relative to tomorrow. This is given by time discounting of the marginal cost payments, which the firm would prefer to backload. Note trivially that if the probability of halted production goes to zero, the firm has no reason to hold inventories. 
The marginal benefit of holding inventories is given by relaxing the sales constraint in the bad state. Denote $P(A,I^\prime)$ the probability that the realization of $A^\prime$ implies a level of demand larger than the firm's inventories, which implies that the firm stocks out. This probability depends on the current state since demand realizations are not independent. Denote $P_{I^\prime}(A,I^\prime)=\partial P(A,I^\prime) /\partial I^\prime$. Then, the following holds
\begin{align*}
	%\frac{\partial \mathbb{E}_{A^\prime|A} V^B(I^\prime,A^\prime)}{\partial I^\prime}&=p(P(A,I^\prime)+P_{I^\prime}(A,I^\prime)(I^\prime-\mathbb{E}[A^\prime|q(A^\prime)>I^\prime])+\\
	%&\beta P_{I^\prime}(A,I^\prime) \left[ \mathbb{E}_{A^\prime|A} V^G(0,A^\prime)-V^G(I^\prime-q(A^\prime),A^\prime)   	\right]+\\
	%&\beta (1-P(A,I^\prime))\frac{\partial\mathbb{E}_{A^{\prime\prime}|A} V^G(I^{\prime\prime}-q(A^{\prime\prime}),A^\prime)}{\partial I^\prime}.
    &\frac{\partial \mathbb{E}_{A^\prime|A} V^B(I^\prime,A^\prime)}{\partial I^\prime}=\\
	&pP(A,I^\prime) +\beta P_{I^\prime}(A,I^\prime)\mathbb{E}_{A^{\prime\prime}|A} \left[  V^G(0,A^\prime)-V^G(I^\prime-q(A^\prime),A^\prime)   	\right]+\beta (1-P(A,I^\prime))c>0.
\end{align*}
This states that extra inventories in the bad state imply marginal revenues equal to the price in the event of a stockout. The last two terms state that it makes it less likely that the firm will have to start the next period without inventories and that it will be able to save on marginal cost for production if it does not stock out.

Note that it is immediate that the value of both problems is increasing in the level of inventories the firm starts the period with. It is also straightforward to see that if $\partial \mathbb{E}_{A^\prime|A}/\partial A>0$, namely if shocks are positively autocorrelated, then the expected value in the bad state is non-decreasing in $A$. 
\begin{appprop}[Procyclical Inventories]
\label{propendinventories} Consider two realizations of the demand shifter $A_1>A_2$, then, at the optimum, $I^{\prime\star}(I,A_1)>I^{\prime\star}(I,A_2)$.
\end{appprop}
The optimality condition for inventories shows that the LHS is constant while the RHS increases in inventory holdings and decreases in the level of demand. Evaluating the first order condition at different levels of $A$, it has to be that $I^{\prime\star}(I,A_1)>I^{\prime\star}(I,A_2)$, $\forall A_1>A_2$. In other words, the firm will respond to a positive demand shock by increasing output more than 1-to-1 as it updates inventories procyclically. The reason is that a positive shock today increases the conditional expectation on demand tomorrow. As a consequence, the likelihood of a stock-out for a given level of inventories increases, which implies that the RHS of the first order condition increases as the benefit of an additional unit of inventories rises.

\subsection*{Production Smoothing Motive}\label{smoothingmotive}\addcontentsline{toc}{subsection}{Production Smoothing Motive}
In the main body of the paper, I assume the inventory problem is defined by a quadratic loss function $(I_t-\alpha D_{t+1})^2$. This assumption is a stand-in for the costs of holding inventories or stocking out. However, it imposes two possibly unrealistic restrictions: i) it implies a symmetry between the cost of holding excess inventories and the cost of stocking out; ii) it excludes any production smoothing motive as it implies an optimal constant target rule on expected future sales. In this section, I extend the problem to eliminate these restrictions following \cite{ramey1999inventories} more closely. 
Formally, consider the problem of a firm solving 
\begin{align*}
    \max_{I_t,Y_t}&\mathbb{E}_t\sum_{t} \beta^t\left[D_t-Y_t\left(c +\frac{\theta}{2}Y_t\right) -\frac{\delta}{2}(I_t-\alpha D_{t+1})^2 -\tau I_t\right] \quad st\\\nonumber
    & I_{t}=I_{t-1}+Y_t-D_t,
    \end{align*}
Where the term $Y_t\left(c +\frac{\theta}{2}Y_t\right)$ includes a convex cost of production, which in turn generates a motive to smooth production across periods. The term $\tau I_t$ implies a cost of holding inventories which breaks the symmetry between holding excessive or too little inventories. In what follows, I drop the stage and time indices and denote future periods by $\prime$. The first order condition with respect to end-of-the-period inventories implies
\begin{align*}
    I=(\theta(1+\beta)+\delta)^{-1}\left[ c(\beta-1)-\tau+\delta\alpha \mathbb{E}D^\prime -\theta(D-I_{-1})+\theta\beta\mathbb{E}(D^\prime+I^\prime)\right],
\end{align*}
Define $\mathcal{B}\coloneqq (\theta(1+\beta)+\delta)^{-1}$, taking a derivative with respect to current demand implies
\begin{align*}
    \frac{\partial I}{\partial D}=\mathcal{B}\left(\delta\alpha\rho -\theta(1-\beta\rho)+\theta\beta\frac{\partial \mathbb{E}I^\prime}{\partial D}\right).
\end{align*}
Define $\mathcal{X}\coloneqq\delta\alpha\rho +\theta(1-\beta\rho)$ then
\begin{align*}
    \frac{\partial \mathbb{E}I^\prime}{\partial D}=\mathcal{B}\left(\rho\mathcal{X} +\theta\frac{\partial I}{\partial D}+\theta\beta \frac{\partial\mathbb{E}I^{\prime\prime}}{\partial D}\right).
\end{align*}
Iterating forward and substituting the following obtains.
\begin{appprop}[Cyclicality of Inventories]\label{cyclicalityprop}
A change in demand implies the following change in the optimal inventory level:
 \begin{align*}
    \frac{\partial I}{\partial D}=\left(1-\sum_{i=1}^\infty \left(\mathcal{B}^2\theta^2\beta\right)^i\right)^{-1}\mathcal{B}\mathcal{X}\sum_{j=0}^\infty \left(\mathcal{B}\theta\rho\beta\right)^j.
\end{align*}
If both $\mathcal{B}^2\theta^2\beta$ and $\mathcal{B}\theta\rho\beta$ are in the unit circle then 
\begin{align*}
    \frac{\partial I}{\partial D}=\frac{\mathcal{B}\mathcal{X}}{1-\mathcal{B}\theta\beta\rho}\frac{1-2\mathcal{B}^2\theta^2\beta}{1-\mathcal{B}^2\theta^2\beta}\lessgtr0.
\end{align*}   
\end{appprop}
This states intuitively that if the production smoothing motive is strong enough, then inventories respond countercyclically to changes in demand. This is immediate upon noting that when $\theta=0$ then $\mathcal{B}=\delta^{-1}, \,\mathcal{X}=\delta\alpha\rho$ and therefore $\frac{\partial I}{\partial D}=\alpha\rho>0$, while if $\theta>\mathcal{B}(\beta/2)^{1/2}$ then $\frac{\partial I}{\partial D}<0$. Had the latter effect dominated, then the empirical estimates of the response of inventories to changes in sales would be negative, which is counterfactual given the findings discussed in Appendix \ref{alpha_desc}.

\subsection*{Directed Acyclic Graphs Economies}\label{dagsection}\addcontentsline{toc}{subsection}{Directed Acyclic Graphs Economies}

The model derived in the section \ref{generalnetwork} applies to economies with general networks defined by the input requirement matrix $A$, a vector of input shares $\Gamma$ and a vector of demand weights $B$. As discussed in the main body this economy features finite output under some regularity condition on the intensity of the inventory channel. I now restrict the set of possible networks to Directed Acyclic Graphs (DAGs) by making specific assumptions on $A$ and $\Gamma$. 
 This subset of networks features no cycle between nodes. 
  \begin{definition}[Directed Acyclic Graph]
A Directed Acyclic Graph is a directed graph such that $[A^n]_{rr}=0, \forall r,n$.
\end{definition}

The next trivial lemma provides a bound for the maximal length of a path in such a graph.
 \begin{lemma}[Longest Path in Directed Acyclic Graph]\label{pathdag}
  In an economy with a finite number of sectors $R$, whose production network is a Directed Acyclic Graph, there exists an $N\leq R$ such that $^n\tilde A^n=[0]_{R\times R},\,\forall n\geq N\,\wedge \,  ^n\tilde A^n\neq[0]_{R\times R}, \, \forall n<N$. Such $N$ is the longest path in the network and is finite.
 \end{lemma}
 \begin{proof}
 A path in a graph is a product of the form $\tilde a^{rs}\hdots\tilde a^{uv}>0$. A cycle in such a graph is a path of the form $\tilde a^{rs}\hdots\tilde a^{ur}$ (starts and ends in $r$). The assumption that there are no cycles in this graph implies that all sequences of the form $\tilde a^{rs}\hdots\tilde a^{ur}=0$ for any length of such sequence. Suppose that there is a finite number of industries $R$ such that the matrix $A$ is $R\times R$. Take a path of length $R+1$ of the form  $\tilde a^{rs}\hdots\tilde a^{uv}>0$, it must be that there exists a subpath taking the form $\tilde a^{rs}\hdots\tilde a^{ur}$, which contradicts the assumption of no cycles. 
Hence, the longest path in such a graph can be at most of length $R$.  \end{proof}

With this result, it is straightforward to show that output is finite even if $\tilde{\mathcal{A}}^n\sum_{i=0}^n\omega^i$ has a spectral radius outside the unit circle.

\begin{appprop}[Output in a DAG Economy]\label{propdag}
    If the network is a DAG with $R$ sectors, then output is given by
\begin{align*}
     Y_k=\tilde L_k B D_t+ \alpha\rho\left[\sum_{n=0}^R\tilde{\mathcal{A}}^n\sum_{i=0}^n\omega^i\right]_kB \Delta_t,
 \end{align*}
Which is naturally bounded since the second term is a bounded Neumann series of matrices. 

\end{appprop}
\begin{proof}[Proof of Proposition \ref{propdag}]
    Immediate from the bounded Neumann Series. Since $R<\infty$, the bracket in the second term converges to a positive constant.
\end{proof}
\subsection*{Heterogeneous Inventory Policy in a General Network}\label{hetinv_general_app}\addcontentsline{toc}{subsection}{Heterogeneous Inventory Policy in a General Network}

Consider an extension of the model in section \ref{model} in which $I$ firms have heterogeneous losses from inventories indexed by $0\leq \alpha_s< (1-\rho)^{-1}$ for $I$ firms in sector $s$. Note that I now allow for $I$ firms in some sectors to have no inventory problem as $\alpha_s$ can be equal to 0.
Solving the same control problem as in the main model, broken down by distance from consumption, their optimal policy is given by $I^n_{s,t}=\mathcal{I}^n_s+\alpha_s\rho \mathcal{D}^n_{s,t}$.

In turn this implies that the $n$ output of $I$ firms in sector $r$ is $Y^n_{r,t}=\sum_s a_{rs}\gamma_s [Y^{n-1}_{s,t}+\mathfrak{I}^n_{rs}\rho\Delta_t]$, with $\mathfrak{I}^n_{rs}=\alpha_r(1+\mathfrak{I}_{s}^{n-1}(\rho-1))$ with boundary conditions $\mathfrak{I}_r^0=\alpha_r$ and $\mathfrak{I}^{-1}_r=0$ and $\mathfrak{I}^n_r=\sum_s \tilde{a}_{rs}\mathfrak{I}^n_{rs}$.
With these definitions, I can write the recursive definition
\begin{align*}
    Y^n_{rt}=\sum_s \tilde{a}_{rs}[Y^{n-1}_{st}+\mathfrak{I}^n_{rs}\beta_s\rho\Delta_t],
\end{align*}
which, in matrix form, implies 
\begin{align*}Y_t&=\tilde{L}BD_t+\sum_{n=0}^\infty\hat{\mathfrak{I}}^nB\rho\Delta_t,    
\end{align*}
with 
$
        \hat{\mathfrak{I}}^n=\diag\{\hat{\mathfrak{I}}^{n-1} \tilde{\mathcal{A}}(\underline{1}+\mathfrak{I}^0(\rho-1))\}.
$ 

In this model, it is considerably harder to generalize the results in Section \ref{generalnetwork} as the inventory effect does not monotonically move with the distance from consumers.  Under the assumption that $\alpha_r=\alpha,\,\forall r$, as in the main text, it is immediate to establish that if $\omega=1+\alpha(\rho-1)\in (0,1)$ then $\frac{\partial Y_t^n}{\partial Y_t^{n-1}}$ increases in $n$ in vertically integrated economies that this effect carries through in general networks. This is not true with heterogeneous inventories. As an immediate counterexample, consider a firm in sector $r$ at distance 2, such that it sells to sectors $s$ at distance 1 and sectors $k$ at distance 0 with $\alpha_s=\alpha_k=0,\,\forall s,k$. This firm has no inventory amplification downstream. Consider now a different path of a firm in sector $r$ such that at distance 1 it sells to sector $q$ with $\alpha_q>0$. It is immediate that the branch at distance 1 from $r$ to consumers through $q$ has more amplification than the distance 2 branch to consumers through $s$ and $k$. 

It is possible to characterize a special case under the definition of \textit{pure direct upstreamness}, which is a stronger version of the \textit{pure upstreamness} definition from \cite{carvalho2016supply}.
\begin{definition}[Pure Direct Upstreamness]
Sector $r$ is pure direct upstream to sector $s$ if i) $a_{rs}>0$, ii) $a_{rk}=0,\, \forall k\neq r$ and iii) $a_{sr}=0$. 
\end{definition}
These conditions ensure that firms in sector $r$ sell uniquely to sector $s$ and that there is no feedback look from sector $s$ to $r$. With this definition, I provide the necessary and sufficient condition for upstream amplification under heterogeneous inventories.
\begin{lemma}[Heterogenous Inventories in General Network]
    If sector $s$ is pure direct upstream to sector $r$, then its output is given by
    \begin{align*}
        Y_{s,t}=a_{sr}\gamma_r \left[Y_{r,t}+\sum_{n=1}^\infty\mathfrak{I}_{sr}^{n-1} \rho \Delta_t\right].
    \end{align*}
\end{lemma}
Then the following holds
\begin{appprop}[Amplification under Heterogeneous Inventories]\label{hetinvgeneral}
    If sector $s$ is pure direct upstream to sector $r$, then upstream amplification occurs if and only if
    \begin{align*}
        \rho \sum_{n=1}^\infty\mathfrak{I}_{sr}^{n-1}> \frac{1-a_{sr}\gamma_r}{a_{sr}\gamma_r}\frac{\partial Y_{r,t}}{\partial D_t}.
    \end{align*}
\end{appprop}
\begin{proof}[Proof of Proposition \ref{hetinvgeneral}]
    Differentiating with respect to $D_t$, the definition of output implies 
    \begin{align*}
        \frac{\partial Y_{st}}{\partial D_t}=a_{sr}\gamma_r\left[\frac{\partial Y_{rt}}{\partial D_t}+\rho\sum_{n=1}^\infty\mathfrak{I}_{sr}^{n-1}\right],
    \end{align*}
    Using the definition of upstream amplification as $\frac{\partial Y_{st}}{\partial D_t}>\frac{\partial Y_{rt}}{\partial D_t}$, the statement follows.
\end{proof}
The proposition characterizes the condition for upstream amplification. Note that while it looks like this condition is defined in terms of endogenous objects, it is actually recursively defined in terms of primitives and shocks through the definition of $Y_{r,t}$. The convenience of the pure direct upstreamness assumption is that it allows a simple comparative statics over $a_{sr}$ and $\gamma_r$. Suppose that $\gamma_r=a_{sr}=1$, then sector $s$ is the sole supplier of sector $r$ and the sufficient condition for upstream amplification is that there is a positive inventory effect along the chain that connects final consumers to $s$ through $r$. This is necessarily the case, provided that $0\leq\alpha_r< (1-\rho)^{-1},\,\forall r$.

\subsection*{Productivity Shocks} \label{productivityshocks}
\addcontentsline{toc}{subsection}{Productivity Shocks}

Consider an extension of the model in section \ref{model} in which $I$ firms are subject to productivity shocks so that their marginal cost follows an AR(1). Specifically for sector $s$, assume that $\zeta_{s,t}=(1-\rho^\zeta)\bar\zeta_s+\rho^\zeta \zeta_{s,t-1}+\varepsilon_t$, with average marginal cost $\bar\zeta_s$, persitence $\rho^\zeta<1$ and innovations $\varepsilon_t\sim F(\cdot)$. The control problem becomes
\begin{align*}
    \max_{Y_{s,t},I_{s,t}}\;&\mathbb{E}_t\sum_{t} \beta^t\left[\mathcal{D}_{s,t} - \zeta_{s,t} Y_{s,t} -\frac{\delta}{2}(I_{s,t}-\alpha \mathcal{D}_{s,t+1})^2 \right] \quad st\quad I_{s,t}=I_{s,t-1}+Y_{s,t}-\mathcal{D}_{s,t},
\end{align*}
The solution to this problem is given in the next Proposition.
\begin{appprop}[Inventories with Productivity Shocks]\label{prodshocksprop}
The optimal inventory policy is given by
\begin{align*}
    I_{st}=\alpha\mathbb{E}_t \mathcal{D}_{s,t+1}+\frac{\beta\mathbb{E}_t\zeta_{s,t+1}-\zeta_s}{\delta}.
\end{align*}
An increase in the firm's marginal cost then implies lower inventories since 
\begin{align*}
    \frac{\partial I_{s,t}}{\partial \zeta_{s,t}}=\beta\rho^\zeta-1<0.
\end{align*}
\end{appprop}
Hence, when the firm's productivity increases ($\zeta_{s,t} \downarrow$), so does the optimal amount of inventories. Since $Y_{s,t}=I_{s,t}-I_{s,t-1}+\mathcal{D}_{s,t}$, $\partial Y_{s,t}/\partial \zeta_{s,t}={\partial I_{s,t}}/{\partial \zeta_{s,t}}=\beta\rho^\zeta-1.$ Note that higher productivity is reflected in higher output $Y_{s,t}$ but not in higher sales $\mathcal{D}_{s,t}$, as the latter are driven by the pricing of $C$ firms. The presence of productivity shocks for $I$ firms generates a positive comovement between output and inventories.

% \begin{figure}[H]
% \caption{Effect of Demand Shocks on Sales Growth by Upstreamness Level - Time Varying Aggregation}
% \label{timevarying}\center
% \includegraphics[width=.65\textwidth]{input/margins_shocks_fe_cardinal_time_varying.png} 
% \\\vspace{5pt}
% \noindent\justifying
%  \footnotesize Note: The figure shows the marginal effect of demand shocks on industry sales changes by industry upstreamness level. The dashed horizontal line represents the average coefficient. The vertical bands illustrate the 95\% confidence intervals around the estimates. Standard errors are clustered at the country-industry level. Note that due to relatively few observations above 7, all values above 7 have been included in the $U \in [6,7]$ category. The shocks are aggregated using lagged sales shares.
% \end{figure}

\newpage

\section{Model to Data}\label{app_model_estimates}
The key difficulty in taking the model back to the data is the calculation of the appropriately weighted shocks. Here, I show that, under the assumption of common inventories, this can be readily computed. 
I start from the result of Proposition \ref{growthmulti} about the growth rate of output
    \begin{align*}
        \Delta\log Y^r_{t}\approx\eta_{t}^r+\alpha\rho\sum_j \mathcal{U}_{j}^r\xi_{j}^r\eta_{j,t}.
    \end{align*}
To compute $\mathcal U^r_j$, note that 
\begin{align*}
    \mathcal U^r_j&=\frac{ \sum_{n=0}^\infty \mathcal A^n \frac{1-\omega^{n+1}}{1-\omega}|_r B_j }{\tilde L_r B_j}\\
    &= \frac{1}{1-\omega} \frac{\sum_{n=0}^\infty \mathcal A^n |_r B_j }{\tilde L_r B_j}-\frac{\omega}{1-\omega}\frac{\sum_{n=0}^\infty \mathcal \omega^nA^n |_r B_j }{\tilde L_r B_j}\\
    &=\frac{1}{1-\omega}\frac{\tilde L_r B_j}{\tilde L_r B_j}-\frac{\omega}{1-\omega}\frac{[I-\omega\tilde A]^{-1}_rB_j}{\tilde L_r B_j}\\
    &=\frac{1}{1-\omega}-\frac{\omega}{1-\omega}\frac{[I-\omega\tilde A]^{-1}_rB_j}{\tilde L_r B_j}.
\end{align*}
Using the definition of $\omega$
\begin{align*}
    \mathcal U^r_j=\frac{1}{\alpha(\rho-1)}-\frac{1+\alpha(\rho-1)}{\alpha(\rho-1)}\frac{[I-(1+\alpha(\rho-1))\tilde A]^{-1}_rB_j}{\tilde L_r B_j}.
\end{align*}
Further, note that $\mathcal U^r_j \xi_j^r$ simplifies to
\begin{align*}
    \mathcal U^r_j \xi^r_j&=\left[\frac{1}{\alpha(\rho-1)}-\frac{1+\alpha(\rho-1)}{\alpha(\rho-1)}\frac{[I-(1+\alpha(\rho-1))\tilde A]^{-1}_rB_j}{\tilde L_r B_j}\right]\frac{\tilde L_r B_j \bar D_j}{\tilde L_r B \bar D}\\
    &=\frac{1}{\alpha(\rho-1)}\frac{\tilde L_r B_j \bar D_j}{\tilde L_r B \bar D}-\frac{1+\alpha(\rho-1)}{\alpha(\rho-1)}\frac{[I-(1+\alpha(\rho-1))\tilde A]^{-1}_rB_j\bar D_j}{\tilde L_r B \bar D}.
\end{align*}
Finally, multiplying by the destination shocks $\eta_{jt}$ and summing over $j$ I have $\hat\upsilon_t^r=\sum_j \mathcal{U}^r_j\xi^r_j \eta_{jt}$:
\begin{align*}
    \hat\upsilon_t^r&=\sum_j\frac{1}{\alpha(\rho-1)}\frac{\tilde L_r B_j \bar D_j}{\tilde L_r B \bar D}\eta_{jt}-\sum_j\frac{1+\alpha(\rho-1)}{\alpha(\rho-1)}\frac{[I-(1+\alpha(\rho-1))\tilde A]^{-1}_rB_j\bar D_j}{\tilde L_r B \bar D}\eta_{jt}\\
    &=\frac{1}{\alpha(\rho-1)}\sum_j \xi^r_j\eta_{jt}-\frac{1+\alpha(\rho-1)}{\alpha(\rho-1)}\sum_j\frac{[I-(1+\alpha(\rho-1))\tilde A]^{-1}_rB_j\bar D_j}{\tilde L_r B \bar D}\eta_{jt}\\
    &=\frac{\eta^r_t}{\alpha(\rho-1)}-\frac{1+\alpha(\rho-1)}{\alpha(\rho-1)}\sum_j\frac{[I-(1+\alpha(\rho-1))\tilde A]^{-1}_rB_j\bar D_j}{\tilde L_r B \bar D}\eta_{jt}.
\end{align*}
Under the assumption of common inventory rule, everything on the right hand side can be computed given estimates for $\alpha$ and $\rho$.

\newpage

\section{Quantitative Model}

I solve the model exactly using the closed-form solutions derived in Section \ref{model}. 
I use the full WIOD input requirement matrix as input for $\mathcal{A}_{2000}$ and the vector of final consumptions $F_{ij,2000}^r$ to back out the vector $B$ such that $\beta_{ij}^r=F^r_{ij,2000}/\sum_{r,i,j} F_{ij,2000}^r$. I set $\rho=.7$ which is similar to the estimated $AR(1)$ coefficient of the destination shocks $\hat\eta_{jt}$. The remaining parameters are the inventory-to-sales ratio $\alpha$ and the stochastic process for the shocks. To obtain an empirically plausible distribution of shocks I estimate the country-specific variances $\sigma_j^2$ and use them directly. The lack of a long enough time series prevents me from estimating the full covariance matrix. I circumvent this problem by generating a covariance matrix such that the variances are given by the vector of $\sigma_j^2$ and the average correlation is $\varrho$. I calibrate $\varrho$ to match the coefficient of the regression $\sigma(\eta_{it}^r)=\beta U_i^r+\varepsilon_i^r$, where $\sigma(\eta_{it}^r)= \left( \sum_t \left(\eta^r_{it}-\frac{1}{T}\sum_t \eta^r_{it}\right)^2    \right)^\frac{1}{2}$.
Table \ref{model_regression_varrho} reports the regression output for both the actual data and the model-generated data. 
\begin{table}[ht]
\caption{Calibration of Covariance Parameter}\label{model_regression_varrho}
\setlength{\tabcolsep}{0.35cm}  
\center\begin{threeparttable}
{
\def\sym#1{\ifmmode^{#1}\else\(^{#1}\)\fi}
\begin{tabular}{l*{2}{c}}
\hline\hline
                    &\multicolumn{1}{c}{Data}&\multicolumn{1}{c}{Model}\\\cmidrule(lr){2-2}\cmidrule(lr){3-3}
                    &\multicolumn{1}{c}{(1)}&\multicolumn{1}{c}{(2)}\\
                    &\multicolumn{1}{c}{std. dev. Demand Shocks}&\multicolumn{1}{c}{std. dev. Demand Shocks}\\
\hline
Upstreamness        &    -0.00837\sym{***}&    -0.00824\sym{***}\\
                    &  (0.000704)         &  (0.000506)         \\
[1em]
Constant            &       0.168\sym{***}&       0.155\sym{***}\\
                    &   (0.00226)         &   (0.00162)         \\
\hline
N                   &        2308         &        2308         \\
 $ R^2 $            &      0.0578         &       0.103         \\
\hline\hline
\multicolumn{3}{l}{\footnotesize Standard errors in parentheses}\\
\multicolumn{3}{l}{\footnotesize \sym{*} \(p<0.10\), \sym{**} \(p<0.05\), \sym{***} \(p<0.01\)}\\
\end{tabular}
}
\begin{tablenotes}[flushleft]
 \item \footnotesize Note: The table reports the results of the regression of demand shock volatility on average upstreamness across industries for both the actual data and the model-generated data. The model-generated data is simulated across 4800 iterations. 
\end{tablenotes}
\end{threeparttable}
\end{table}
I use $\alpha$ to match the average relative volatility of output and demand across industries of 1.26, which the model replicates exactly. The moments are targeted for the year 2001 (the first year in which shocks are identified in the data) in the multiple destination model. I keep the parameters constant for the single destination model. The counterfactuals are implemented by replacing either the input requirement matrix $\mathcal{A}_{2000}$, $\alpha$, or both.

% \begin{table}[ht]
% \caption{Targeted moments and model counterparts}\label{tab:targeted_moments}
% \setlength{\tabcolsep}{0.35cm}  
% \center\begin{threeparttable}
% \begin{tabular}{lcc} 	
% \toprule
% & Data & Model \\
% \midrule
% $\hat\beta$ & -.0084  & -.0082  \\
% \\[-1ex]
% $\sigma_y/\sigma_\eta$ &1.26 &1.26 \\
% \bottomrule
% \end{tabular} 
% \begin{tablenotes}[flushleft]
%  \item \footnotesize Note: The Table reports the targeted moments in the data and in the model. $\hat\beta$ is the estimated coefficient of the regression $\sigma(\eta_{it}^r)=\beta U_i^r+\varepsilon_i^r$, while $\sigma_y/\sigma_\eta$ is the ratio between the cross-sectional dispersion of output growth rates and demand in 2000.
% \end{tablenotes}
% \end{threeparttable}
% \end{table}

\end{appendices}
\end{document}